%% file: main.tex
\documentclass[11pt]{article}

\usepackage[letterpaper,margin=1in]{geometry}
\IfFileExists{lmodern.sty}{%
  \usepackage[T1]{fontenc}%
  \usepackage{lmodern}%
  \usepackage[protrusion=true,expansion=true]{microtype}%
}{%
  \usepackage[protrusion=true,expansion=false]{microtype}%
}

\usepackage{graphicx}
\usepackage{booktabs}
\usepackage{caption}
\usepackage{subcaption}

\usepackage{amsmath}
\usepackage{amssymb}
\usepackage{mathtools}
\usepackage{amsthm}

\usepackage{enumitem}

\usepackage[round,authoryear]{natbib}

\usepackage{tikz}
\usetikzlibrary{positioning,arrows.meta}
\usepackage{pgfplots}
\pgfplotsset{compat=1.17}
\usepgfplotslibrary{groupplots}
\usepgfplotslibrary{fillbetween}

\usepackage[colorlinks=true,allcolors=blue,breaklinks=true]{hyperref}

\usepackage{fancyhdr}
\newcommand{\V}{\mathcal{V}}
\newcommand{\E}{\mathbb{E}}
\newcommand{\Prob}{\mathbb{P}}
\newcommand{\R}{\mathbb{R}}
\newcommand{\N}{\mathbb{N}}

\newcommand{\F}{\mathcal{F}}
\newcommand{\Fhat}{\hat{\F}}
\newcommand{\Chat}{\hat{C}}
\newcommand{\phat}{\hat{p}}
\newcommand{\KL}[2]{D_{\mathrm{KL}}(#1 \,\|\, #2)}
\newcommand{\KLs}{\mathrm{KL}}

\newcommand{\TV}{\mathrm{TV}}

\theoremstyle{plain}
\newtheorem{theorem}{Theorem}[section]
\newtheorem{lemma}[theorem]{Lemma}
\newtheorem{corollary}[theorem]{Corollary}
\newtheorem{conjecture}[theorem]{Conjecture}
\newtheorem{proposition}[theorem]{Proposition}
\theoremstyle{definition}
\newtheorem{definition}[theorem]{Definition}
\newtheorem{assumption}[theorem]{Assumption}
\theoremstyle{remark}
\newtheorem{remark}[theorem]{Remark}

\begin{document}

\title{\bfseries Polynomial Context-Truncation Sensitivity in Autoregressive
Language Models:\\[2pt]
Sequential Wyner--Ziv Bounds for KV Cache Compression}
\author{Munsik Kim\thanks{Independent Researcher.\ Correspondence:
\href{mailto:physicist456@gmail.com}{physicist456@gmail.com}.}\\
Independent Researcher}
\date{}
\maketitle

\begin{abstract}
We study the rate--distortion limits of online KV cache compression
in autoregressive language models, formulating it as sequential
Wyner--Ziv source coding on the filtration induced by the model,
with the next-step query as decoder side information. Empirically,
across four models spanning two families and $0.5$--$3$B parameters,
we find that the next-token distribution's sensitivity to context
truncation decays \emph{polynomially} rather than
\emph{geometrically}: a power law improves on an exponential fit by
an order of magnitude in extrapolation, the fitted exponent is
recovered independently from a sink-plus-recent KL measurement, and
the decay is verified to be free of positional-encoding artifacts by
a position-preserving ablation. Under a corresponding
\emph{polynomial truncation-sensitivity} assumption, our main result
characterizes the per-token memory requirement of \emph{suffix-only}
cache policies: a sliding-window scheme attains distortion
$\varepsilon$ with window $w = O(\varepsilon^{-1/\alpha})$, and---under
an additional two-sided Bayes-risk condition---a converse shows
$w = \Omega(\varepsilon^{-1/\alpha})$ is necessary within this policy
class, so the scaling is $\Theta(\varepsilon^{-1/\alpha})$ for
suffix-only policies. Whether recurrent or propagating cache
summaries can beat this scaling is left open. At the level of coding
rate (rather than distortion) we give a \emph{delayed}, multi-letter
operational Wyner--Ziv comparison: under quantized query side
information and explicit information-spectrum and de-binning-stability
conditions, a suffix-only block code matches the full-context rate up to a single
side-information continuity penalty and a redundancy that vanishes
with block length. We do \emph{not} claim an online, single-letter
scheme that attains the converse exponent. Returning to the empirical
side, the same polynomial law predicts the degradation curves of
concrete cache policies: recency-based eviction (sliding,
sink-plus-recent) suppresses distortion by roughly two orders of
magnitude over random retention at equal budget, with a power-law
decay in the budget.
\end{abstract}

\section{Introduction}
\label{sec:intro}

Long-context inference with autoregressive language models is
bottlenecked by the linear growth of the key-value (KV) cache
\citep{vaswani2017attention,shazeer2019fast,ainslie2023gqa}. A rich
ecosystem of compression schemes has emerged: heavy-hitter eviction
\citep{zhang2023h2o,liu2023scissorhands}, sliding-window attention
with attention sinks \citep{xiao2024efficient,jiang2023mistral7b},
multi-head latent attention \citep{deepseek2024mla}, and low-bit
quantization \citep{liu2024kivi,kang2024gear}. Each scheme is
justified empirically; none is accompanied by a quantitative
statement of what is achievable in principle for the online
compression problem.

A recent and rapidly growing line of work brings rate--distortion
theory to bear on KV cache \emph{quantization}---the number of bits
used to encode each \emph{retained} entry. Near-optimal online vector
quantization \citep{zandieh2025turboquant} and mixed-precision
bit allocation across heads via reverse waterfilling
\citep{zuo2026ratequant} trade off a per-entry distortion against a
bit budget. This axis is orthogonal to the one we study: we ask not
how finely retained entries are encoded, but \emph{which}
entries---equivalently, how much of the past context---must be
retained at all. To our knowledge, this temporal (eviction) axis has
not previously received a sequential rate--distortion treatment.

We formulate KV cache compression as a sequential Wyner--Ziv source
coding problem on the filtration $\{\F_t\}$ induced by the language
model, with the next-step query $Q_t$ serving as decoder side
information \citep{wyner1976rate,witsenhausen1980indirect}. The
sequential setting departs from classical Wyner--Ziv in two ways:
the source is non-stationary, and the decoder receives token-level
side information.

\paragraph{The truncation-sensitivity question.}
Classical compression bounds under memory constraints rely on a
\emph{mixing} or \emph{ergodicity} condition specifying how quickly
the influence of distant past decays. The standard assumption is
$\rho$-mixing or geometric ergodicity of the underlying process:
\[
\Bigl|\Prob(X_t = x \mid X_{1:t-1}) - \Prob(X_t = x \mid X_{t-w:t-1})\Bigr|
\leq C\rho^w
\]
for some $\rho \in [0, 1)$. This form is natural for Markov chains
with spectral gap and is mathematically convenient: it yields
$\log(1/\varepsilon)$-size windows for $\varepsilon$ accuracy and
clean block-Markov coding theorems. However, the form is an
assumption about the source, and is not derived from the
language-modeling setup.

A central empirical contribution of this work is to test this
assumption directly on a trained model. Since the data-generating
process of natural language is inaccessible, we instead measure the
trained model's sensitivity to context truncation:
$\widehat{\TV}_w := \TV(p_\theta(\cdot \mid X_{1:t-1}),
p_\theta(\cdot \mid X_{t-w:t-1}))$, the change in the model's
next-token distribution when conditioning on only the most recent
$w$ tokens. On Qwen2.5-0.5B \citep{qwen2024}, sweeping $w \in
\{2, 4, \ldots, 256\}$ over $100$ prefixes from each of two domains
(NLTK Gutenberg books, GitHub Python source), the data are well
described by a power law $\widehat{\TV}_w \propto w^{-\alpha}$
(log-RMSE $0.08$--$0.14$) rather than an exponential
(log-RMSE $0.20$--$0.31$), with measured exponents
$\alpha_{\text{nat}} = 0.44$ and $\alpha_{\text{code}} = 0.38$. We
refer to this property as \emph{polynomial truncation sensitivity}
to emphasize that it is a property of the trained model under the
evaluation distribution, rather than a claim about the underlying
data process. The finding is qualitatively consistent with the
heavy-tailed statistics of natural language
\citep{piantadosi2014zipf} and the power-law spectra observed in
scaling-law analyses of large models \citep{kaplan2020scaling},
both of which would predict polynomial rather than geometric
behavior at the model level.

\paragraph{Our contribution.}
Motivated by the empirical observation, we develop an
information-theoretic framework under a \emph{polynomial
truncation-sensitivity} assumption
(Definition~\ref{def:poly-mixing}). The central contribution is a
characterization of the \emph{memory requirement} of suffix-only
cache policies---tight within that class under a two-sided Bayes-risk
condition---together with a rigorous \emph{delayed} operational
comparison at the rate level; an online single-letter achievability
and its universal and structural extensions are stated only as
conjectures.

\begin{enumerate}[label=\textbf{(C\arabic*)},leftmargin=*,itemsep=2pt,topsep=2pt]
\item \emph{Sequential Wyner--Ziv converse.} For any causal online
scheme achieving average KL distortion $\leq D$,
$\liminf_n \E[R_n] \geq R^\star(D)$
(Theorem~\ref{thm:asymptotic}), with a finite-horizon expectation form in
Appendix~\ref{app:thm-asym} (Proposition~\ref{thm:finite-sample},
Corollary~\ref{cor:expectation-bound}).

\item \emph{Window scaling for suffix-only policies.} A
sliding-window scheme achieves TV distortion $\varepsilon$ with
window $w = O(\varepsilon^{-1/\alpha})$ (Theorem~\ref{thm:E2}); a
\emph{conditional} converse shows $w = \Omega(\varepsilon^{-1/\alpha})$
is necessary \emph{within the class of suffix-only reconstructions}
under a two-sided Bayes-risk sensitivity condition
(Definition~\ref{def:two-sided}, Theorem~\ref{thm:window-lb}). The
per-token memory requirement of suffix-only policies is thus
$\Theta(\varepsilon^{-1/\alpha})$ (Corollary~\ref{cor:tight}).
Whether recurrent or propagating cache summaries can achieve a better
scaling is an open problem.

\item \emph{Operational comparison (rigorous); conjectured online
achievability.} At the level of coding rate we give a \emph{rigorous}
delayed, multi-letter operational Wyner--Ziv comparison: under
quantized query side information and explicit information-spectrum and
de-binning-stability assumptions, a suffix-only block code matches the
full-context rate up to a side-information continuity penalty and a redundancy that
vanishes with block length (Appendix~\ref{app:operational}). We
\emph{conjecture}, but do not prove, that an online single-letter
block-Markov scheme with window $w_n = n^{1/(\alpha+1)}$ attains
distortion within $O(n^{-\alpha/(\alpha+1)})$ of $D$ and rate within
the converse exponent (Conjectures~\ref{thm:achievability},
\ref{thm:achievability-tight}); the supporting covering argument is
not rigorous, as it assumes the auxiliary inherits the source's
truncation sensitivity.

\item \emph{Sink-plus-recent implication.} Under
$\alpha$-polynomial truncation sensitivity \emph{together with a
sink-stability assumption} (the fixed sink prefix does not spoil the
recent-window decay; not implied by truncation sensitivity alone,
since the sink set is non-contiguous), sink-plus-recent eviction---a
deployed baseline combining attention sinks with a recent-token
window---with $k$ tokens incurs KL distortion
$D = D^\star + O(k^{-2\alpha})$---a KL exponent double the TV exponent
$\alpha$ of Theorem~\ref{thm:E2}---(Corollary~\ref{cor:topk}, \S\ref{sec:topk}).
The factor of $2$ in the exponent arises from a local quadratic
relation $\KLs \leq C(\epsilon_{\min}) \cdot \TV^2$ valid in the
bounded-floor regime under bounded log-density
(Lemma~\ref{lem:kl-tv-bounded}); the standard Pinsker inequality
provides only the converse direction $\TV \leq \sqrt{\KLs/2}$. The
empirical KL decay is well-fit by a power law with exponents
$\alpha_{\KLs} \in [0.74, 1.04]$ across the two cross-model domains,
satisfying $\alpha_{\KLs}/\alpha_{\TV} \in [1.93, 2.38]$, bracketing
the factor-of-two value predicted by the local quadratic KL--TV
approximation. We interpret this agreement as empirical evidence
consistent with the local quadratic approximation, rather than as a
deterministic information-theoretic identity: the quadratic exponent
requires the bounded-floor condition of
Lemma~\ref{lem:kl-tv-bounded}, and the observed ratios are consistent
with the measured distributions operating in this bounded-floor
regime. The operational-smoothing lemma
(Lemma~\ref{lem:kl-tv-smoothed}) provides finite KL control when the
worst-case floor is vacuous, but it yields only a \emph{linear}
KL--TV exponent and therefore does not by itself explain the observed
doubling. A matched-policy estimate---measuring $\TV$ and $\KLs$
under the \emph{same} retention policy and budget grid---sharpens
this to $\alpha_{\KLs}/\alpha_{\TV} \in [1.93, 2.04]$ across both the
sliding and sink-plus-recent policies and across budget cutoffs
$k \ge 16$; the wider range reported above reflects a cross-policy
comparison (sink-plus-recent $\KLs$ against sliding-window $\TV$),
which inflates the apparent factor because the sink-plus-recent
$\TV$ exponent ($\approx 0.6$) is steeper than the sliding-window
one ($\approx 0.44$).
\end{enumerate}

Three further results are deferred to the appendix as they are not
needed for the main story: a \emph{universal scheme} oblivious to
$\alpha$ (Conjecture~\ref{thm:universal}), a continuous-latent
intrinsic-dimension lower bound (Conjecture~\ref{thm:intrinsic-dim}), and
a multi-layer Lagrangian rate allocation
(Proposition~\ref{thm:multi-WZ}).

\paragraph{Implications for deployed compression schemes.}
The polynomial truncation-sensitivity finding has quantitative
implications. Under the classical $\rho$-mixing assumption, a
window of $w$ tokens suppresses distortion exponentially in $w$,
and a few dozen tokens suffice for negligible error. Under
$\alpha$-polynomial sensitivity with $\alpha \approx 0.5$, doubling
the window halves the distortion only when $\alpha \approx 1$;
smaller $\alpha$ values imply sub-linear improvement and motivate
the empirical observation that deployed sliding-window schemes
(Mistral's $w = 4096$, StreamingLLM's $w$-plus-sinks) require
windows substantially larger than would be predicted under
geometric mixing.

\paragraph{Organization.}
Section~\ref{sec:related} situates our framework in the classical
information theory, mixing theory, and KV cache compression
literature. Section~\ref{sec:setup} introduces the formal setup.
Section~\ref{sec:main} states the four core results---the
sequential Wyner--Ziv converse, the polynomial sliding-window
upper bound, the conditional suffix-only lower bound with explicit
achievability, and the sink-plus-recent implication---and summarizes
the deferred universal scheme and structural extensions.
Section~\ref{sec:experiments} reports the empirical validation.
Section~\ref{sec:discussion} connects to deployed systems, and
Section~\ref{sec:conclusion} concludes. All proofs are deferred to
the appendices.

\section{Related Work}
\label{sec:related}

\paragraph{Information theory and mixing.}
Source coding with decoder side information \citep{wyner1976rate} and
its indirect \citep{witsenhausen1980indirect} and nonanticipative
\citep{massey1990causality,charalambous2013nonanticipative} variants
provide the rate--distortion backbone of our setup; rate--distortion
dimension \citep{kawabata1994rate} and high-resolution quantization
\citep{bennett1948spectra,gersho1979asymptotically} underlie the
intrinsic-dimension result (Conjecture~\ref{thm:intrinsic-dim}). On the
source side, strong mixing conditions \citep{bradley2005basic} give
geometric dependence decay $\rho^w$ for spectral-gap Markov chains,
whereas \emph{subgeometric}/\emph{polynomial} mixing
\citep{douc2009subgeometric} decays only as $w^{-\alpha}$. We adopt
polynomial truncation sensitivity as the operative assumption,
motivated by our measurements (\S\ref{sec:experiments}) and
consistent with the heavy-tailed statistics of language
\citep{piantadosi2014zipf} and neural scaling laws
\citep{kaplan2020scaling}.

\paragraph{Theoretical analyses of attention.}
Lipschitz properties of self-attention \citep{kim2021lipschitz,
dasoulas2021lipschitz}, rank collapse toward a low-rank subspace
\citep{dong2021attention}, the long-depth clustering limit
\citep{geshkovski2023mathematical}, and feed-forward layers as
key-value memories \citep{geva2020transformer,phuong2022formal}
inform our multi-layer analysis (\S\ref{sec:main}) but do not
directly characterize mixing.

\paragraph{KV cache compression and empirical observations.}
Deployed schemes fall into families that our rate--distortion
hierarchy organizes under one truncation-sensitivity condition:
heavy-hitter eviction \citep{zhang2023h2o,liu2023scissorhands,
ge2024fastgen}, sliding-window-plus-sinks
\citep{xiao2024efficient,jiang2023mistral7b}, prompt-aware selection
\citep{li2024snapkv,tang2024quest}, layer/head-adaptive budgets
\citep{cai2024pyramidkv,feng2024adakv}, continuous-latent compression
\citep{deepseek2024mla}, quantization \citep{liu2024kivi,kang2024gear},
and architectural variants \citep{ainslie2023gqa,shazeer2019fast,
llama3-2024}. A distinct, recent thread casts the
\emph{quantization} of retained entries as a rate--distortion
problem---near-optimal vector quantization
\citep{zandieh2025turboquant} and reverse-waterfilling bit allocation
\citep{zuo2026ratequant}---whereas the present work characterizes the
orthogonal \emph{temporal} axis (which tokens to retain, not how many
bits per token). The sink-plus-recent corollary
(Corollary~\ref{cor:topk}) applies directly to the sliding-window
family; the heavy-hitter family requires the additional (empirical)
assumption that high-attention tokens approximately coincide with
sinks plus recency. Empirically, attention-score concentration on
recent tokens and a sink \citep{xiao2024efficient} and sublinear
effective receptive fields \citep{chen2023longlora} are consistent
with our finding but do not measure decay exponents; intrinsic-dimension
measurements in vision \citep{pope2021intrinsic,ansuini2019intrinsic}
serve as order-of-magnitude anchors only.

\section{Preliminaries and Problem Formulation}
\label{sec:setup}

\paragraph{Notation.}
Let $X_1, X_2, \ldots \in \V$ denote tokens with $|\V| = V$, and let
$\F_t := \sigma(X_1, \ldots, X_t)$. The language model is a fixed
function $f$ producing logits $Z_t = f(X_{<t})$ via $L$ attention
layers; the next-token distribution is $p_t := \mathrm{softmax}(Z_t)$.
For each $t$, the side information is the query vector
$Q_t \in \R^k$ computed at the final layer.

\paragraph{On the abstraction of $Q_t$.}
We treat $Q_t$ as \emph{idealized exogenous} decoder side
information: the decoder receives the query the uncompressed model
would compute and reconstructs the next-token distribution from the
compressed cache and this query. This is a deliberate simplification:
in an actual transformer the query is itself produced from the
(possibly compressed) hidden states, so $Q_t$ and $\Chat_{\leq t}$
are not strictly independent, and the abstraction collapses the
multi-layer query/key/value structure into a single side channel.
Our results are thus bounds for this idealized single-channel
formulation; a fully recursive layer-wise treatment is a natural
extension (\S\ref{sec:limitations}).

\begin{assumption}[Bounded logits]\label{ass:bounded-logits}
There exists $B(n_{\max}) > 0$ such that
$\|Z_t\|_\infty \leq B(n_{\max})$ a.s.\ for $t \leq n_{\max}$;
consequently, $p_t \geq \epsilon := V^{-1}e^{-2B}$ componentwise.
\end{assumption}

\paragraph{Online compression.}
An \emph{online encoder} is a sequence
$\phi_t : \V^{t-1} \to \mathcal{C}_t$ producing codes
$\Chat_t = \phi_t(X_{1:t-1})$. The encoder is causal in the strict
sense that $\Chat_t$ depends only on the strict past $X_{<t}$, not on
$X_t$ itself; this matches the KV-cache compression setting, where
the cache at the time of predicting $X_t$ summarizes the previously
generated tokens $X_{<t}$. The compressed filtration is
$\Fhat_t := \sigma(\Chat_1, \ldots, \Chat_t)$ and satisfies
$\Fhat_t \subseteq \F_{t-1}$ by construction. The \emph{decoder}
produces $\tilde p_t := \psi_t(\Chat_{\leq t}, Q_t)$. The per-token
rate and distortion are
\[
R_n := \frac{1}{n}\sum_{t=1}^n H(\Chat_t \mid \Fhat_{t-1}),
\quad
D_n := \frac{1}{n}\sum_{t=1}^n \KL{p_t}{\tilde p_t}.
\]

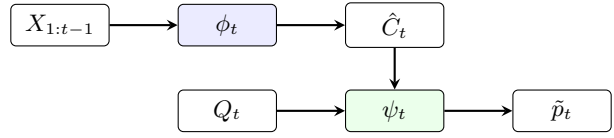
\begin{figure}[t]
\centering
\begin{tikzpicture}[node distance=1.2cm, font=\footnotesize]
\tikzstyle{box} = [draw, rectangle, minimum width=1.3cm, minimum height=0.55cm, rounded corners=2pt]
\tikzstyle{arrow} = [->, thick, >=stealth]
\node[box] (tokens) {$X_{1:t-1}$};
\node[box, right=0.9cm of tokens, fill=blue!8] (enc) {$\phi_t$};
\node[box, right=0.9cm of enc] (code) {$\Chat_t$};
\node[box, below=0.55cm of code, fill=green!8] (dec) {$\psi_t$};
\node[box, left=0.9cm of dec] (query) {$Q_t$};
\node[box, right=0.9cm of dec] (dist) {$\tilde p_t$};
\draw[arrow] (tokens) -- (enc);
\draw[arrow] (enc) -- (code);
\draw[arrow] (code) -- (dec);
\draw[arrow] (query) -- (dec);
\draw[arrow] (dec) -- (dist);
\end{tikzpicture}
\caption{Sequential Wyner--Ziv compression of the KV cache. The
encoder $\phi_t$ maps past tokens to $\Chat_t$; the decoder
$\psi_t$ reconstructs $\tilde p_t$ from
$(\Chat_{\leq t}, Q_t)$. Distortion $\KL{p_t}{\tilde p_t}$.}
\label{fig:setup}
\end{figure}

\begin{definition}[Sequential WZ rate function]
\label{def:Rstar}
An \emph{auxiliary process} $\{U_t\}$ is an $\F_t$-adapted sequence
that is $\Fhat_t$-measurable. The per-step pointwise WZ rate is
$\mathcal{R}_t^{\mathrm{WZ}}
:= I(X_t; U_t \mid \Fhat_{t-1}) - I(U_t; Q_t \mid \Fhat_{t-1})$.
The auxiliary is \emph{$D$-admissible} if some measurable decoder
$\psi_t = g_t(U_t, Q_t, \Fhat_{t-1})$ achieves
$\E[\KL{p_t}{\tilde p_t}] \leq D$ on average (the decoder may use the
realized code history, which is $\Fhat_{t-1}$-measurable). The sequential
Wyner--Ziv rate function is
\[
R^\star(D)
:= \inf_{\{U_t\}\ D\text{-admissible}}
\liminf_{n} \frac{1}{n}\sum_{t=1}^n \E[\mathcal{R}_t^{\mathrm{WZ}}].
\]
\end{definition}

\paragraph{Truncation-sensitivity assumption.}
The mixing-style condition we adopt below describes the
\emph{trained model's} sensitivity to context truncation, not the
mixing of the underlying data process (which is inaccessible). The
distinction is important: the quantity we measure and assume
properties of is the next-token distribution under the model
$p_\theta$, not the marginal of a hypothetical generating process.
We make this explicit in the definition.

\begin{definition}[Polynomial truncation sensitivity (average form)]\label{def:poly-mixing}
The language model $p_\theta$ has \emph{$\alpha$-polynomial truncation
sensitivity} with constants $C_{\mathrm{TS}} > 0$ and $\alpha > 0$ if
the average truncation error decays polynomially in the window:
\[
\sup_n \frac1n\sum_{t=1}^n
\E\,\TV\bigl(p_\theta(\cdot \mid X_{1:t-1}),\; p_\theta(\cdot \mid X_{t-w:t-1})\bigr)
\;\leq\; C_{\mathrm{TS}}\,w^{-\alpha}
\quad\text{for all } w \geq 1,
\]
where the expectation and average are taken over the context
distribution at evaluation and
$\TV(p, q) := \frac{1}{2}\sum_{x \in \V} |p(x) - q(x)|$. This is exactly
the quantity estimated in Section~\ref{sec:experiments} (mean/median
truncation curves across probe prefixes), so the assumption corresponds
directly to the experimental protocol.
\end{definition}

\begin{definition}[Uniform polynomial truncation sensitivity]\label{def:poly-mixing-unif}
The model has \emph{uniform} $\alpha$-polynomial truncation sensitivity
if the stronger almost-sure, per-position bound
\[
\TV\bigl(p_\theta(\cdot \mid X_{1:t-1}),\; p_\theta(\cdot \mid X_{t-w:t-1})\bigr)
\;\leq\; C_{\mathrm{TS}}\,w^{-\alpha}
\quad\text{a.s., for all } t, w
\]
holds (with respect to the joint distribution of $X_{1:t-1}$ at
evaluation). This implies the average form of
Definition~\ref{def:poly-mixing}. Proofs that invoke a per-step or
per-block truncation bound---the KL form
(Lemma~\ref{lem:poly-mixing-KL}), the block-Markov covering of
Appendix~\ref{app:achievability}, and the per-step sink/floor
steps---use this uniform form; the average-distortion results
(Theorem~\ref{thm:E2}, Corollary~\ref{cor:topk}) need only the average
form of Definition~\ref{def:poly-mixing}.
\end{definition}

\begin{remark}[Equivalent pointwise sufficient condition]
\label{rmk:pointwise-sufficient}
A stronger uniform-pointwise bound
$\sup_x |p_\theta(x \mid X_{1:t-1}) - p_\theta(x \mid X_{t-w:t-1})|
\leq C_{\mathrm{ptw}}\,w^{-\alpha}$
implies the uniform form (Definition~\ref{def:poly-mixing-unif}), and
hence the average form (Definition~\ref{def:poly-mixing}), with
$C_{\mathrm{TS}} \leq \frac{V}{2} C_{\mathrm{ptw}}$. The pointwise
form is convenient in some proofs (e.g.,
Lemma~\ref{lem:poly-mixing-KL}) but is generally not necessary; the
average-distortion main results require only the average TV form of
Definition~\ref{def:poly-mixing}.
\end{remark}

\begin{remark}[Naming]
We avoid the term ``polynomial mixing'' in isolation because it
suggests a property of the data-generating process. The quantity in
Definition~\ref{def:poly-mixing} is a property of the trained model
$p_\theta$ under the context distribution at evaluation. The average
form is precisely what we measure (Section~\ref{sec:experiments}); the
uniform a.s.\ form (Definition~\ref{def:poly-mixing-unif}) is a
stronger sufficient condition invoked only where a per-step bound is
needed. The classical $\rho$-mixing
condition for processes is strictly stronger when applied to
$p_\theta$: any process with $\rho^w$ truncation sensitivity is
$\alpha$-polynomially sensitive for every $\alpha > 0$, but the
converse fails. Polynomial and geometric truncation sensitivity are
therefore distinct model classes: the polynomial family $w^{-\alpha}$
does not reduce to a geometric law $C\rho^w$ under any $\alpha\to\infty$
limit, and we treat the two conditions separately rather than as
limits of one another.
\end{remark}

\section{Main Theoretical Results}
\label{sec:main}

\subsection{Core Result 1: Asymptotic Lower Bound}

\begin{theorem}[Asymptotic lower bound]\label{thm:asymptotic}
Under Assumption~\ref{ass:bounded-logits}, for any
encoder--decoder pair with $\limsup_n \E[D_n] \leq D$,
\[
\liminf_{n} \E[R_n] \;\geq\; R^\star(D).
\]
\end{theorem}

The proof (Appendix~\ref{app:thm-asym}) constructs three martingales
adapted to the compressed filtration $\{\Fhat_t\}$: $M_t^{(1)}$
encoding source entropy, $N_t$ encoding the log-likelihood ratio of
the decoder estimate, and $S_t$ encoding the pointwise
Wyner--Ziv rate. The argument is mixing-free---no assumption on
$\alpha$ or $\rho$ is needed for the asymptotic lower bound. A
finite-horizon expectation converse (Proposition~\ref{thm:finite-sample}),
with asymptotic form $\liminf_n\E[R_n] \ge R^\star(D)$ from
Theorem~\ref{thm:asymptotic}, is given in
Appendix~\ref{app:thm-asym} (Proposition~\ref{thm:finite-sample},
Corollary~\ref{cor:expectation-bound}).

\subsection{Core Result 2: Polynomial Sliding-Window Approximation}

The asymptotic bound is uninformative for finite-window schemes. The
following result bounds the \emph{distortion} of the suffix-only
reconstruction, with a window that is polynomial---not
logarithmic---in the target accuracy.

\begin{theorem}[Polynomial sliding-window distortion bound]\label{thm:E2}
Under Assumption~\ref{ass:bounded-logits} and the
$\alpha$-polynomial truncation-sensitivity condition
(Definition~\ref{def:poly-mixing}), the suffix-only reconstruction
$q_t = p_\theta(\cdot \mid X_{t-w:t-1})$ with window
\[
w \;=\; \left\lceil (C_{\mathrm{TS}}/\varepsilon)^{1/\alpha} \right\rceil
\]
incurs average total-variation distortion at most $\varepsilon$
relative to the full-context predictor:
\[
\tfrac1n\sum_{t=1}^n \E\,\TV\bigl(p_\theta(\cdot\mid X_{1:t-1}), q_t\bigr) \le \varepsilon.
\]
Equivalently, average TV distortion $\varepsilon$ is attained at
window $w = O(\varepsilon^{-1/\alpha})$. Under bounded log-density the
KL form follows from $\KLs \le C(\epsilon_{\min})\,\TV^2$
(Lemma~\ref{lem:kl-tv-bounded}), giving KL distortion $\varepsilon$ at
$w = O(\varepsilon^{-1/(2\alpha)})$.
\end{theorem}

The bound is immediate from the polynomial truncation-sensitivity
condition (Definition~\ref{def:poly-mixing}) and, for the KL form,
Lemma~\ref{lem:kl-tv-bounded}; no coding construction is needed for
this \emph{distortion} statement (the appendix version,
Theorem~\ref{thm:E2-poly}, records the explicit constants). The
required window $w \sim \varepsilon^{-1/\alpha}$ is \emph{polynomial}
in $1/\varepsilon$, in contrast to the \emph{logarithmic}
$w \sim \log(1/\varepsilon)$ under geometric mixing. For the empirical
$\alpha \in [0.38, 0.58]$ (\S\ref{sec:experiments}), KL distortion
$0.01$ requires windows of order $10^2$ tokens, consistent with
deployed sliding-window schemes
\citep{xiao2024efficient,jiang2023mistral7b}.

\begin{conjecture}[Mixing-sharpened concentration of the lower bound; not established]
\label{cor:mixing-conc}
Under $\alpha$-polynomial truncation sensitivity, and \emph{for the
realized-code-length rate} of Remark~\ref{rmk:pathwise-needs-codelen}
(not the average-conditional-entropy rate of
Proposition~\ref{thm:finite-sample}, which admits no pathwise
converse), we conjecture a high-probability concentration
$R_n \geq R^\star(D) - O\bigl(n^{-\alpha/(2(\alpha+1))} \sqrt{\log n}\bigr)$ (the sharper exponent $\alpha/(\alpha+1)$ under the
forward covariance-decay condition of
Remark~\ref{rmk:exponent-gap}), replacing the $1/(1-\rho)$ factor of
the geometric case by a polynomial-in-$\alpha$ dependence. The
supporting argument shares the covering/auxiliary-mixing step of
Conjecture~\ref{thm:achievability} and is therefore not established
here.
\end{conjecture}

\subsection{Core Result 3: Suffix-only window scaling (lower bound; conjectured rate-level achievability)}
\label{sec:achievability}

The lower bound (Theorem~\ref{thm:asymptotic}) and the
sliding-window approximation (Theorem~\ref{thm:E2}) leave open two
questions: (a) whether the window size $w \asymp
\varepsilon^{-1/\alpha}$ is \emph{necessary}, and (b) whether an
explicit causal scheme achieves $R^\star(D)$ at a finite rate of
convergence. We address both. For (a), the window scaling is optimal
\emph{within the suffix-only policy class}, which we establish by a
conditional window lower bound (requiring the two-sided Bayes-risk
condition of Definition~\ref{def:two-sided}). For (b), we do
\emph{not} establish an explicit online achievability; we give a
rigorous \emph{delayed} operational comparison instead
(Appendix~\ref{app:operational}, summarized below) and state the
information-rate optimality only as a conjecture
(Conjecture~\ref{thm:achievability}).

\paragraph{Window scaling for suffix-only policies (tight within the class).}
The sliding-window approximation (Theorem~\ref{thm:E2}) shows
$w = O(\varepsilon^{-1/\alpha})$ \emph{suffices} for TV distortion
$\varepsilon$. The following converse shows it is \emph{necessary}
\emph{within the suffix-only class}, under a two-sided sharpening of
the truncation-sensitivity condition that we state explicitly as an
additional assumption.

\begin{definition}[Two-sided truncation sensitivity, Bayes-risk form]
\label{def:two-sided}
Let $\mathcal{M}_w$ denote the class of suffix-only reconstructions
at step $t$: distributions $q_t$ measurable with respect to
$\sigma(X_{t-w:t-1}, Q_t)$---the length-$w$ suffix \emph{together
with} the decoder's query side information $Q_t$ (so the class
matches the information actually available to the decoder, which
receives $Q_t$ but no code about the deep past). The model has
\emph{$\alpha$-tight truncation
sensitivity} with constant $c_{\mathrm{TS}} > 0$ if, in addition to
Definition~\ref{def:poly-mixing}, the
\emph{Bayes risk} of approximating the full-context conditional by
the best suffix-only reconstruction is bounded below: for every
$w \geq 1$,
\[
\begin{split}
&\liminf_{n \to \infty} \frac{1}{n}\sum_{t=1}^n
\E\!\Bigl[\inf_{q_t \in \mathcal{M}_w}
   \TV\bigl(p_\theta(\cdot \mid X_{1:t-1}),\, q_t\bigr)\Bigr]\\
&\qquad\;\geq\; c_{\mathrm{TS}}\, w^{-\alpha}.
\end{split}
\]
This is the operationally relevant lower bound: it asserts directly
that \emph{no} window-$w$ reconstruction---not merely the truncated
conditional $p_\theta(\cdot \mid X_{t-w:t-1})$, and \emph{even one
that uses the full-context query $Q_t$ as side information}---can
approximate the full conditional better than $c_{\mathrm{TS}}
w^{-\alpha}$ on average. It avoids assuming that the truncated conditional is the
Bayes-optimal suffix reconstruction (which would require a tower
property not guaranteed for trained models under positional-encoding
effects). The companion upper bound of
Definition~\ref{def:poly-mixing} shows the truncated conditional
\emph{achieves} the $w^{-\alpha}$ rate, so the two-sided condition
brackets the Bayes risk at $\Theta(w^{-\alpha})$.
\end{definition}

Empirically validating the lower-bound half of
Definition~\ref{def:two-sided} would require estimating the
conditional variability among contexts that share a similar
length-$w$ suffix and query $Q_t$---a measurement we do not attempt
here and leave to
future work; the present evidence (\S\ref{sec:experiments}) directly
supports only the upper-bound half (the $w^{-\alpha}$ decay of
$\widehat{\TV}_w$).

\begin{theorem}[Conditional suffix-only lower bound (Bayes-risk implication)]
\label{thm:window-lb}
Call a scheme \emph{suffix-only} if its reconstruction at step $t$
is measurable with respect to $(X_{t-w:t-1}, Q_t)$---it may use the
decoder's query side information $Q_t$, but its codes carry no
information about the deep past $X_{1:t-w-1}$ beyond what is in the
length-$w$ suffix. Under
$\alpha$-tight truncation sensitivity
(Definition~\ref{def:two-sided}), any suffix-only scheme with
window $w$ incurs average TV distortion at least
$\Omega(w^{-\alpha})$. Consequently, achieving TV distortion
$\varepsilon$ requires $w = \Omega(\varepsilon^{-1/\alpha})$,
matching the sufficiency of Theorem~\ref{thm:E2}: the window
scaling $w \asymp \varepsilon^{-1/\alpha}$ is optimal within the
suffix-only class.
\end{theorem}

The proof (Appendix~\ref{app:window-lb}) is immediate from the
Bayes-risk form of Definition~\ref{def:two-sided}: any suffix-only
window-$w$ scheme produces a reconstruction in the class
$\mathcal{M}_w$ (it uses the query $Q_t$ and the windowed cache, but
no deep-past code), so its distortion is at least the Bayes risk
$\inf_{q \in \mathcal{M}_w}\TV(p_\theta(\cdot \mid X_{1:t-1}), q)
\geq c_{\mathrm{TS}} w^{-\alpha}$, with no claim about which
reconstruction is optimal. The restriction to suffix-only schemes
is essential: a general causal scheme may \emph{propagate} old
information through a recurrently updated code (e.g.\ a running
summary in $C_t$), in which case the last $w$ codes can encode
arbitrarily old context and the lower bound need not hold. Our
achievability scheme (Conjecture~\ref{thm:achievability}) is itself
suffix-only, so the upper and lower bounds apply to the same class
and the memory requirement of suffix-only sequential KV-cache
compression is $\Theta(\varepsilon^{-1/\alpha})$. Whether a
propagating scheme can beat this scaling is an open question.

\paragraph{Explicit candidate construction (conjectured finite-$n$ rate).}
We complement the window-scaling result with an explicit candidate
online scheme and a non-asymptotic analysis, whose rate guarantee is
conjectural (Conjecture~\ref{thm:achievability}).

\begin{conjecture}[Polynomial achievability (information-rate level; not established)]\label{thm:achievability}
Suppose Assumption~\ref{ass:bounded-logits} and
$\alpha$-polynomial truncation sensitivity
(Definition~\ref{def:poly-mixing}) with $\alpha \in (0, 1]$. There
exists a sequence of causal online (suffix-only) encoder--decoder
pairs $\{(\phi_t^{(n)}, \psi_t^{(n)})\}$ using window size
$w_n = \lceil n^{1/(\alpha+1)} \rceil$ with the following
guarantees.
\begin{enumerate}[label=(\alph*),leftmargin=*,itemsep=1pt,topsep=2pt]
\item \emph{Expected rate (sharp, no forward-decay hypothesis).}
\[
\begin{gathered}
\E[R_n] \leq R^\star(D) + O\!\bigl(n^{-\alpha/(\alpha+1)}\log n\bigr),\\
\E[D_n] \leq D + O\!\bigl(n^{-\alpha/(\alpha+1)}\bigr).
\end{gathered}
\]
\item \emph{High-probability rate (unconditional).} For any
$\delta \in (0,1)$, with probability $\geq 1-\delta$,
\[
\begin{split}
R_n \leq R^\star(D)
   + O\!\Bigl(&n^{-\alpha/(2(\alpha+1))}(\log n)\\
   &\sqrt{\log V \cdot \log(1/\delta)}\Bigr).
\end{split}
\]
\end{enumerate}
The implicit constants depend on $C_{\mathrm{TS}}, \alpha, V$, and
the auxiliary alphabet size, but not on $n$. For $\alpha > 1$ the
distortion exponent saturates at $\min(\alpha,1)/(\alpha+1)$ (the
boundary term $w^{-1}\log w$ then dominates the truncation term
$w^{-\alpha}$).
\end{conjecture}

We state this as a conjecture rather than a theorem. The block-Markov
covering argument attempted in Appendix~\ref{app:achievability} relies
on a window-restricted covering step and a block-typicality stability
claim (the non-rigorous step recorded in
Remark~\ref{rmk:failed-rate-transfer}) that treats the auxiliary
process as inheriting the next-token truncation sensitivity of the
source. This fails in general: an adapted auxiliary such as
$U_t = X_1$ retains the full deep past while the next-token
conditional need not, so the joint block law over $(U_t)$ is not
controlled by next-token sensitivity alone (see
Remark~\ref{rmk:failed-rate-transfer}). Consequently the
information-rate optimality is \emph{not} established here. A rigorous
rate-level statement, under different---delayed, multi-letter,
quantized side-information---assumptions, is the operational
comparison of Appendix~\ref{app:operational}.

\begin{remark}[Three rate statements and what each requires]
\label{rmk:exponent-gap}
The distortion exponent $\alpha/(\alpha+1)$ matches the converse
window scaling unconditionally. The three rate statements separate by
hypothesis: \emph{(a)} in expectation the sharp exponent
$\alpha/(\alpha+1)$ holds with no forward-decay hypothesis (only the
finite-memory covering regularity of
Assumption~\ref{ass:covering-reg}); \emph{(b)} in high probability
an Azuma--Hoeffding bound gives the looser
$\alpha/(2(\alpha+1))$; and \emph{(c)} the sharp exponent is restored
in high probability only under the forward covariance-decay condition
$(\dagger)$, which is \emph{not} implied by truncation sensitivity
(Conjecture~\ref{thm:achievability-tight}; details in
Appendix~\ref{app:achievability}).
\end{remark}

\paragraph{Operational (physical message-length) version.}
The conjecture above concerns an \emph{information} rate (a
conditional mutual information). Appendix~\ref{app:operational} gives
a \emph{rigorous} rate-level counterpart in terms of the physical
message length:
using a one-shot likelihood-encoder bound valid for the dependent
block law (Lemma~\ref{lem:op-oneshot}) and a coordinate-hybrid
transfer that avoids an $M^B$ alphabet blow-up
(Lemma~\ref{lem:op-hybrid}), Theorem~\ref{thm:op-transfer} shows a
suffix-only \emph{delayed} block code matches the full-context
optimum up to a single side-information continuity penalty
$\Omega_M(\bar\tau^Y_w)$ and a redundancy that vanishes under
explicit information-spectrum and de-binning-stability hypotheses
(Assumptions~\ref{ass:op-spectrum}, \ref{ass:op-debin}). That statement is delayed---the
encoder reads a whole block before coding, so it is not a zero-delay
online cache---and its convergence exponent is governed by the slower
query exponent $\alpha_Y$ rather than a single truncation exponent
(Corollary~\ref{cor:op-exponent}).

\begin{corollary}[Memory characterization for suffix-only policies]\label{cor:tight}
Under $\alpha$-tight truncation sensitivity
(Definition~\ref{def:two-sided}), the minimal window size
(equivalently, the per-token KV-cache memory) required for average
TV distortion $\varepsilon$ \emph{by a suffix-only policy} is
$\Theta(\varepsilon^{-1/\alpha})$, with matching upper
(Theorem~\ref{thm:E2}) and lower (Theorem~\ref{thm:window-lb})
bounds \emph{within that class}. Under the additional bounded-floor
condition of Lemma~\ref{lem:kl-tv-bounded}, the corresponding
KL-distortion memory scaling is
$\Theta(\varepsilon_{\KLs}^{-1/(2\alpha)})$ via the local quadratic
KL--TV relation. The
characterization does not preclude that a more general propagating
(non-suffix-only) scheme achieves a smaller window; establishing or
refuting this is an open problem.
\end{corollary}

\begin{remark}[On the strength of the characterization]
\label{rmk:char-strength}
The two bounds of Corollary~\ref{cor:tight} are close to operational
restatements of the two-sided condition itself: the upper bound is the
averaged form of the truncation-sensitivity assumption
(Definition~\ref{def:poly-mixing}) and the lower bound is its
Bayes-risk companion (Definition~\ref{def:two-sided}), each following
almost immediately. The contribution is therefore the
\emph{identification} of the suffix-only per-token memory scaling with
the truncation-sensitivity exponent $\alpha$, together with the
empirical finding (\S\ref{sec:experiments}) that $\alpha$ is
\emph{polynomial} rather than geometric; we do not claim a nontrivial
derivation beyond the stated two-sided condition, and the lower-bound
half is assumed, not measured.
\end{remark}

\paragraph{Additional results (deferred to the appendix).}
Two further results round out the framework but are not needed for
the main story. (i) A \emph{universal scheme}, oblivious to $\alpha$,
attains the rate exponent of Conjecture~\ref{thm:achievability}
simultaneously for every $\alpha \in [\alpha_{\min}, \alpha_{\max}]$
at the cost of a $\log n$ factor, using a logarithmic grid of window
sizes with online selection (Conjecture~\ref{thm:universal},
Appendix~\ref{app:universal}). (ii) Two \emph{structural extensions}
connect the framework to deployed architectures: a continuous-latent
intrinsic-dimension bound matching the Kawabata--Dembo
rate--distortion dimension (Conjecture~\ref{thm:intrinsic-dim},
Appendix~\ref{app:cont-extension}), and a multi-layer
reverse-water-filling rate allocation
(Proposition~\ref{thm:multi-WZ}, Appendix~\ref{app:multilayer-proofs}).
Both are mixing-independent; the multi-layer allocation is proved in
the appendix, while the intrinsic-dimension bound is stated there as a
conjecture.

\subsection{Core Result 4: Sink-Plus-Recent Eviction}
\label{sec:topk}

A common KV cache compression baseline retains the first few tokens
(attention sinks \citep{xiao2024efficient}) together with the most
recent $k - O(1)$ tokens. We refer to this scheme as
\emph{sink-plus-recent eviction}; it is distinct from
attention-score-based heavy-hitter eviction (e.g., H2O
\citep{zhang2023h2o}, which we discuss separately in
Section~\ref{sec:discussion}). Under polynomial truncation sensitivity, the
distortion location of sink-plus-recent eviction follows
from Theorem~\ref{thm:E2} together with a sink-stability condition
(stated in the corollary).

\begin{corollary}[Sink-plus-recent rate--distortion]\label{cor:topk}
Under $\alpha$-polynomial truncation sensitivity
(Definition~\ref{def:poly-mixing}), Assumption~\ref{ass:bounded-logits},
and a \emph{uniform reconstruction floor}
$\tilde p_t(x) \geq \epsilon_{\mathrm{rec}} > 0$ for all $x$ (e.g.\
a decoder outputting softmax of bounded logits, or a fixed smoothed
decoder whose smoothing bias is below the target distortion), and a
\emph{sink-stability} condition
\[
\tfrac1n\textstyle\sum_t\E\,\TV\bigl(p_\theta(\cdot\mid X_{1:t-1}),\,
p_\theta(\cdot\mid \mathrm{sink},X_{t-k:t-1})\bigr)\le
C_{\mathrm{sink}}\,k^{-\alpha}
\]
(the fixed $O(1)$-size sink prefix does not spoil the recent-window
decay), sink-plus-recent eviction at a $k$-token budget incurs KL
distortion
\[
D \;=\; D^\star + O(k^{-2\alpha}),
\]
where $D^\star$ is the full-cache distortion floor. The factor of
$2$ in the exponent arises from the small-perturbation quadratic
relation $\KLs \leq (2/\epsilon_{\mathrm{rec}})\TV^2$ under the
reconstruction floor (Lemma~\ref{lem:kl-tv-bounded}), not from
Pinsker's inequality. Without a uniform floor, the
operational-smoothing route gives only the linear relation
$\KLs = O(\TV)$ and hence the weaker exponent $\alpha$
(Remark~\ref{rmk:application-sink}).
\end{corollary}

\begin{proof}[Proof sketch]
By the sink-stability condition, the sink-plus-recent context
approximates the full conditional in TV at rate $C_{\mathrm{sink}}
k^{-\alpha}$. This is an \emph{added assumption}, not a consequence
of Definition~\ref{def:poly-mixing} alone: the sink set is a
non-contiguous subset (a fixed prefix plus a recent window), and
adding sink tokens to a recent window is not guaranteed by
truncation sensitivity to preserve the $k^{-\alpha}$ rate. Under
Assumption~\ref{ass:bounded-logits}, the model's next-token
distribution satisfies $p_t(x) \geq \epsilon_{\min}$ for all $x$;
Lemma~\ref{lem:kl-tv-bounded} (appendix) then yields
$\KLs(p_t \,\|\, \tilde p_t) \leq (2/\epsilon_{\min}) \TV(p_t,
\tilde p_t)^2$ in the small-TV regime, giving the KL distortion gap
$O(k^{-2\alpha})$ at a $k$-token budget.
\end{proof}

\paragraph{Why $\KLs \leq C \cdot \TV^2$ and not Pinsker.}
The standard Pinsker inequality is $\TV \leq \sqrt{\KLs/2}$,
yielding $\KLs \geq 2\TV^2$---a lower bound on $\KLs$, not an upper
bound. The opposite direction $\KLs \leq C \cdot \TV^2$ requires
additional structure on the distributions. Under
Assumption~\ref{ass:bounded-logits}, the next-token distribution
$p_t$ is bounded away from zero by
$\epsilon_{\min} := V^{-1} e^{-2B}$. For two such distributions
with $\TV(p, q) \leq \epsilon_{\min}/2$, the quadratic local
approximation $\KLs(p \,\|\, q) = \chi^2(p, q)/2 + O(\TV^3)$
combined with $\chi^2(p, q) \leq \TV(p, q)^2/\epsilon_{\min}$ yields
the desired bound; we prove this as Lemma~\ref{lem:kl-tv-bounded}.

The corollary predicts a power-law KL distortion gap with exponent
$2\alpha$, where $\alpha$ is the TV-decay exponent of
Theorem~\ref{thm:E2}. This yields a quantitative
internal-consistency prediction with two empirical checks: (i) the
absolute values of $\alpha$ from the TV measurement and $2\alpha$
from the KL measurement, and (ii) the ratio between them should
bracket $2$. Both are tested in
Section~\ref{sec:experiments}.

\paragraph{Comparison to geometric mixing.}
Under the geometric assumption $D \propto \rho^k$, the distortion
gap vanishes exponentially: $k \approx 30$ suffices for $D < 10^{-2}$
at $\rho = 0.85$. Under polynomial truncation sensitivity with $\alpha = 0.5$ (so
KL decay exponent $2\alpha = 1.0$), achieving $D < 10^{-2}$ requires
$k \sim 10^{2}$ tokens---a $3$--$4\times$ increase, not the
exponential separation that one might naively expect from changing
the decay form. The empirical sink-plus-recent measurements
(Appendix Figure~\ref{fig:topk}) follow the polynomial decay closely:
Natural at $k = 512$ yields $D \approx 0.04$, consistent with
$k^{-1.04}$.

\paragraph{Sub-optimality vs.\ continuous-latent schemes.}
The sink-plus-recent rate is $O(k \log V)$ bits per token; for
$k = 128$ and $V = 32{,}000$, this is approximately $1900$ bits per
token, two orders of magnitude above the lower bound $R^\star(D)$.
The gap motivates continuous-latent schemes such as MLA
\citep{deepseek2024mla}, analyzed under our framework in
Conjecture~\ref{thm:intrinsic-dim}.

\paragraph{Relation to heavy-hitter eviction.}
The corollary applies to sink-plus-recent eviction, where the
retained tokens are determined by position (first few + last
$k - O(1)$). Heavy-hitter eviction schemes such as H2O
\citep{zhang2023h2o} and Scissorhands
\citep{liu2023scissorhands} select tokens based on attention scores
rather than position. The connection to our framework requires an
additional lemma stating that high-attention tokens approximately
coincide with sinks plus recent positions; this is an empirical
property of trained models and is not subsumed by mixing alone. We
do not analyze heavy-hitter schemes in this work.

\section{Empirical Validation}
\label{sec:experiments}

We first test two predictions of the framework on Qwen2.5-0.5B
\citep{qwen2024}, then replicate across models and domains:

\begin{enumerate}[label=\textbf{(P\arabic*)},leftmargin=*,itemsep=2pt,topsep=2pt]
\item \textbf{Polynomial TV decay.} The
conditional-distribution shift $\widehat{\TV}_w$ decays as
$w^{-\alpha}$ for some $\alpha > 0$, with power-law fit superior to
exponential fit.

\item \textbf{Sink-plus-recent KL decay with doubled exponent.} The
sink-plus-recent KL distortion decays as $k^{-2\alpha}$, where
$\alpha$ is the TV-decay exponent from (P1). The factor of $2$ in
the exponent follows from the quadratic local relation between
$\KLs$ and $\TV$ under bounded log-density
(Lemma~\ref{lem:kl-tv-bounded}, Corollary~\ref{cor:topk}).
\end{enumerate}

\paragraph{Experimental setup.}
We use four open models spanning two families and a $6\times$ range
of parameter counts---Qwen2.5-0.5B/1.5B/3B \citep{qwen2024} and
SmolLM2-1.7B \citep{allal2025smollm2}---with FP16 precision on a
single $16$\,GB GPU (NVIDIA RTX 5080), batch size $1$ and no
gradient computation. Qwen2.5-0.5B serves as the primary model for
the detailed ablations; the remaining models are used for the
cross-model replication. For the primary short-window and cross-model
measurements we sample $100$ prefixes of length $1024$
tokens from each of two domains: NLTK Gutenberg books (concatenation
of classic literary works, $\sim 6.4$M characters) as a natural
language source, and a concatenation of large Python source files
from prominent open-source projects (numpy, pandas, CPython,
requests; $\sim 0.6$M characters) as a code source. The
domain-robustness experiment (\S\ref{app:crossdomain}) adds Wikipedia
prose and structured JSON, and the long-window tail
(\S\ref{sec:exp-longwindow}) uses $16{,}500$-token prefixes. All
measurements use a position-preserving truncation protocol (below),
which we verify is free of positional-encoding artifacts
(\S\ref{sec:exp-ablation}).

\paragraph{Measurement protocol.}
For each prefix $X_{1:t}$ and each window size $w$, we measure
$\widehat{\TV}_w = \TV(p_t, p_t^{(w)})$ where $p_t$ is the full
next-token distribution and $p_t^{(w)} = p(\cdot \mid X_{t-w:t-1})$
is the distribution conditioned on the last $w$ tokens. To avoid
conflating genuine truncation sensitivity with a positional-encoding
artifact (\S\ref{sec:limitations}), we adopt a
\emph{position-preserving} protocol: the retained window is fed with
its \emph{original} absolute position indices $t-w, \ldots, t-1$
rather than reset to $0, \ldots, w-1$. We use a short window grid
$w \in \{2, 4, 8, 16, 32, 64, 128, 256\}$ for the primary and
cross-model measurements, an extended grid up to $w = 4096$ for the
fit competition, and a long-prefix tail sweep up to $w = 8192$
(with $16{,}500$-token prefixes) for the robustness check in
\S\ref{sec:exp-longwindow}. For
each $k \in \{8, 16, 32, 64, 128, 256, 512\}$, we measure
$\KL{p_t}{p_t^{\text{Top-}K}}$ where $p_t^{\text{Top-}K}$ is the
distribution conditioned on the first $4$ tokens (attention sinks)
plus the most recent $k - 4$ tokens. We also evaluate a Random-$K$
baseline conditioning on a uniformly random subset of $k$ tokens
preserving order. Curve fits are performed in log space; for model
selection we compare against exponential, stretched-exponential, and
broken-power-law alternatives by AIC (\S\ref{sec:exp-longwindow}).

\paragraph{Results: power-law decay of TV.}
Figure~\ref{fig:mixing} shows the measured TV decay on the primary
model (Qwen2.5-0.5B). Over the short-window grid, power-law fits
yield exponents $\alpha_{\text{nat}} = 0.44 \pm 0.05$ and
$\alpha_{\text{code}} = 0.38 \pm 0.06$ (bootstrap $95\%$ CI from
$300$ resamples over prefixes; intervals
$[0.39, 0.48]$ and $[0.33, 0.46]$), with log-RMSE values:

\begin{center}
\small
\setlength{\tabcolsep}{4pt}
\begin{tabular}{lcc}
\toprule
Fit form & Natural (log-RMSE) & Code (log-RMSE) \\
\midrule
Exponential $C\rho^w$ & $0.31$ & $0.20$ \\
Power law $C w^{-\alpha}$ & $\mathbf{0.14}$ & $\mathbf{0.08}$ \\
\bottomrule
\end{tabular}
\end{center}

The power-law fit improves log-RMSE by roughly $2$--$2.5\times$
relative to exponential, supporting prediction (P1). We do not claim a
\emph{strict} power law: a stretched exponential can match the
in-sample curvature marginally better, so the defensible claim is that
the decay is robustly \emph{sub-geometric}, with a power law providing
a parsimonious and extrapolatively strong fit
(\S\ref{sec:exp-longwindow}). The exponents
are recovered independently from the sink-plus-recent KL decay
(below) and replicate across models (\S\ref{sec:exp-crossmodel}).

\paragraph{Results: doubled sink-plus-recent exponent.}
The sink-plus-recent KL decay
(Appendix~\ref{app:robustness}, Figure~\ref{fig:topk})
yields power-law exponents $\alpha_{\text{nat,KL}} = 1.04$ and
$\alpha_{\text{code,KL}} = 0.74$, with high fit quality in both
cases. The ratios to the TV-decay exponents from (P1) are
\[
\frac{\alpha_{\text{nat,KL}}}{\alpha_{\text{nat,TV}}}
= \frac{1.04}{0.44} = 2.38,
\quad
\frac{\alpha_{\text{code,KL}}}{\alpha_{\text{code,TV}}}
= \frac{0.74}{0.38} = 1.93,
\]
bracketing the factor-of-two predicted by the local quadratic
KL--TV approximation (Corollary~\ref{cor:topk},
Lemma~\ref{lem:kl-tv-bounded}). A direct check of the underlying
relation is shown in Appendix~\ref{app:kltv-ablation}
(Figure~\ref{fig:kltv}): plotting $\KLs$ against
$\TV^2$ across budgets yields an approximately linear,
through-the-origin trend in both domains, consistent with the
bounded-floor relation $\KLs \le C(\epsilon_{\min})\,\TV^2$ holding
with an effectively constant $C$ over the measured range. This is the
internal-consistency check (P2): two independent
measurements---window-truncation TV decay and sink-plus-recent
KL decay---are related by the local quadratic KL--TV approximation,
and the measured ratios are consistent with the prediction. The
agreement is empirical evidence supporting the framework's claim
that a single underlying TV-exponent $\alpha$ governs both the
sliding-window approximation and the sink-plus-recent distortion
gap. We do not claim a deterministic information-theoretic
identity between the two exponents; the doubling holds in the
bounded-floor regime of Lemma~\ref{lem:kl-tv-bounded}, which the
$\KLs$-vs-$\TV^2$ trend indicates the measured distributions occupy.

\paragraph{Results: sink-plus-recent versus Random-$K$.}
The Random-$K$ baseline conditions on a uniformly random subset of
$k$ tokens (preserving order); this isolates the contribution of
\emph{positional structure} to the sink-plus-recent performance.
Random-$K$ exhibits essentially
no decay over the measured range (Figure~\ref{fig:topk}): the KL distortion at $k = 8$
($5.6$ Natural, $9.6$ Code) is comparable to that at $k = 512$
($3.4$, $4.1$). At $k = 512$, sink-plus-recent attains roughly
$85\times$ (Natural) and $30\times$ (Code) lower median KL distortion
than Random-$K$ (a ratio of median KL values; raw medians in
Table~\ref{tab:randomk}). The asymmetric retention of sinks and
recent tokens---not the budget alone---is the operative mechanism.

\paragraph{Robustness (positional ablation and long-window tail).}
Two robustness checks, detailed in
Appendix~\ref{app:robustness}, support the measurement. First, a
\emph{positional-encoding ablation} confirms the decay is genuine
content forgetting rather than a position-shift artifact: a
position-preserving protocol (retained tokens keep their original
absolute indices) and the naive fresh-sequence protocol give
exponents that agree to three decimals
($\alpha_{\text{nat}} = 0.437$ vs.\ $0.438$;
$\alpha_{\text{code}} = 0.383$ for both;
Figure~\ref{fig:ablation}). All measurements in this section use the
position-preserving protocol. Second, with a long prefix
($16{,}500$ tokens) the power law persists with the short-window
exponent out to $w = 8192$---well beyond the deployed $4$k
regime---while an exponential is rejected by an order of magnitude in
extrapolation (Figure~\ref{fig:longwindow}); we use a long prefix
specifically so that the largest window still truncates a substantial
fraction of the context, avoiding the measurement-boundary effect
that arises when the window approaches the prefix length.

\paragraph{Cross-model replication.}
\label{sec:exp-crossmodel}
To test whether polynomial truncation sensitivity is an artifact of
a single model, we measure the TV exponent on four models spanning
$0.5$--$3$B parameters and two families
(Table~\ref{tab:crossmodel}). The power law holds in every case, and
two trends emerge. First, the natural-language exponent increases
monotonically with model scale within the Qwen2.5 family
($0.44 \to 0.51 \to 0.58$ from $0.5$B to $3$B), indicating that
larger models recover the next-token distribution from a short
recent window more quickly. This scaling trend is suggestive but
rests on three Qwen2.5 sizes only (with the larger extension point
quantized), so we read it as indicative within the Qwen2.5 family
rather than a general claim.
Second, the ordering $\alpha_{\text{nat}} > \alpha_{\text{code}}$ is
consistent across all four models. The phenomenon is thus not
specific to one model or family, and the exponent---while
model-dependent in magnitude---is a stable, measurable quantity, as
the theory requires. A separate $4$-bit run extending the Qwen2.5
family to $7$B (natural-language domain) gives
$\alpha_{\text{nat}} \approx 0.57$, confirming that polynomial
truncation sensitivity persists at the $7$B scale; the exponent is
essentially unchanged from $3$B, suggesting it saturates rather than
continuing to grow.

\begin{table}[t]
\centering
\small
\setlength{\tabcolsep}{6pt}
\begin{tabular}{lccc}
\toprule
Model & Params & $\alpha_{\text{nat}}$ & $\alpha_{\text{code}}$ \\
\midrule
Qwen2.5-0.5B   & $0.5$B & $0.437$ & $0.383$ \\
Qwen2.5-1.5B   & $1.5$B & $0.508$ & $0.439$ \\
SmolLM2-1.7B   & $1.7$B & $0.574$ & $0.371$ \\
Qwen2.5-3B     & $3$B   & $0.575$ & $0.417$ \\
\bottomrule
\end{tabular}
\caption{Cross-model TV-decay exponents (position-preserving
protocol, short-window grid, $100$ prefixes per domain). The power
law holds for every model; $\alpha_{\text{nat}} > \alpha_{\text{code}}$
throughout, and $\alpha_{\text{nat}}$ grows with scale within the
Qwen2.5 family.}
\label{tab:crossmodel}
\end{table}

\paragraph{Comparison of domains: natural language is more truncation-sensitive than code.}
\label{sec:exp-domains}
Across all four models, the natural-language exponent exceeds the
code exponent: on the primary model, $\alpha_{\text{nat}} = 0.44$
versus $\alpha_{\text{code}} = 0.38$, and the ordering
$\alpha_{\text{nat}} > \alpha_{\text{code}}$ holds for every model in
the cross-model sweep (\S\ref{sec:exp-crossmodel},
Table~\ref{tab:crossmodel}). A larger exponent corresponds to
\emph{faster} decay of truncation sensitivity, so the next-token
distribution in literary text is, on average, more quickly
recovered from a short recent window than in Python source. We offer
two non-mutually-exclusive explanations as hypotheses:
(i) code carries persistent, long-range structure---open scopes,
imported names, variables defined far earlier---that keeps distant
tokens weakly informative for the next token, slowing the average
decay; (ii) much of literary next-token prediction is governed by
local lexical and grammatical context, so that beyond a modest
window the residual dependence on distant tokens is small on
average. We emphasize that the relevant quantity is the
\emph{average} dependence decay across positions, not the worst-case
long-range structure: code clearly contains long-range dependencies,
but they bind a minority of next-token decisions. To test whether
the polynomial form is specific to literary text and Python, we
measured the TV exponent on two further domains---Wikipedia-style
prose and structured JSON---on the primary model. The power law is
favored over an exponential by $2$--$4\times$ in log-RMSE in every
domain, while the exponent varies only modestly
($\alpha \in [0.39, 0.44]$): the polynomial \emph{form} is
domain-robust, with a weak, not strongly systematic, dependence of
$\alpha$ on domain structure at this scale
(Appendix~\ref{app:crossdomain}, Table~\ref{tab:crossdomain}).

\paragraph{Cache-policy degradation.}
\label{sec:exp-policy}
The distribution-level findings above translate directly into the
degradation curves of concrete cache policies. On the primary model
we compare five policies at a fixed budget $k$---full context,
sliding window, sink-plus-recent, Random-$K$ (a uniformly random
$k$-token subset), and a simple attention-score heavy-hitter---by the
median $\KLs(\text{full}\,\|\,\text{policy})$ and the mean increase in
negative log-likelihood (NLL) of the actual next token
(Figure~\ref{fig:policy}, Appendix~\ref{app:robustness}). The heavy-hitter policy here selects the
$k$ tokens with the largest attention mass from the final query
position, averaged over the last layer's heads (a one-step,
last-query score rather than H2O-style cumulative attention), with
the anchor and most-recent token always retained and the
position-preserving protocol applied to the selected set; it is thus
a lightweight attention-score baseline, not a tuned eviction policy.
Three observations stand out. First, the
recency-based policies (sliding and sink-plus-recent) dominate the
others by one to two orders of magnitude: at $k = 512$ they incur
median KL $\approx 0.011$ (Natural) and $\approx 0.003$ (Code),
versus $1.9$ and $2.1$ for Random-$K$---a roughly $100$--$700\times$
reduction. Second, the heavy-hitter policy sits between: better than
random (attention scores do carry signal) but well behind recency,
consistent with the sink-plus-recent positions capturing most of the
high-mass context. Third, every policy's distortion \emph{decays as a
power law in $k$}, the predicted consequence of polynomial truncation
sensitivity. Sliding and sink-plus-recent are nearly
indistinguishable here (the attention sink adds little at prefix
length $1024$); their separation is expected to grow at the much
longer contexts where sinks were originally motivated. Full per-domain
curves, the NLL panels, and the heavy-hitter/sink overlap are in
Appendix~\ref{app:robustness}.

\paragraph{Joint and query-channel sensitivity (validating the side-information rate).}
The operational achievability of Appendix~\ref{app:operational} is
driven by the decay of two channels: the predictive channel
(sensitivity $\alpha_X$, established above) and the \emph{quantized
next-step query} that serves as decoder side information (sensitivity
$\alpha_Y$). We measure both, and the joint operator-aware exponent
$\alpha_J$ of the super-symbol $(p_t, \kappa_M(\cdot\mid q_t))$, on two
model scales (Qwen2.5-0.5B and 1.5B; final layer, context $2048$,
$100$ held-out contexts each). The side information is a fixed random
projection ($d_r{=}64$) of the pre-RoPE query, softmax-quantized to
$M$ codepoints with the inverse temperature calibrated on a disjoint
split to a target informativeness $I(Q;Y)/\log M \approx 0.20$.
Figure~\ref{fig:joint-sensitivity} summarizes the result. \emph{(i)}
The predictive channel is $M$-independent, with $\alpha_X=0.41$
$[0.35,0.47]$ at $0.5$B and $0.45$ $[0.39,0.51]$ at $1.5$B.
\emph{(ii)} The quantized query channel decays \emph{more slowly} than
the predictive channel at both scales, and its exponent \emph{falls
sharply with scale}, $\alpha_Y\approx0.33\to0.19$ at $M{=}16$;
consequently the query total variation overtakes the predictive one
near $w\approx128$ at $1.5$B (panel~a), whereas at $0.5$B the two only
converge by $w{=}512$ without crossing. We read the slow query decay
as side information carrying context dependence that persists past the
window---the property that makes it valuable to the Wyner--Ziv
decoder---while cautioning that this is an interpretation of the
measured crossover, not a separate measurement. \emph{(iii)} The joint
exponent satisfies $\alpha_Y<\alpha_J<\alpha_X$ at both scales and
across all $M$ (primary $\alpha_J\in[0.34,0.38]$ at $0.5$B,
$[0.29,0.33]$ at $1.5$B), and is stable across operating point,
projection seed, projection dimension, and layer band ($19$
configurations per scale, the trivial layer-$0$ query excepted);
$\alpha_J$ drifts down only mildly with scale ($0.35\to0.32$ at
$M{=}16$). As in E1, a power law decisively beats a pure exponential
at both scales---by AIC (e.g.\ $-45$ vs $-34$ at $0.5$B, $-48$ vs
$-35$ at $1.5$B) and by held-out extrapolation (mean-squared error
$0.30$ vs $0.68$ and $0.24$ vs $0.67$); a stretched exponential
captures mild residual curvature slightly better in-sample, but the
decay is clearly sub-geometric. The cross-scale comparison in
panel~(b) is made at $M{=}16$, where the quantizer is well-calibrated
at both scales; at $0.5$B the calibration saturates for
$M\in\{32,64\}$ (achieved informativeness $\approx0.03$--$0.05$ rather
than $0.20$, a consequence of more clustered query geometry at the
smaller scale), so those operating points are reported but excluded
from the cross-scale trend. Consistent with
Corollary~\ref{cor:op-exponent}, the resulting suffix-only convergence
is governed by the slower query exponent $\alpha_Y$, not by
$\alpha_X$ or the joint summary $\alpha_J$. We observe $\alpha_Y$
decrease across the two scales we measured; with only two scales this
is a measured decrease rather than an established scaling trend.

\definecolor{cPred}{HTML}{1f77b4}
\definecolor{cQuery}{HTML}{d62728}
\definecolor{cJoint}{HTML}{2ca02c}
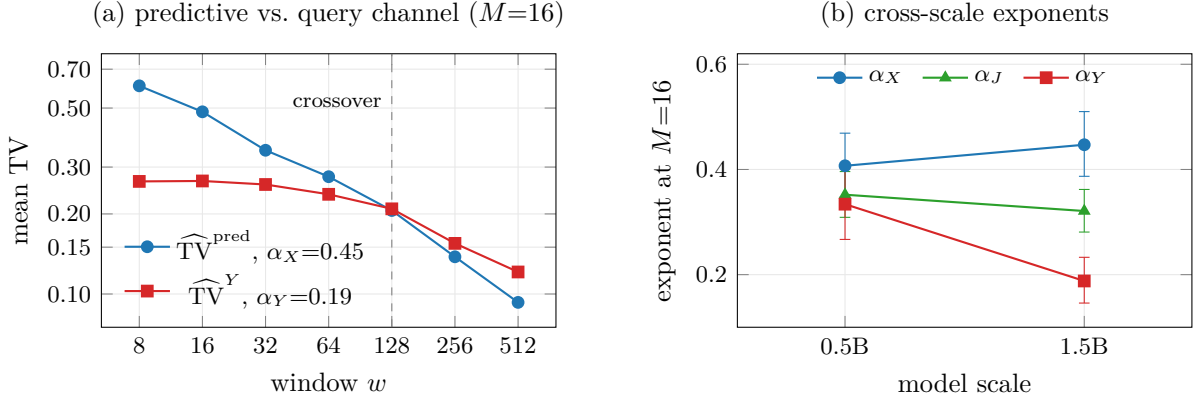
\begin{figure*}[t]
\centering
\begin{tikzpicture}
\begin{groupplot}[
  group style={group size=2 by 1, horizontal sep=2.4cm},
  width=0.46\textwidth, height=5.2cm,
  tick label style={font=\footnotesize}, label style={font=\small},
  title style={font=\small},
  legend style={font=\footnotesize, draw=none, fill=none},
]
\nextgroupplot[
  xmode=log, ymode=log, log basis x=2,
  xlabel={window $w$}, ylabel={mean TV},
  title={(a) predictive vs.\ query channel ($M{=}16$)},
  xtick={8,16,32,64,128,256,512}, xticklabels={8,16,32,64,128,256,512},
  ytick={0.1,0.15,0.2,0.3,0.5,0.7},
  yticklabels={0.10,0.15,0.20,0.30,0.50,0.70},
  ymin=0.075, ymax=0.80, legend pos=south west,
  grid=both, grid style={gray!18}]
\addplot[cPred, mark=*, thick] coordinates {
 (8,0.606)(16,0.484)(32,0.347)(64,0.276)(128,0.206)(256,0.138)(512,0.093)};
\addlegendentry{$\widehat{\TV}^{\mathrm{pred}}$, $\alpha_X{=}0.45$}
\addplot[cQuery, mark=square*, thick] coordinates {
 (8,0.265)(16,0.266)(32,0.258)(64,0.237)(128,0.209)(256,0.155)(512,0.121)};
\addlegendentry{$\widehat{\TV}^{Y}$, $\alpha_Y{=}0.19$}
\draw[gray, dashed] (axis cs:128,0.075) -- (axis cs:128,0.80);
\node[font=\scriptsize, anchor=south east] at (axis cs:128,0.46) {crossover};
\nextgroupplot[
  xlabel={model scale}, ylabel={exponent at $M{=}16$},
  title={(b) cross-scale exponents},
  xtick={0,1}, xticklabels={0.5B,1.5B}, xmin=-0.45, xmax=1.45,
  ymin=0.10, ymax=0.62,
  legend style={at={(0.5,0.98)}, anchor=north, legend columns=-1,
  /tikz/every even column/.append style={column sep=5pt}},
  grid=both, grid style={gray!18}]
\addplot[cPred, mark=*, thick, error bars/.cd, y dir=both, y explicit]
 coordinates {
 (0,0.407) += (0,0.062) -= (0,0.059)
 (1,0.447) += (0,0.063) -= (0,0.060)};
\addlegendentry{$\alpha_X$}
\addplot[cJoint, mark=triangle*, thick, error bars/.cd, y dir=both, y explicit]
 coordinates {
 (0,0.352) += (0,0.044) -= (0,0.043)
 (1,0.321) += (0,0.041) -= (0,0.040)};
\addlegendentry{$\alpha_J$}
\addplot[cQuery, mark=square*, thick, error bars/.cd, y dir=both, y explicit]
 coordinates {
 (0,0.334) += (0,0.068) -= (0,0.067)
 (1,0.188) += (0,0.045) -= (0,0.042)};
\addlegendentry{$\alpha_Y$}
\end{groupplot}
\end{tikzpicture}
\caption{\textbf{Joint and query-channel truncation sensitivity.}
(a) At $1.5$B (Qwen2.5-1.5B, $M{=}16$): mean total variation of the
predictive channel ($\widehat{\TV}^{\mathrm{pred}}$) and the
quantized-query side-information channel ($\widehat{\TV}^{Y}$) versus
window $w$ (log--log). The query channel decays more slowly
($\alpha_Y<\alpha_X$) and overtakes the predictive channel near
$w\approx128$. (b) Cross-scale fitted exponents at $M{=}16$ with
bootstrap $95\%$ confidence intervals (Qwen2.5-0.5B and 1.5B). At both
scales $\alpha_Y<\alpha_J<\alpha_X$; with scale the predictive
exponent rises slightly while the query exponent falls sharply
($0.33\to0.19$), making the query channel the slower of the two at
$1.5$B---the source of the overtake in panel~(a).}
\label{fig:joint-sensitivity}
\end{figure*}

\paragraph{Scope and limitations of the empirical study.}
The measurements span four models ($0.5$--$3$B, two families), four
domains on the primary model (Gutenberg, Python, Wikipedia prose, and
structured JSON) with two domains used for cross-model replication,
$n = 100$ prefixes per domain, windows to $w = 8192$, a
positional-encoding ablation, and a five-policy cache-degradation
comparison. The power law and the ordering
$\alpha_{\text{nat}} > \alpha_{\text{code}}$ replicate across all
four models, the positional confound is ruled out
(\S\ref{sec:exp-ablation}), and the polynomial form holds out to
$w = 8192$ without the boundary artifact of shorter-prefix sweeps
(\S\ref{sec:exp-longwindow}). The principal remaining gaps are the
restriction to models $\le 3$B (set by the $16$\,GB single-GPU
budget) and the relatively small code corpus. The specific numerical
values of $\alpha$ may still depend on training-data composition and
tokenizer. Extension to $7$--$8$B models and broader domain coverage
(ArXiv/LaTeX, conversational, multilingual, and larger code corpora)
is left to future work with larger GPUs.

\section{From Bounds to Systems: Predicting Downstream Behaviour}
\label{sec:systems}

Corollary~\ref{cor:topk} predicts that a sink-plus-recent policy with
budget $k$ incurs distortion $D = D^\star + O(k^{-2\alpha})$, and
Theorem~\ref{thm:E2} ties the required window to the
truncation-sensitivity exponent. These are statements about predictive
\emph{distributions}. We now ask whether the same exponent governs
quantities a practitioner observes---next-token behaviour, perplexity,
and long-context task accuracy. All measurements use the
position-preserving protocol of \S\ref{sec:experiments}; the first two
are on the primary model (Qwen2.5-0.5B), the third on Qwen2.5-1.5B.

\paragraph{Next-token agreement tracks the TV exponent.}
The most immediate consequence of a small TV distortion is that the
budgeted cache should usually predict the \emph{same} next token as the
full cache. Figure~\ref{fig:dist} reports the top-$1$ \emph{disagreement}
rate---the fraction of positions at which the $\arg\max$ of the
budget-$k$ distribution differs from the full-cache $\arg\max$---on the
Natural domain. Disagreement falls as a power law in $k$ with exponent
$0.37$ (sliding) and $0.42$ (sink-plus-recent), close to the TV exponent
$\alpha_{\text{nat}} = 0.44$ of Figure~\ref{fig:mixing} (log-RMSE
$\le 0.09$); the exponents sit slightly below $\alpha_{\text{nat}}$
because disagreement is a coarsened, thresholded functional of TV. This
is the truncation-sensitivity law surfacing in a quantity that requires
no access to the reference distribution.

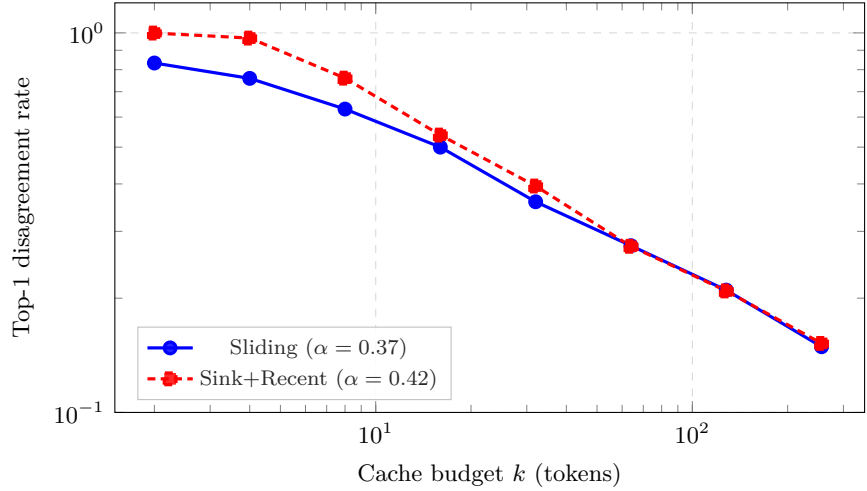
\begin{figure}[t]
\centering
\begin{tikzpicture}
\begin{loglogaxis}[
  width=11.5cm, height=7cm,
  xlabel={Cache budget $k$ (tokens)},
  ylabel={Top-1 disagreement rate},
  xmin=1.5, xmax=350, ymin=0.1, ymax=1.2,
  legend pos=south west,
  legend style={font=\scriptsize, fill opacity=0.85, draw=gray!50},
  grid=major, grid style={dashed,gray!30},
  tick label style={font=\footnotesize}, label style={font=\footnotesize}
]
\addplot[blue, very thick, mark=*, mark size=2.2pt] coordinates {
  (2,0.833) (4,0.759) (8,0.630) (16,0.500)
  (32,0.359) (64,0.275) (128,0.210) (256,0.149)
};
\addlegendentry{Sliding ($\alpha=0.37$)}
\addplot[red, very thick, mark=square*, mark size=2.2pt, densely dashed] coordinates {
  (2,0.999) (4,0.969) (8,0.760) (16,0.538)
  (32,0.395) (64,0.274) (128,0.209) (256,0.152)
};
\addlegendentry{Sink+Recent ($\alpha=0.42$)}
\end{loglogaxis}
\end{tikzpicture}
\caption{Top-$1$ disagreement with the full cache on Qwen2.5-0.5B
(Natural domain): the fraction of positions where the budgeted and full
caches predict a different most-likely token. Disagreement follows a
power law in the budget with exponent close to the TV exponent
$\alpha_{\text{nat}}=0.44$, the predicted downstream signature of
polynomial truncation sensitivity.}
\label{fig:dist}
\end{figure}

\paragraph{Memory--quality tradeoff.}
Figure~\ref{fig:tradeoff} plots perplexity on held-out Natural text
against retained budget for three suffix policies. Both recency-based
policies recover the full-cache perplexity ($13.2$) to within $3\%$ by
$k = 512$, whereas Random-$K$ stays catastrophic ($>3000$ at $k = 512$)
---the distributional gap of \S\ref{sec:exp-policy} now visible in a
task-level metric. The mean NLL increase over the full cache, which in
expectation equals the mean KL distortion, decays as $k^{-1.05}$; this
matches the doubling prediction $2\,\alpha_{\TV} = 1.21$ for the
sink-plus-recent policy, whose own TV exponent ($0.61$) is steeper than
the sliding-window value and is the correct base for the factor of two
(Corollary~\ref{cor:topk}, Lemma~\ref{lem:kl-tv-bounded}). The residual
gap ($1.05$ vs.\ $1.21$) reflects the smallest budgets lying outside the
asymptotic regime and the local---not exact---nature of the quadratic
KL--TV relation.

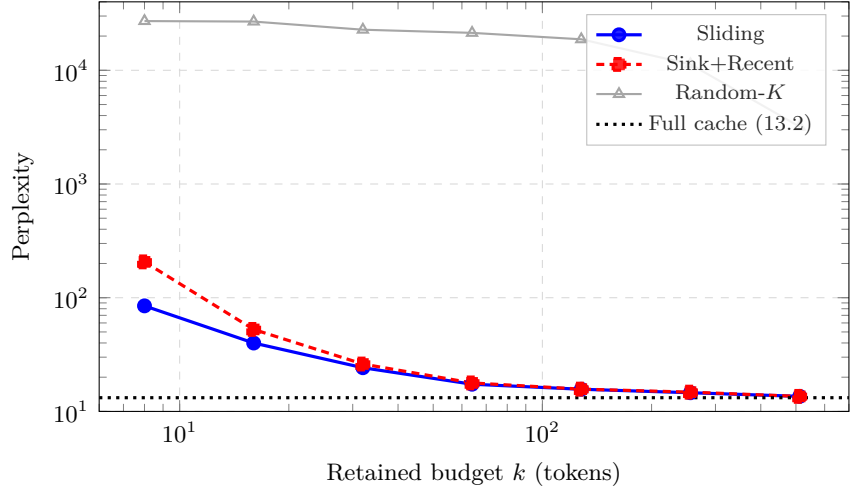
\begin{figure}[t]
\centering
\begin{tikzpicture}
\begin{loglogaxis}[
  width=11.5cm, height=7cm,
  xlabel={Retained budget $k$ (tokens)},
  ylabel={Perplexity},
  xmin=6, xmax=700, ymin=10, ymax=40000,
  legend pos=north east,
  legend style={font=\scriptsize, fill opacity=0.85, draw=gray!50},
  grid=major, grid style={dashed,gray!30},
  tick label style={font=\footnotesize}, label style={font=\footnotesize}
]
\addplot[blue, very thick, mark=*, mark size=2.2pt] coordinates {
  (8,84.80) (16,39.96) (32,24.23) (64,17.35) (128,15.67) (256,14.60) (512,13.57)
};
\addlegendentry{Sliding}
\addplot[red, very thick, mark=square*, mark size=2.2pt, densely dashed] coordinates {
  (8,206.61) (16,52.61) (32,26.11) (64,17.81) (128,15.71) (256,14.78) (512,13.56)
};
\addlegendentry{Sink+Recent}
\addplot[gray!70, thick, mark=triangle, mark size=2.2pt] coordinates {
  (8,27095.6) (16,26826.5) (32,22749.5) (64,21339.9) (128,18734.0) (256,11297.3) (512,3188.7)
};
\addlegendentry{Random-$K$}
\addplot[black, dotted, very thick] coordinates {(6,13.2) (700,13.2)};
\addlegendentry{Full cache ($13.2$)}
\end{loglogaxis}
\end{tikzpicture}
\caption{Perplexity versus retained budget on Qwen2.5-0.5B (Natural
domain). Recency-based policies recover the full-cache perplexity by
$k=512$ while Random-$K$ does not; the mean NLL increase (mean KL)
decays as $k^{-1.05}$, matching the factor-of-two prediction
$2\alpha_{\TV}=1.21$ for sink-plus-recent.}
\label{fig:tradeoff}
\end{figure}

\paragraph{Long-context retrieval.}
Finally we probe a retrieval task in the RULER style
\citep{hsieh2024ruler}: a key--value ``magic number'' (the needle) is
inserted at a controlled depth in a Natural-language haystack of length
$L \in \{512, 1024, 2048, 4096\}$ and the model must return the value,
in a single-needle and a multi-key (distractor) variant. With the full
cache the model is essentially perfect ($\ge 0.98$ at every length),
confirming the task is well posed. Under a budget
(Figure~\ref{fig:ruler}), accuracy is controlled by whether the retained
window \emph{reaches} the needle: it is near zero until the budget
covers the needle's depth, then rises sharply, and longer contexts
demand proportionally larger budgets ($k = 512$ suffices at $L = 512$
but not at $L \ge 2048$). Sliding and sink-plus-recent behave almost
identically, since covering a mid-context needle depends on window
\emph{span} rather than on sinks or recency. This is the task-level
shadow of the window-scaling law (Theorem~\ref{thm:E2})---a suffix
policy must retain enough span to cover the information it needs. We
stress that retrieval accuracy is \emph{coverage}-dominated and hence a
coarser probe than the distributional metrics above; the polynomial
decay law itself is seen most cleanly in
Figures~\ref{fig:dist}--\ref{fig:tradeoff}.

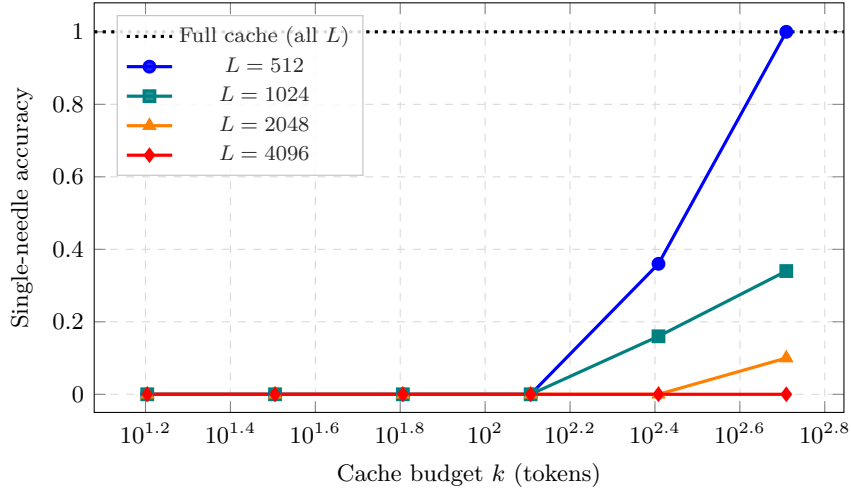
\begin{figure}[t]
\centering
\begin{tikzpicture}
\begin{semilogxaxis}[
  width=11.5cm, height=7cm,
  xlabel={Cache budget $k$ (tokens)},
  ylabel={Single-needle accuracy},
  xmin=12, xmax=700, ymin=-0.05, ymax=1.08,
  legend pos=north west,
  legend style={font=\scriptsize, fill opacity=0.85, draw=gray!50},
  grid=major, grid style={dashed,gray!30},
  tick label style={font=\footnotesize}, label style={font=\footnotesize}
]
\addplot[black, dotted, very thick] coordinates {(12,1.0) (700,1.0)};
\addlegendentry{Full cache (all $L$)}
\addplot[blue, very thick, mark=*, mark size=2pt] coordinates {
  (16,0) (32,0) (64,0) (128,0) (256,0.36) (512,1.0)
};
\addlegendentry{$L=512$}
\addplot[teal, very thick, mark=square*, mark size=2pt] coordinates {
  (16,0) (32,0) (64,0) (128,0) (256,0.16) (512,0.34)
};
\addlegendentry{$L=1024$}
\addplot[orange, very thick, mark=triangle*, mark size=2pt] coordinates {
  (16,0) (32,0) (64,0) (128,0) (256,0) (512,0.10)
};
\addlegendentry{$L=2048$}
\addplot[red, very thick, mark=diamond*, mark size=2pt] coordinates {
  (16,0) (32,0) (64,0) (128,0) (256,0) (512,0)
};
\addlegendentry{$L=4096$}
\end{semilogxaxis}
\end{tikzpicture}
\caption{RULER-style single-needle retrieval on Qwen2.5-1.5B (sliding
policy). The full cache solves the task at every length (dotted); under
a budget, accuracy rises only once the retained window reaches the
needle's depth, and longer contexts require proportionally larger
budgets---the task-level signature of the window-scaling law.}
\label{fig:ruler}
\end{figure}

In sum, the truncation-sensitivity exponent is not merely a property of
predictive distributions: the same $\alpha$ governs next-token agreement
and, through the factor of two, perplexity degradation, while the
window-scaling characterisation predicts the budget a retrieval task
requires.

\section{Connections to Deployed Compression Schemes}
\label{sec:discussion}

We sketch how the framework relates to deployed schemes. Quantitative
values for windows $w \gg 256$ are \emph{extrapolations} of the
$w \le 8192$ measurements on Qwen2.5-0.5B and should be read as
predictions to be tested, not established claims; in particular we do
not assert specific distortion values for Mistral-class models, which
we did not measure.

\paragraph{Sliding-window schemes (StreamingLLM, Mistral).}
A sliding window plus a small initial-token prefix
\citep{xiao2024efficient} is the direct instantiation of
Theorem~\ref{thm:E2}; the attention sinks compensate for the initial
$O(1)$ tokens where the polynomial-sensitivity assumption is expected
to fail. Mistral's fixed $w = 4096$ \citep{jiang2023mistral7b} and
StreamingLLM's $w + k_{\mathrm{sink}}$ configurations are
\emph{qualitatively} consistent with polynomial rather than geometric
sensitivity: geometric mixing would predict $\rho^{w} \approx 0$ and
hence negligible distortion at such windows, whereas the reported
minor degradation versus full attention matches a slow polynomial
decay. The precise distortion at these windows on those models
remains to be measured.

\paragraph{Latent, quantized, and selection-based schemes.}
DeepSeek-V2's latent dimension $d_c = 512$ across $60$ layers
\citep{deepseek2024mla} is consistent with the intrinsic-dimension
bound (Conjecture~\ref{thm:intrinsic-dim}) and the multi-layer
allocation (Proposition~\ref{thm:multi-WZ}); this bound is
mixing-independent. INT2 mixed-precision quantization
\citep{liu2024kivi,kang2024gear} corresponds to the Bennett--Gersho
overhead $d_c \log(1/\Delta)$ in the same theorem. Prompt-aware
(SnapKV \citep{li2024snapkv}, Quest \citep{tang2024quest}) and
layer-adaptive (PyramidKV \citep{cai2024pyramidkv}, Ada-KV
\citep{feng2024adakv}) schemes implicitly exploit the truncation
sensitivity established here; a formal account of their gains over a
naive sliding window requires a selection-aware refinement of
Corollary~\ref{cor:topk}, which we leave to future work.

\section{Limitations}
\label{sec:limitations}

\paragraph{Empirical scope.}
Our direct measurements use four models (Qwen2.5-0.5B/1.5B/3B,
SmolLM2-1.7B), four domains on the primary model (two for cross-model
replication), $n = 100$ prefixes per domain, and window ranges up to
$w = 8192$. The power law and the ordering
$\alpha_{\text{nat}} > \alpha_{\text{code}}$ replicate across all
four models. A $4$-bit extension to Qwen2.5-7B (natural-language
domain) confirms the phenomenon persists at $7$B
($\alpha_{\text{nat}} \approx 0.57$); the FP16 measurements are
otherwise $\le 3$B parameters (the limit of a $16$\,GB single-GPU
budget), and the code corpus is relatively small. The
numerical values of $\alpha$ may still depend on training-data
composition and tokenizer. Full-precision replication on $7$--$8$B Llama- and
Mistral-class models and expanded domain coverage (ArXiv/LaTeX,
conversational, multilingual, and larger code corpora) remain
valuable and are left to future work with larger accelerators.

\paragraph{Positional-encoding treatment (addressed).}
Truncating $X_{1:t-1}$ to $X_{t-w:t-1}$ and re-feeding it shifts the
position indices of retained tokens under rotary or learned
encodings, which could inflate the measured TV. We use a
position-preserving protocol throughout (\S\ref{sec:exp-ablation},
retained tokens keep their original indices); it agrees with the
naive fresh-sequence protocol to three decimals
(Figure~\ref{fig:ablation}), so the decay is genuine content
forgetting rather than a positional artifact. We verify this on the
primary model only.

\paragraph{Measurement protocol vs.\ true cache eviction.}
Our truncation-sensitivity measurements re-feed the length-$w$ suffix
as a fresh, position-preserving sequence and read the model's
next-token distribution. This is \emph{not} identical to true KV-cache
eviction, in which the retained keys and values were computed while
the model still attended to the full past and are merely dropped at
read-out; the re-fed suffix never saw the evicted context. The two
coincide only insofar as attention to evicted positions is precisely
what truncation removes. Quantifying the gap---re-fed truncated
context versus a full forward pass with subsequent K/V deletion---is
left to future work, and our distortion statements should be read
against the re-fed protocol.

\paragraph{Rate-convergence gap (achievability is conjectural).}
The memory (window) scaling is tight: $\Theta(\varepsilon^{-1/\alpha})$,
matching the distortion upper bound (Theorem~\ref{thm:E2}) and the
suffix-only lower bound (Theorem~\ref{thm:window-lb}). At the level of
coding rate, the rigorous statement is the delayed operational
comparison of Appendix~\ref{app:operational}; an online single-letter
scheme attaining the converse rate exponent is only \emph{conjectured}
(Conjectures~\ref{thm:achievability},~\ref{thm:achievability-tight}),
because the supporting covering/auxiliary-mixing argument is not
established (Remark~\ref{rmk:exponent-gap}). We regard establishing an
online achievability---or discharging the information-spectrum
condition of Appendix~\ref{app:operational}---as the main theoretical
gap.

\paragraph{Dependent-source assumptions, isolated.}
The achievability analysis adapts finite-blocklength source-coding
machinery (developed for independent sources) to the dependent,
window-restricted source. Rather than assert the classical results
transfer automatically, we isolate the three places where dependence
enters as named assumptions---additive tilted-information
(Assumption~\ref{ass:additive-tilted}), finite-memory covering
regularity (Assumption~\ref{ass:covering-reg}), and forward
covariance-decay $(\dagger)$. Each holds for independent components
and is plausibly establishable for finite-memory sources
\citep{douc2018markov,tasci2024dispersion}, which we do not carry
out; we regard discharging them as the main theoretical gap.

\paragraph{Tail behavior at large $w$ (characterized to $8$k).}
Using a long ($16{,}500$-token) prefix to avoid boundary effects, we
measure the decay out to $w = 8192$ (\S\ref{sec:exp-longwindow}). The
form remains a single power law across more than five octaves, with a
long-range exponent close to---if anything slightly below---the
short-window estimate; the decay does \emph{not} accelerate at long
range. On the primary model the polynomial form is thus established
well past the deployed
$4$k regime. Behavior at $w \gg 8192$ (where prefixes must be longer
still) remains to be measured, and the Code measurement is noisier
than Natural owing to the smaller corpus.

\paragraph{Architecture assumption.}
The framework assumes pre-LayerNorm Transformers with standard
attention and feed-forward layers. State-space models such as
Mamba and RWKV violate the attention-specific structure used in
Lipschitz error propagation (Lemma~\ref{lem:prop}), and the
truncation sensitivity may take a different functional form; we do
not analyze these architectures.

\paragraph{Heavy-hitter eviction.}
Corollary~\ref{cor:topk} applies to sink-plus-recent eviction,
where retained tokens are determined by position. Heavy-hitter
schemes such as H2O \citep{zhang2023h2o} and Scissorhands
\citep{liu2023scissorhands} select tokens by attention score; the
connection to our framework requires an additional empirical
property of trained models (high-attention tokens approximately
coincide with sinks plus recent positions), which we do not
establish.

\paragraph{Mechanistic explanation of $\alpha$.}
We measure $\alpha$ empirically but do not derive it from
architectural properties. A theoretical prediction of $\alpha$ in
terms of model parameters (depth, width, attention head structure)
and training-data statistics would be a natural next step and is
presently open.

\section{Conclusion}
\label{sec:conclusion}

We developed an information-theoretic framework for KV cache
compression based on a sequential Wyner--Ziv formulation. The
central empirical observation is that the trained language model's
sensitivity to context truncation follows a \emph{polynomial}
rather than \emph{geometric} decay---a property of the trained
model under the evaluation distribution, distinct from a mixing
condition on the underlying data process. The central theoretical
contribution is a characterization of the memory requirement
\emph{of suffix-only cache policies}: a sliding-window scheme
achieves distortion $\varepsilon$ with window
$O(\varepsilon^{-1/\alpha})$, and---under an additional two-sided
Bayes-risk condition (Definition~\ref{def:two-sided})---a converse
shows $\Omega(\varepsilon^{-1/\alpha})$ is necessary within this
policy class (Theorems~\ref{thm:E2}, \ref{thm:window-lb},
Corollary~\ref{cor:tight}), pinning the per-token memory of
suffix-only schemes at $\Theta(\varepsilon^{-1/\alpha})$. Whether a
more general propagating or recurrent cache summary can beat this
scaling is an open problem. At the level of coding rate, a \emph{delayed} multi-letter
operational Wyner--Ziv comparison (Appendix~\ref{app:operational})
shows a suffix-only block code matches the full-context rate up to a
single side-information continuity penalty and a redundancy that
vanishes with block length, under quantized query side information and
explicit information-spectrum and de-binning-stability assumptions; we do \emph{not}
establish an online single-letter scheme attaining the converse
exponent.

Beyond these rigorous results, several extensions are stated only as
conjectures: an information-rate optimal-window-scaling argument
(Conjecture~\ref{thm:achievability}), a window-adaptive universal
scheme, and a continuous-latent dimension bound. A multi-layer
Lagrangian rate allocation and the sink-plus-recent corollary connect
the framework to deployed schemes. The factor of $2$ in the
sink-plus-recent KL exponent $O(k^{-2\alpha})$ arises from a
quadratic local relation between $\KLs$ and $\TV$ under bounded
log-density (Lemma~\ref{lem:kl-tv-bounded}), \emph{not} from
Pinsker's inequality, which provides only the converse direction.

Two independent measurements support the framework: the TV-decay
exponent of context truncation and the KL-decay exponent of the
sink-plus-recent scheme are consistent with the factor-of-two
predicted by the local quadratic KL--TV approximation, with
measured ratios $2.38$ (Natural) and $1.93$ (Code) bracketing the
predicted value. Confirmed across four models, four domains, windows
to $w = 8192$, and a positional-encoding-free protocol
(\S\ref{sec:experiments}), the polynomial law also predicts the
degradation curves of concrete cache policies---recency-based
eviction suppresses distortion by roughly two orders of magnitude
over random retention at equal budget. Extension to larger models
and broader domains is the natural next step.

\section*{Impact Statement}

This paper studies theoretical limits and empirical properties of
KV-cache compression in autoregressive language models. Better
understanding of these limits may enable inference systems that
reduce memory usage, cost, and energy consumption for long-context
models. The same gains may also reduce the deployment barrier for
large-scale long-context inference, potentially amplifying
downstream misuse risks such as surveillance and high-volume
automated persuasion. This work does not introduce a new dataset, a
deployable product, or human-subject experiments. Downstream
practitioners should evaluate privacy, copyright, and misuse risks
in their specific deployment contexts.

\bibliographystyle{icml2026}
\bibliography{references}
\appendix

\section*{Note on this Appendix}

This appendix proves the main results under two assumptions:
the polynomial truncation-sensitivity condition of
Definition~\ref{def:poly-mixing} (the form supported by
Section~\ref{sec:experiments}) and the classical geometric
$\rho$-mixing condition (Definition~\ref{def:mixing} below). The
geometric version is presented first as the cleaner template;
Section~\ref{app:polynomial-extension} extends each lemma to the
polynomial setting and proves Theorem~\ref{thm:E2-poly} (the
polynomial version of Theorem~\ref{thm:E2}). The achievability
conjecture (Conjecture~\ref{thm:achievability}) is argued
heuristically in Appendix~\ref{app:achievability}, the window lower
bound (Theorem~\ref{thm:window-lb}) in Appendix~\ref{app:window-lb}, and
the universal scheme (Conjecture~\ref{thm:universal}) in
Appendix~\ref{app:universal}; together the window upper and lower
bounds establish the tight memory characterization of
Corollary~\ref{cor:tight}. Theorem~\ref{thm:asymptotic}, Proposition~\ref{thm:finite-sample},
Conjecture~\ref{thm:intrinsic-dim}, and Proposition~\ref{thm:multi-WZ} are
mixing-independent and the proofs apply under either definition.
The KL--TV conversion underpinning Corollary~\ref{cor:topk} is
established as Lemma~\ref{lem:kl-tv-bounded}
(Appendix~\ref{app:kltv}).

\begin{definition}[Geometric $\rho$-mixing, used in
proofs of Theorem~\ref{thm:E2} and Conjecture~\ref{cor:mixing-conc}
below]\label{def:mixing}
The language model is \emph{$\rho$-mixing} with constants
$C_{\mathrm{mix}} > 0$ and $\rho \in [0, 1)$ if
$|\,p_\theta(X_t = x \mid X_{1:t-1}) - p_\theta(X_t = x \mid X_{t-w:t-1})\,|
\leq C_{\mathrm{mix}}\rho^w$ for all $t, x, w$ almost surely. This
is the stronger geometric counterpart of the polynomial
sensitivity in Definition~\ref{def:poly-mixing}.
\end{definition}

\section{Robustness of the empirical measurements}
\label{app:robustness}
This appendix details the robustness checks summarized in
Section~\ref{sec:experiments}: cross-domain exponents
(\S\ref{app:crossdomain}), a positional-encoding ablation
(\S\ref{app:exp-ablation}), the long-window tail
(\S\ref{app:exp-longwindow}), and the full cache-policy degradation
curves (\S\ref{app:exp-policy}). It also reports the raw
median KL values behind the sink-plus-recent versus Random-$K$
comparison (Table~\ref{tab:randomk}).

\begin{table}[ht]
\centering
\small
\setlength{\tabcolsep}{5pt}
\begin{tabular}{lcccc}
\toprule
& \multicolumn{2}{c}{Sink+Recent} & \multicolumn{2}{c}{Random-$K$} \\
\cmidrule(lr){2-3}\cmidrule(lr){4-5}
$k$ & Natural & Code & Natural & Code \\
\midrule
$8$   & $3.0$  & $4.0$  & $5.6$ & $9.6$ \\
$64$  & $0.30$ & $0.65$ & $4.3$ & $5.1$ \\
$512$ & $0.040$ & $0.14$ & $3.4$ & $4.1$ \\
\bottomrule
\end{tabular}
\caption{Median KL distortion (nats) for sink-plus-recent versus
Random-$K$ at representative budgets $k$ (Qwen2.5-0.5B,
position-preserving). Sink-plus-recent decays as a power law in $k$
(Figure~\ref{fig:topk}); Random-$K$ is essentially flat. The
$85\times$/$30\times$ figures quoted in the main text are the
Natural/Code ratios at $k = 512$.}
\label{tab:randomk}
\end{table}

\begin{figure}[ht]
\centering
\begin{tikzpicture}
\begin{loglogaxis}[
  width=11.5cm, height=7cm,
  xlabel={Sink-plus-recent budget $k$ (tokens)},
  ylabel={Median KL distortion},
  xmin=6, xmax=600,
  ymin=0.003, ymax=10,
  legend pos=south west,
  legend style={font=\scriptsize, fill opacity=0.85, draw=gray!50},
  grid=major, grid style={dashed,gray!30},
  tick label style={font=\footnotesize},
  label style={font=\footnotesize}
]
\addplot[blue, very thick, mark=*, mark size=2.2pt] coordinates {
  (8, 3.0) (16, 1.6) (32, 0.75) (64, 0.30)
  (128, 0.13) (256, 0.08) (512, 0.04)
};
\addlegendentry{Sink+Recent Natural ($\alpha_{\KLs}=1.04$)}
\addplot[red, very thick, mark=square*, mark size=2.2pt, densely dashed] coordinates {
  (8, 4.0) (16, 2.1) (32, 1.05) (64, 0.65)
  (128, 0.45) (256, 0.28) (512, 0.14)
};
\addlegendentry{Sink+Recent Code ($\alpha_{\KLs}=0.74$)}
\addplot[gray!70, thick, mark=triangle, mark size=2.2pt] coordinates {
  (8, 5.587) (16, 5.007) (32, 4.224) (64, 4.330)
  (128, 3.981) (256, 4.336) (512, 3.365)
};
\addlegendentry{Random-$K$ (Natural)}
\end{loglogaxis}
\end{tikzpicture}
\caption{Measured sink-plus-recent KL decay on Qwen2.5-0.5B. The
scheme follows the power law $D \propto k^{-2\alpha}$ where $\alpha$
is the TV-decay exponent of Figure~\ref{fig:mixing}. The measured
KL exponents ($1.04$ Natural, $0.74$ Code) are $2.38\times$ and
$1.93\times$ the corresponding TV exponents ($0.44$, $0.38$),
bracketing the factor-of-two prediction from the
$\KLs \leq C(\epsilon_{\min}) \cdot \TV^2$ quadratic local relation
under bounded log-density (Lemma~\ref{lem:kl-tv-bounded},
Corollary~\ref{cor:topk}). These figures compare sink-plus-recent
$\KLs$ against sliding-window $\TV$; estimating both under the same
policy yields a tighter ratio $\alpha_{\KLs}/\alpha_{\TV} \in [1.93,
2.04]$ (both policies, cutoffs $k \ge 16$). Random-$K$ baseline shows no decay,
isolating the contribution of recency structure. At $k = 512$,
sink-plus-recent attains roughly $85\times$ (Natural) and
$30\times$ (Code) lower KL distortion than Random-$K$ (a ratio of
median KL values; see Table~\ref{tab:randomk} for raw numbers).}
\label{fig:topk}
\end{figure}
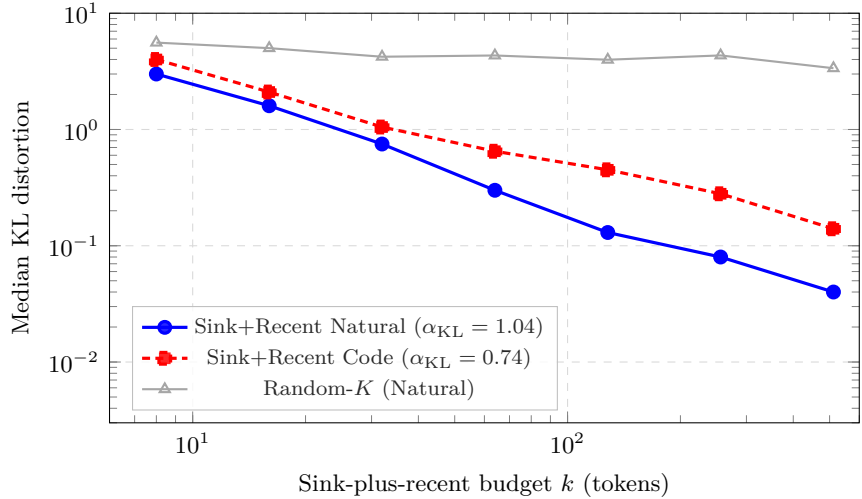

\subsection{Cache-policy degradation}
\label{app:exp-policy}
Figure~\ref{fig:policy} shows the median
$\KLs(\text{full}\,\|\,\text{policy})$ on the Natural domain for the
five policies of \S\ref{sec:exp-policy}. Recency-based policies
(sliding, sink-plus-recent) suppress distortion by roughly two orders
of magnitude relative to Random-$K$ at equal budget; the
attention-score heavy-hitter sits in between; and every policy decays
as a power law in $k$, the predicted consequence of polynomial
truncation sensitivity. The Code domain is qualitatively identical
(sliding/sink-plus-recent KL $\approx 0.003$ at $k = 512$ versus
$2.1$ for Random-$K$), and the NLL-increase panels track the KL
panels.

\begin{figure}[ht]
\centering
\begin{tikzpicture}
\begin{loglogaxis}[
  width=11.5cm, height=7cm,
  xlabel={Budget $k$ (tokens)},
  ylabel={median $\KLs(\text{full}\,\|\,\text{policy})$},
  xmin=28, xmax=560,
  ymin=0.007, ymax=8,
  legend pos=south west,
  legend style={font=\scriptsize, fill opacity=0.85, draw=gray!50},
  grid=major, grid style={dashed,gray!30},
  tick label style={font=\footnotesize},
  label style={font=\footnotesize}
]
\addplot[red, very thick, mark=*, mark size=2pt] coordinates {
  (32,5.30) (64,4.98) (128,3.85) (256,2.84) (512,1.90)
};
\addlegendentry{Random-$K$}
\addplot[violet, very thick, mark=triangle*, mark size=2.4pt] coordinates {
  (32,1.89) (64,1.27) (128,0.72) (256,0.41) (512,0.185)
};
\addlegendentry{Heavy-hitter}
\addplot[orange, very thick, mark=square*, mark size=2pt] coordinates {
  (32,0.324) (64,0.143) (128,0.070) (256,0.033) (512,0.0115)
};
\addlegendentry{Sliding}
\addplot[teal, very thick, mark=diamond*, mark size=2.4pt, densely dashed] coordinates {
  (32,0.433) (64,0.141) (128,0.075) (256,0.032) (512,0.0105)
};
\addlegendentry{Sink+recent}
\end{loglogaxis}
\end{tikzpicture}
\caption{Cache-policy degradation on Qwen2.5-0.5B (Natural domain,
position-preserving, $100$ prefixes). Recency-based policies (sliding,
sink-plus-recent) suppress distortion by $\sim$$100\times$ relative to
Random-$K$ at equal budget; the lightweight last-query attention-score
heavy-hitter sits in
between. Every policy decays as a power law in $k$, as predicted by
polynomial truncation sensitivity. Sliding and sink-plus-recent
coincide at this context length.}
\label{fig:policy}
\end{figure}
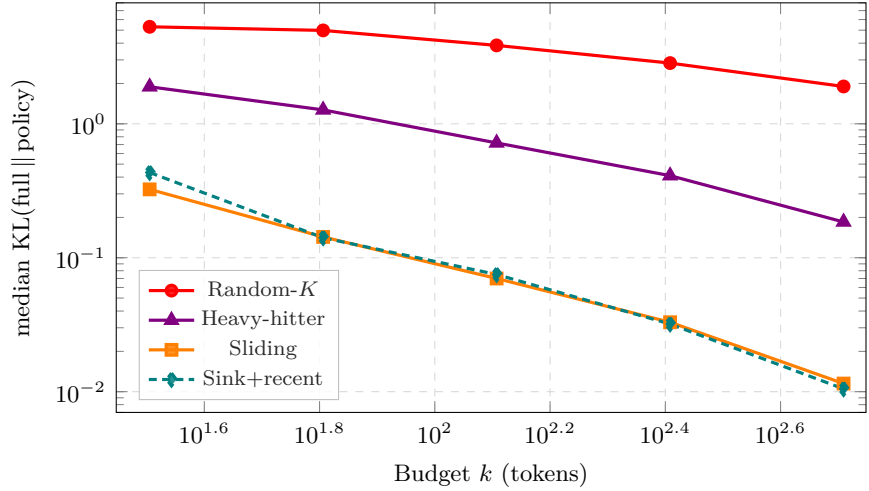

\subsection{Cross-domain exponents}
\label{app:crossdomain}
Table~\ref{tab:crossdomain} reports the TV-decay exponent on four
domains for the primary model. A power law is favored over an
exponential by $2$--$4\times$ in log-RMSE in every domain, while the
exponent varies only modestly ($\alpha \in [0.39, 0.44]$): the
polynomial form is domain-robust, with $\alpha$ weakly
domain-dependent at this scale.

\begin{table}[ht]
\centering
\small
\setlength{\tabcolsep}{5pt}
\begin{tabular}{lccc}
\toprule
Domain & $\alpha_{\TV}$ & PL log-RMSE & exp.\ log-RMSE \\
\midrule
Natural (Gutenberg) & $0.438$ & $0.14$ & $0.32$ \\
Code (Python)       & $0.389$ & $0.08$ & $0.31$ \\
Wikipedia prose     & $0.392$ & $0.08$ & $0.30$ \\
JSON (structured)   & $0.426$ & $0.17$ & $0.26$ \\
\bottomrule
\end{tabular}
\caption{Cross-domain TV-decay exponents on Qwen2.5-0.5B
(position-preserving, $100$ prefixes per domain). A power law is
favored over an exponential by $2$--$4\times$ in log-RMSE in every
domain; the exponent varies only modestly, so the polynomial form is
domain-robust while $\alpha$ is weakly domain-dependent.}
\label{tab:crossdomain}
\end{table}

\subsection{Positional-encoding ablation}
\label{app:exp-ablation}
\label{sec:exp-ablation}
A natural concern is that the measured TV decay reflects a
positional-encoding artifact rather than genuine content forgetting:
when a truncated window is re-fed as a fresh sequence, its tokens
receive new position indices, so part of the measured change could
be attributed to position shift. We rule this out by comparing two
protocols on identical prefixes and windows. In the \emph{fresh}
protocol the retained window is re-indexed from $0$; in the
\emph{position-preserving} protocol it keeps its original absolute
indices. The two are indistinguishable: on the primary model the
fitted exponents agree to three decimals
($\alpha_{\text{nat}} = 0.438$ fresh vs.\ $0.437$ preserved;
$\alpha_{\text{code}} = 0.383$ for both), with overlapping bootstrap
intervals (Figure~\ref{fig:ablation}). The measured decay is
therefore a property of the conditional distribution under content
truncation, not of the positional encoding. All measurements in
Section~\ref{sec:experiments} use the position-preserving protocol.

\begin{figure}[ht]
\centering
\begin{tikzpicture}
\begin{loglogaxis}[
  width=11.5cm, height=7cm,
  xlabel={Window size $w$ (tokens)},
  ylabel={$\widehat{\TV}_w$},
  xmin=1.5, xmax=300,
  ymin=0.09, ymax=1.0,
  legend pos=south west,
  legend style={font=\scriptsize, fill opacity=0.85, draw=gray!50},
  grid=major, grid style={dashed,gray!30},
  tick label style={font=\footnotesize},
  label style={font=\footnotesize}
]
\addplot[blue, very thick, mark=*, mark size=2.4pt] coordinates {
  (2,0.7934) (4,0.6816) (8,0.5794) (16,0.4728) (32,0.3398) (64,0.2155) (128,0.1307) (256,0.1086)
};
\addlegendentry{Natural, fresh ($\alpha=0.438$)}
\addplot[cyan!70!black, only marks, mark=o, mark size=3pt] coordinates {
  (2,0.7933) (4,0.6819) (8,0.5794) (16,0.4730) (32,0.3399) (64,0.2158) (128,0.1307) (256,0.1087)
};
\addlegendentry{Natural, position-preserved ($\alpha=0.437$)}
\end{loglogaxis}
\end{tikzpicture}
\caption{Positional-encoding ablation on Qwen2.5-0.5B (Natural
domain shown; Code is analogous). The \emph{fresh} protocol
(re-indexed window) and the \emph{position-preserving} protocol
(original absolute indices) yield indistinguishable curves; fitted
exponents agree to three decimals ($0.438$ vs.\ $0.437$). The
measured truncation sensitivity is therefore genuine content
forgetting, not a positional-encoding artifact.}
\label{fig:ablation}
\end{figure}
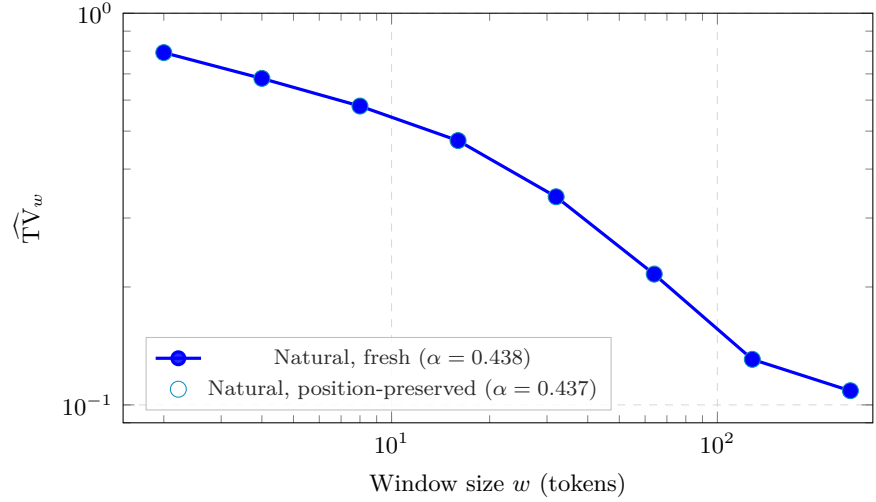

\subsection{Long-window tail and fit competition}
\label{app:exp-longwindow}
\label{sec:exp-longwindow}
The short-window grid leaves open whether the power law persists to
the $w \sim 10^3$--$10^4$ windows used by deployed schemes. Probing
this range requires care: total-variation under window truncation is
measured by comparing the full-context conditional to the conditional
on the last $w$ tokens, so if $w$ approaches the prefix length almost
nothing is truncated and the measured TV collapses toward zero---a
boundary artifact rather than a property of the model. (An earlier
sweep with prefixes of length $4200$ exhibited exactly this: the
$w = 4096$ point fell far below the trend because only $104$ tokens
were truncated.) We avoid the artifact by using a \emph{long prefix}
of $16{,}500$ tokens and sweeping $w \in \{256, \ldots, 8192\}$, so
that even the largest window truncates more than $8000$ tokens (the
window is at most $\approx 50\%$ of the prefix). We fit a single power
law on the resulting curve (Figure~\ref{fig:longwindow}). The decay
remains polynomial across the entire range: on Natural text the
fitted exponent is $\alpha = 0.36$ with log-RMSE $0.06$, close to the
short-window estimate of $0.44$ (\S\ref{app:exp-ablation}), and the
$w = 8192$ point sits on the trend rather than collapsing. The power
law thus extends cleanly more than five octaves, well past the
deployed $4$k regime, with the long-range exponent if anything
slightly \emph{smaller} than the short-range one (the decay does not
accelerate at long range). The Code curve is noisier
(log-RMSE $0.13$, fitted $\alpha = 0.51$), reflecting the smaller code
corpus and its repetitive structure, but follows the same polynomial
form. None of this affects the qualitative predictions (P1)--(P2),
which concern the polynomial (non-geometric) character of the decay.

\begin{figure}[ht]
\centering
\begin{tikzpicture}
\begin{loglogaxis}[
  width=11.5cm, height=7cm,
  xlabel={Window size $w$ (tokens)},
  ylabel={$\widehat{\TV}_w$},
  xmin=200, xmax=10000,
  ymin=0.02, ymax=0.25,
  legend pos=south west,
  legend style={font=\scriptsize, fill opacity=0.85, draw=gray!50},
  grid=major, grid style={dashed,gray!30},
  tick label style={font=\footnotesize},
  label style={font=\footnotesize}
]
\addplot[blue, very thick, mark=*, mark size=2pt] coordinates {
  (256,0.1657) (512,0.1357) (1024,0.1072) (2048,0.0835) (4096,0.0692) (8192,0.0450)
};
\addlegendentry{Natural (measured)}
\addplot[blue!55, thick, domain=256:8192, samples=80] {1.30/(x^0.362)};
\addlegendentry{Power law $\alpha=0.36$}
\addplot[red, very thick, mark=square*, mark size=2pt, densely dashed] coordinates {
  (256,0.1933) (512,0.1533) (1024,0.0911) (2048,0.0772) (4096,0.0622) (8192,0.0292)
};
\addlegendentry{Code (measured)}
\addplot[red!55, thick, domain=256:8192, samples=80] {3.0/(x^0.508)};
\addlegendentry{Power law $\alpha=0.51$}
\end{loglogaxis}
\end{tikzpicture}
\caption{Long-window TV decay to $w = 8192$ (Qwen2.5-0.5B,
position-preserving, prefix length $16{,}500$). A single power law
describes the data across more than five octaves; the $w = 8192$
point sits on the trend rather than collapsing, because the long
prefix keeps the window below $\approx 50\%$ of the context (avoiding
the boundary artifact that affected shorter-prefix sweeps). The
Natural exponent ($\alpha = 0.36$) is close to the short-window
estimate ($0.44$); the decay does not accelerate at long range. The
Code curve is noisier (smaller corpus) but follows the same
polynomial form.}
\label{fig:longwindow}
\end{figure}
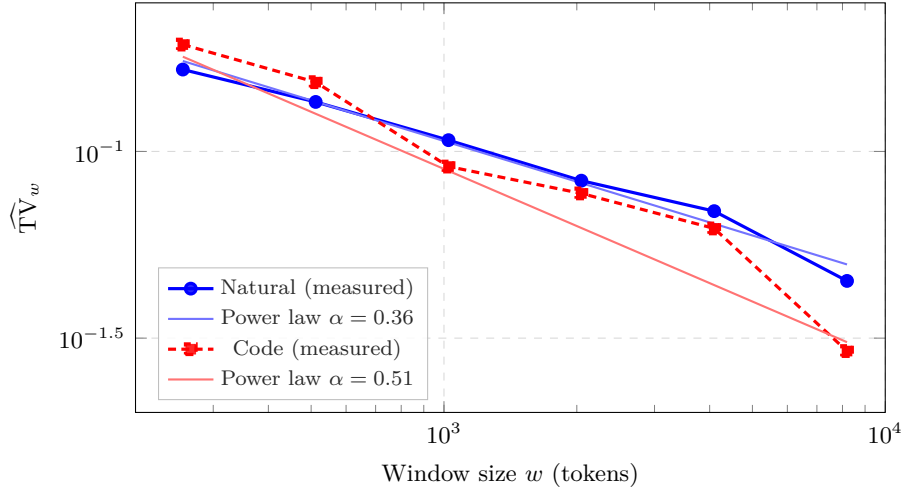

\subsection{KL--TV quadratic ablation}
\label{app:kltv-ablation}
The sink-plus-recent KL exponent is roughly twice the TV exponent
(\S\ref{sec:experiments}), which we attribute to a local quadratic
relation $\KLs \le C(\epsilon_{\min})\,\TV^2$ valid in a bounded-floor
regime rather than to Pinsker's inequality (which gives only the
converse direction). We test the stability of this doubling under
explicit smoothing $\tilde p = (1-\mu)p + \mu/V$, sweeping
$\mu \in \{10^{-4}, 10^{-2}, 0.1, 0.3\}$ and re-measuring TV and KL of
the sink-plus-recent reconstruction as functions of budget $k$
(Figure~\ref{fig:kltv}). Two findings support the bounded-floor
interpretation. First, plotting KL against $\TV^2$ yields a clean
through-the-origin line at every smoothing level, confirming the
locally quadratic relation. Second, the fitted exponent ratio
$\alpha_{\KLs}/\alpha_{\TV}$ stays at $2.00$ across the entire range,
including the heaviest smoothing $\mu = 0.3$: the measured regime
remains within the bounded-floor regime throughout, so the doubling is
robust rather than an artifact of the unsmoothed tail. (The
budget-$k$ TV exponent here, $\alpha_{\TV} \approx 0.67$, is the decay
in the eviction budget and differs from the window-truncation
exponent $\alpha \approx 0.44$ of \S\ref{sec:experiments}; only the
\emph{ratio} is the object of interest.)

\begin{figure}[ht]
\centering
\begin{tikzpicture}
\begin{axis}[
  width=11.5cm, height=7cm,
  xlabel={$\TV^2$},
  ylabel={$\KLs$},
  xmin=0, xmax=0.13,
  ymin=0, ymax=0.47,
  legend pos=north west,
  legend style={font=\scriptsize, fill opacity=0.85, draw=gray!50},
  grid=major, grid style={dashed,gray!30},
  tick label style={font=\footnotesize},
  label style={font=\footnotesize}
]
\addplot[blue, very thick, mark=*, mark size=2pt] coordinates {
  (0.00219,0.00870) (0.00862,0.0336) (0.01443,0.0560) (0.03387,0.1417) (0.11488,0.4401)
};
\addlegendentry{$\mu=10^{-4}$}
\addplot[red, very thick, mark=square*, mark size=2pt, densely dashed] coordinates {
  (0.00108,0.00590) (0.00422,0.0234) (0.00707,0.0385) (0.01661,0.0969) (0.05630,0.3032)
};
\addlegendentry{$\mu=0.3$}
\end{axis}
\end{tikzpicture}
\caption{KL versus $\TV^2$ for sink-plus-recent on Qwen2.5-0.5B
(Natural), at the lightest and heaviest smoothing levels. Both lie on
a through-the-origin line, confirming the locally quadratic KL--TV
relation; the fitted exponent ratio $\alpha_{\KLs}/\alpha_{\TV}$
equals $2.00$ at every $\mu \in \{10^{-4}, 10^{-2}, 0.1, 0.3\}$, so the
exponent doubling is robust across smoothing.}
\label{fig:kltv}
\end{figure}
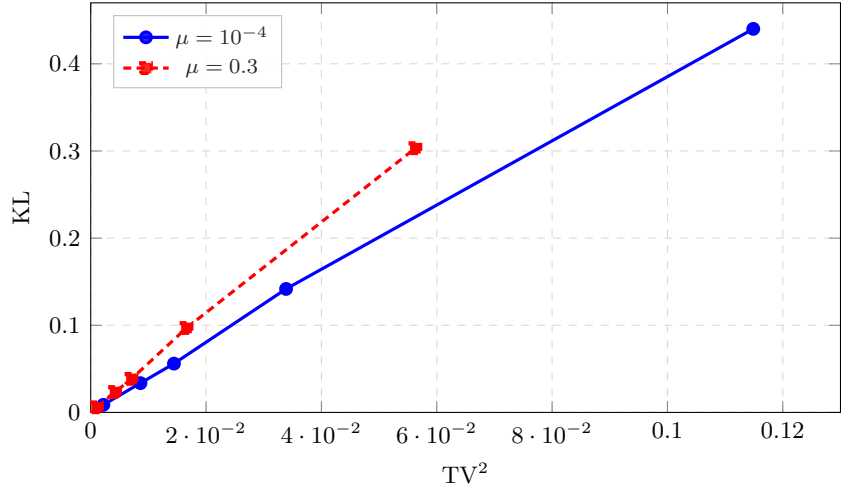

\section{Notation}\label{app:notation}

We collect the principal notation, organized by category.

\begin{table*}[ht]
\centering
\renewcommand{\arraystretch}{1.15}
\small
\begin{tabular}{ll}
\toprule
\textbf{Symbol} & \textbf{Meaning} \\
\midrule
\multicolumn{2}{l}{\emph{Probability and information theory}} \\
$\Omega, \F, \Prob$ & Underlying probability space \\
$\F_t = \sigma(X_1, \ldots, X_t)$ & Natural filtration of token process \\
$\Fhat_t = \sigma(\Chat_1, \ldots, \Chat_t)$ & Compressed filtration \\
$H(\cdot), h(\cdot)$ & Entropy, differential entropy \\
$I(\cdot;\cdot)$ & Mutual information (discrete or mixed) \\
$\KL{\cdot}{\cdot}$ & Kullback--Leibler divergence \\
$\mathrm{TV}(\cdot, \cdot)$ & Total variation \\
\midrule
\multicolumn{2}{l}{\emph{Sequential setup}} \\
$X_t \in \V$ & Token at time $t$, vocabulary size $V$ \\
$Q_t \in \R^k$ & Query at time $t$ (decoder side info) \\
$Z_t \in \R^V$ & Next-token logits \\
$p_t = \mathrm{softmax}(Z_t)$ & Next-token distribution \\
$\phi_t : \V^{t-1} \to \mathcal{C}_t$ & Encoder, $\Chat_t = \phi_t(X_{1:t-1})$ \\
$\psi_t$ & Decoder producing reconstruction $\tilde p_t$ \\
$R_n, D_n$ & Average rate and distortion at horizon $n$ \\
$R^\star(D)$ & Sequential Wyner--Ziv rate function \\
$U_t \in \mathcal{U}_t$ & Auxiliary process \\
$\mathcal{R}_t^{\mathrm{WZ}}$ & Per-step pointwise WZ rate \\
\midrule
\multicolumn{2}{l}{\emph{LLM-specific}} \\
$f$ & Layer logit function: $Z_t = f(h_t^{(L-1)})$ \\
$g^{(\ell)}$ & Layer-$\ell$ post-attention transformation \\
$L_g^{(\ell)}$ & Layer-$\ell$ full-block Lipschitz constant ($\approx 1.05$) \\
$L_b^{(\ell)}$ & Layer-$\ell$ non-residual branch Lipschitz ($\approx 0.1$ with LN) \\
$L_{\mathrm{LN}}$ & LayerNorm Lipschitz constant \\
$\tilde L_g^{(\ell)} = \sqrt{1+(L_b L_{\mathrm{LN}})^2}$ & Effective Lipschitz (pre-LN with skip, $\approx 1.005$) \\
$\rho, \rho^{(\ell)}, \rho_{\mathrm{eff}}$ & Mixing rate (full / layer / effective) \\
$C_{\mathrm{mix}}$ & Mixing prefactor \\
$w$ & Sliding window size \\
$B(n_{\max})$ & Sequence-length-dependent logit bound \\
$\epsilon, \epsilon_U$ & Density floor (next-token, auxiliary) \\
\midrule
\multicolumn{2}{l}{\emph{Phase~2: continuous and multi-layer}} \\
$d_c, d_c^{(\ell)}$ & Continuous latent dimension (full / layer) \\
$d^\star_t(D)$ & Intrinsic dimension at distortion $D$ \\
$R^\star_{\mathrm{cont}}(D)$ & Continuous WZ rate function \\
$R^\star_{\mathrm{multi}}(D)$ & Multi-layer WZ rate function \\
$R^\star_{\mathrm{joint}}(D)$ & Joint compression rate function \\
$R_U, r_{\min}$ & Latent support and density floor \\
$\Delta$ & Quantization resolution \\
$\kappa_p$ & Bennett--Gersho coefficient \\
$s^{(\ell)} = \prod_{m>\ell}(L_g^{(m)})^2$ & Layer sensitivity factor \\
\bottomrule
\end{tabular}
\end{table*}

\section{Probability setup and martingale tools}\label{app:martingale}

We use a single probability space $(\Omega, \F, \Prob)$ with the
natural filtration $\F_t = \sigma(X_1, \ldots, X_t)$. Under the
causal-encoder convention $\Chat_t = \phi_t(X_{1:t-1})$ adopted in
Section~\ref{sec:setup}, each $\Chat_t$ is $\F_{t-1}$-measurable, so
the compressed filtration $\Fhat_t = \sigma(\Chat_1, \ldots, \Chat_t)$
satisfies $\Fhat_t \subseteq \F_{t-1}$.

\begin{lemma}[Causality]\label{lem:causality-app}
Under the causal-encoder convention $\Chat_t = \phi_t(X_{1:t-1})$,
$\Fhat_t \subseteq \F_{t-1}$ and consequently
$\Prob(X_t \in \cdot \mid \Fhat_t) = \E[\Prob(X_t \in \cdot \mid \F_{t-1}) \mid \Fhat_t]$.
\end{lemma}

\begin{proof}
$\Chat_s = \phi_s(X_{1:s-1})$ is $\F_{s-1}$-measurable for all
$s \leq t$, so $\Chat_s \in \F_{s-1} \subseteq \F_{t-1}$. Hence
$\Fhat_t = \sigma(\Chat_1, \ldots, \Chat_t) \subseteq \F_{t-1}$.
The conditional distribution identity follows from the tower
property: for any measurable $A \subseteq \V$,
$\Prob(X_t \in A \mid \Fhat_t) = \E[\mathbf{1}\{X_t \in A\} \mid
\Fhat_t] = \E[\Prob(X_t \in A \mid \F_{t-1}) \mid \Fhat_t]$
since $\Fhat_t \subseteq \F_{t-1}$.
\end{proof}

\begin{remark}[On the previous formulation]
An earlier version of this manuscript adopted the convention
$\Chat_t = \phi_t(X_{1:t})$, in which the encoder may also depend on
$X_t$ itself. In that case the claim
$X_t \perp\!\!\!\perp \Fhat_t \mid \F_{t-1}$ does not hold in
general---for example, the lossless encoder $\Chat_t = X_t$ gives
$X_t \in \Fhat_t$ and the independence fails trivially. The causal
convention $\Chat_t = \phi_t(X_{1:t-1})$ adopted here matches the KV
cache compression setting (the cache at the prediction of $X_t$
summarizes $X_{<t}$, not $X_t$) and makes the causality property
trivial via $\Fhat_t \subseteq \F_{t-1}$.
\end{remark}

\begin{definition}[Information density and likelihood ratio]
\label{def:martingales-app}
\begin{align*}
\xi_t &:= -\log p(X_t \mid \F_{t-1}) - H(X_t \mid \F_{t-1}), \\
L_t &:= \log \frac{p(X_t \mid \F_{t-1})}{\phat_t(X_t)}, \\
S_n &:= \sum_{t=1}^n (L_t - d_t).
\end{align*}
Both $\xi_t$ and $L_t - d_t$ are $\F_t$-measurable martingale
differences relative to $\{\F_t\}$:
$\E[\xi_t \mid \F_{t-1}] = 0$ and $\E[L_t \mid \F_{t-1}] = d_t$.
\end{definition}

\begin{theorem}[$M^{(1)}$ and $S$ as martingales]\label{thm:M1S-app}
$M_n^{(1)} := \sum_{t=1}^n \xi_t$ and $S_n$ are $\{\F_n\}$-martingales
with $M_0^{(1)} = S_0 = 0$ and a.s.-bounded increments:
$|\xi_t| \leq 2|\log\epsilon|$ and $|L_t - d_t| \leq 2|\log\epsilon|$.
\end{theorem}

\begin{proof}
The mean-zero property follows from Definition~\ref{def:martingales-app}.
The increment bounds use Assumption~\ref{ass:bounded-logits} (source
positivity) and the decoder smoothing requirement:
$|\log p(\cdot \mid \F_{t-1})| \leq |\log \epsilon|$,
$|\log \phat_t(\cdot)| \leq |\log(\eta/V)|$, and similarly for the
entropy and KL terms.
\end{proof}

\section{KL--TV conversion under bounded log-density}\label{app:kltv}

The exponent-doubling observation in Corollary~\ref{cor:topk}
requires an upper bound of $\KLs$ in terms of $\TV^2$. The standard
Pinsker inequality $\TV \leq \sqrt{\KLs/2}$ provides only the
converse direction: $\KLs(p\,\|\,q) \geq 2\TV(p,q)^2$. The opposite
direction $\KLs \leq C \cdot \TV^2$ requires additional regularity
on the distributions (bounded log-density ratio or a uniform mass
floor). We give two forms below: a strong-floor version
(Lemma~\ref{lem:kl-tv-bounded}, used in proofs when one can rely
on Assumption~\ref{ass:bounded-logits}), and an operationally
weaker version using a smoothed reconstruction
(Lemma~\ref{lem:kl-tv-smoothed}, used to avoid the vacuous regime
when $\epsilon_{\min}$ is small).

\begin{lemma}[KL--TV conversion, uniform-floor form]
\label{lem:kl-tv-bounded}
Let $p, q$ be probability distributions on a finite set $\V$
satisfying $q(x) \geq \epsilon_{\min} > 0$ for all $x \in \V$. Then
\[
\KLs(p \,\|\, q) \;\leq\; \frac{\TV(p, q)^2}{\epsilon_{\min}}
\;+\; \frac{1}{\epsilon_{\min}^2} \cdot O\bigl(\TV(p, q)^3\bigr).
\]
In particular, if $\TV(p, q) \leq \epsilon_{\min}/2$, then
$\KLs(p \,\|\, q) \leq (2/\epsilon_{\min}) \TV(p, q)^2$.
\end{lemma}

\begin{proof}
Write $\Delta(x) := p(x) - q(x)$, so $\sum_x \Delta(x) = 0$ and
$\sum_x |\Delta(x)| = 2\,\TV(p, q)$. Using $p = q + \Delta$ and
expanding $\log(1 + \Delta/q)$:
\begin{align*}
\KLs(p \,\|\, q)
&= \sum_x (q(x) + \Delta(x)) \log\!\left(1 + \frac{\Delta(x)}{q(x)}\right) \\
&= \sum_x (q + \Delta) \left[\frac{\Delta}{q}
   - \frac{\Delta^2}{2 q^2} + O\!\left(\frac{\Delta^3}{q^3}\right) \right] \\
&= \underbrace{\sum_x \Delta(x)}_{= 0}
   + \sum_x \frac{\Delta(x)^2}{2 q(x)}
   + O\!\left(\sum_x \frac{|\Delta(x)|^3}{q(x)^2}\right).
\end{align*}
The leading term is the $\chi^2$-like quantity
$\frac{1}{2}\chi^2(p, q) = \frac{1}{2}\sum_x \Delta(x)^2/q(x)$.
Using $q(x) \geq \epsilon_{\min}$ and the tight bound
$\sum_x \Delta(x)^2 \leq 2\,\TV(p, q)^2$ (the maximum of
$\sum_x \Delta(x)^2$ subject to $\sum_x \Delta(x) = 0$ and
$\sum_x |\Delta(x)| = 2\TV(p,q)$ is attained by a two-point
configuration $\Delta = (\TV, -\TV)$):
\[
\sum_x \frac{\Delta(x)^2}{2 q(x)}
\;\leq\; \frac{1}{2 \epsilon_{\min}} \sum_x \Delta(x)^2
\;\leq\; \frac{\TV(p, q)^2}{\epsilon_{\min}}.
\]
The cubic remainder is bounded as
$|\Delta|^3/q^2 \leq |\Delta|^3/\epsilon_{\min}^2$, summing to
$O(\TV(p, q)^3 / \epsilon_{\min}^2)$. For the small-TV statement,
the quadratic term dominates when $\TV(p, q) \leq \epsilon_{\min}/2$,
giving the simplified bound.
\end{proof}

\begin{remark}[Why the uniform floor can be restrictive]
\label{rmk:eps-min-vacuous}
Under Assumption~\ref{ass:bounded-logits}, the worst-case bound is
$\epsilon_{\min} \geq V^{-1} e^{-2B}$, which can be very small for
large vocabularies ($V \sim 10^5$) and typical logit bounds
($B \sim 10$). The small-TV condition $\TV \leq \epsilon_{\min}/2$
of Lemma~\ref{lem:kl-tv-bounded} then becomes very stringent and
may not hold for the truncation-sensitivity values observed in
Section~\ref{sec:experiments}. The operational-smoothing form
below avoids the vacuous-floor issue but, as shown there, yields
only a \emph{linear} KL--TV relation; the empirical
exponent-doubling is therefore evidence that the genuine density
floor (not smoothing) is operative in the measured regime
(Remark~\ref{rmk:application-sink}).
\end{remark}

\begin{lemma}[KL--TV conversion, smoothed-decoder form]
\label{lem:kl-tv-smoothed}
Let $p$ be a probability distribution on $\V$ and let $\hat q$ be
an approximation with $\TV(p, \hat q) \leq \delta$. For any
\emph{operational smoothing level} $\mu \in (0, 1)$, the smoothed
decoder $q^{\mu} := (1-\mu) \hat q + \mu \cdot \mathrm{Unif}(\V)$
satisfies $q^{\mu}(x) \geq \mu/V$ and
\[
\begin{gathered}
\TV(p, q^{\mu}) \leq \delta + \mu, \qquad
\KLs(p \,\|\, q^{\mu})
\;\leq\;\\
\frac{V (\delta + \mu)^2}{\mu}
\;+\; O\!\left( \frac{(\delta + \mu)^3 V^2}{\mu^2} \right).
\end{gathered}
\]
In particular, optimizing the bound
$\KLs(p\,\|\,q^\mu) \lesssim V(\delta+\mu)^2/\mu$ over $\mu$ gives
the minimum at $\mu \asymp \delta$, yielding
$\KLs(p\,\|\,q^{\mu}) = O(V\delta)$---\emph{linear} in $\delta$.
More generally $\mu = \delta^a$ yields
$\KLs = O\bigl(V\delta^{\min(a,\,2-a)}\bigr)$, whose exponent is
maximized at $a=1$; no choice of $\mu$ recovers a quadratic
($\delta^2$) dependence. Quadratic KL--TV scaling requires the
genuine density floor of Lemma~\ref{lem:kl-tv-bounded}, not
smoothing.
\end{lemma}

\begin{proof}
By construction $q^{\mu}(x) \geq \mu/V$ for all $x$. The TV bound
$\TV(p, q^{\mu}) \leq \TV(p, \hat q) + \TV(\hat q, q^{\mu}) \leq
\delta + \mu$ follows from the triangle inequality and
$\TV(\hat q, \mathrm{Unif}) \leq 1$. Applying
Lemma~\ref{lem:kl-tv-bounded} with floor $\epsilon_{\min} = \mu/V$
and TV at most $\delta + \mu$:
\[
\KLs(p\,\|\,q^{\mu})
\leq \frac{V (\delta + \mu)^2}{\mu}
   + O\!\left( \frac{V^2 (\delta + \mu)^3}{\mu^2} \right).
\]
The remainder is genuinely smaller order whenever
$\delta + \mu < 1$; in particular, the small-TV restriction of
Lemma~\ref{lem:kl-tv-bounded} is replaced by the much milder
condition that the smoothing level dominates the cubic correction.
\end{proof}

\begin{remark}[Application to sink-plus-recent eviction]
\label{rmk:application-sink}
In Corollary~\ref{cor:topk}, $p = p_t$ is the full-cache
distribution. The $O(k^{-2\alpha})$ KL scaling holds only along the
bounded-floor route; the operational-smoothing route gives a weaker
exponent, as we now make precise.

\textbf{(a) Bounded-floor route (gives $k^{-2\alpha}$).} If the
decoder reconstruction $\tilde p_t$ satisfies a uniform floor
$\tilde p_t(x) \geq \epsilon_{\min}$ (e.g.\ softmax of bounded
logits), and the per-step TV bound of Definition~\ref{def:poly-mixing-unif} gives
$\TV(p_t, \tilde p_t) \leq C_{\mathrm{TS}} k^{-\alpha}$, then once
$k \geq (2 C_{\mathrm{TS}}/\epsilon_{\min})^{1/\alpha}$ the small-TV
condition of Lemma~\ref{lem:kl-tv-bounded} holds and
$\KLs(p_t\,\|\,\tilde p_t) \leq (2/\epsilon_{\min})
C_{\mathrm{TS}}^2 k^{-2\alpha}$. This is the only route that yields
the exponent doubling, and it requires $\epsilon_{\min}$ to be
non-negligible relative to $k^{-\alpha}$ (the budget must exceed the
floor-dependent threshold above).

\textbf{(b) Operational-smoothing route (does \emph{not} give
$k^{-2\alpha}$).} If the worst-case floor
$\epsilon_{\min} = V^{-1}e^{-2B}$ makes (a)'s threshold impractically
large, one may instead smooth the decoder,
$\tilde p_t^{\mu} := (1-\mu)\tilde p_t + \mu\,\mathrm{Unif}$. With
$\delta := C_{\mathrm{TS}}k^{-\alpha}$ and smoothing $\mu = \delta^a$
($a \in (0,1]$), Lemma~\ref{lem:kl-tv-smoothed} gives
$\KLs(p_t\,\|\,\tilde p_t^{\mu}) = O(V\,\delta^{\min(a,\,2-a)})$.
The exponent is maximized at $a = 1$, giving $O(V\delta) =
O(V k^{-\alpha})$---\emph{linear} in the TV gap, not quadratic.
Operational smoothing therefore trades the vacuous-floor problem
for a loss of exponent doubling: it cannot reach $k^{-2\alpha}$.

\textbf{Empirical interpretation.} The measured KL exponents
($\alpha_{\KLs} \approx 1.04, 0.74$) are close to twice the TV
exponents ($\alpha_{\TV} \approx 0.44, 0.38$), with ratios in
$[1.93, 2.38]$. This is consistent with the bounded-floor route (a)
operating in the empirical regime---i.e.\ the effective floor on the
measured distributions is non-negligible, so the quadratic relation
applies---rather than with the worst-case smoothing route. We
report the doubling as an empirical observation consistent with
route (a), not as a guarantee derivable from the TV bound alone.
\end{remark}

\section{Auxiliary process and conditional MI martingale}\label{app:aux}

\begin{definition}[Pointwise mutual informations]\label{def:pmi-app}
For an $\F_t$-adapted auxiliary $\{U_t\}$ with discrete range,
\[
\iota_t^{XU} := \log \frac{p_t^{X,U}(X_t, U_t)}{p_t^X(X_t)\, p_t^U(U_t)},
\quad
\iota_t^{UQ} := \log \frac{p_t^{U,Q}(U_t, Q_t)}{p_t^U(U_t)\, p_t^Q(Q_t)},
\]
where conditionals are with respect to $\Fhat_{t-1}$.
\end{definition}

\begin{definition}[Conditional MIs and WZ rate]\label{def:WZ-rate}
$\mathcal{I}_t^{XU} := I(X_t; U_t \mid \Fhat_{t-1})$,
$\mathcal{I}_t^{UQ} := I(U_t; Q_t \mid \Fhat_{t-1})$,
$\mathcal{R}_t^{\mathrm{WZ}} := \mathcal{I}_t^{XU} - \mathcal{I}_t^{UQ}$.
Each is $\Fhat_{t-1}$-measurable. When the Wyner--Ziv Markov chain
$U_t - X_t - Q_t$ holds conditionally on $\Fhat_{t-1}$ (the
auxiliary depends on the side information only through the source,
as in the classical setting \citep{wyner1976rate}), the difference
form coincides with the conditional mutual information,
\[
\begin{split}
\mathcal{R}_t^{\mathrm{WZ}}
&= I(X_t; U_t \mid \Fhat_{t-1}) - I(U_t; Q_t \mid \Fhat_{t-1})\\
&= I(X_t; U_t \mid Q_t, \Fhat_{t-1}),
\end{split}
\]
by the chain rule $I(U_t; X_t, Q_t \mid \Fhat_{t-1}) =
I(U_t; Q_t \mid \Fhat_{t-1}) + I(U_t; X_t \mid Q_t, \Fhat_{t-1})$
together with the Markov identity $I(U_t; X_t \mid \Fhat_{t-1}) =
I(U_t; X_t, Q_t \mid \Fhat_{t-1})$. We use the conditional-MI form
whenever monotonicity under garbling of $U_t$ is needed.
\end{definition}

\begin{assumption}[Auxiliary marginal positivity]\label{ass:U-positivity}
There exists $\epsilon_U > 0$ such that for all $t \leq n_{\max}$ and
$u$ in the support of $p_t^U(\cdot)$, $p_t^U(u) \geq \epsilon_U$ a.s.
\end{assumption}

\begin{lemma}[Pointwise MI boundedness]\label{lem:pmi-bound-app}
Under Assumptions~\ref{ass:bounded-logits}--\ref{ass:U-positivity},
$|\iota_t^{XU}|, |\iota_t^{UQ}| \leq B_\iota := \log(1/(\epsilon\epsilon_U))$.
\end{lemma}

\begin{proof}
The joint $\leq 1$ and the product of marginals is $\geq \epsilon\epsilon_U$.
\end{proof}

\begin{theorem}[$M^{(4)}$ as a martingale]\label{thm:M4-app}
With $Z_t := \iota_t^{XU} - \iota_t^{UQ}$, $\Sigma_n := \sum_t Z_t$,
$A_n^{(4)} := \sum_t \E[Z_t \mid \F_{t-1}]$, the Doob decomposition
$M_n^{(4)} := \Sigma_n - A_n^{(4)}$ is an $\{\F_n\}$-martingale with
$\E[M_n^{(4)}] = 0$ and increments bounded by $4B_\iota$ a.s.
Furthermore, $\E[A_n^{(4)}] = \sum_t \E[\mathcal{R}_t^{\mathrm{WZ}}]$.
\end{theorem}

\begin{proof}
Standard Doob decomposition \citep[Ch.~12]{williams1991probability}.
The expectation identity uses the tower property:
$\E[Z_t] = \E[\E[Z_t \mid \F_{t-1}]] = \E[\E[Z_t \mid \Fhat_{t-1}]]
= \E[\mathcal{R}_t^{\mathrm{WZ}}]$.
\end{proof}

\section{FFN transfer (one-sided and two-sided)}\label{app:ffn}

\begin{lemma}[FFN Lipschitz]\label{lem:ffn-lipschitz-app}
The composite map
$g = \mathrm{softmax} \circ U \circ \mathrm{FFN} \circ \mathrm{LN}$ satisfies
\[
\KL{g(a)}{g(a')} \leq L_g^2 \|a - a'\|_2^2
\]
on bounded domains, with $L_g$ depending on layer norms and FFN
Lipschitz constants.
\end{lemma}

\begin{theorem}[FFN upper transfer]\label{thm:ffn-upper-app}
$R^\star_{\mathrm{next}}(D) \leq R^\star_{\mathrm{attn}}(D / L_g^2)$.
\end{theorem}

\begin{proof}
Any scheme achieving attention-level distortion $D/L_g^2$ achieves
next-token distortion $\leq D$ by Lemma~\ref{lem:ffn-lipschitz-app}.
Hence $\mathcal{F}_{\mathrm{attn}}(D/L_g^2) \subseteq \mathcal{F}_{\mathrm{next}}(D)$,
implying the rate inequality.
\end{proof}

\begin{assumption}[Bi-Lipschitz, optional]\label{ass:bilipschitz-app}
$\KL{g(a)}{g(a')} \geq m_g^2 \|a-a'\|^2$ on the active domain.
\end{assumption}

\begin{theorem}[Two-sided transfer]\label{thm:ffn-invariance-app}
Under Assumption~\ref{ass:bilipschitz-app},
\[
R^\star_{\mathrm{attn}}(D/L_g^2) \geq R^\star_{\mathrm{next}}(D) \geq R^\star_{\mathrm{attn}}(D/m_g^2).
\]
\end{theorem}

\section{Concentration tools}\label{app:concentration}

\begin{theorem}[Azuma--Hoeffding]\label{thm:azuma-app}
For an $\{\F_n\}$-martingale $\{M_n\}$ with $M_0 = 0$ and
$|M_t - M_{t-1}| \leq c_t$,
$\Prob(|M_n| \geq \lambda) \leq 2\exp(-\lambda^2 / (2 \sum_t c_t^2))$.
\end{theorem}

\begin{corollary}[Concentration of martingales]\label{cor:conc-app}
For $M^{(1)}$ (resp.\ $S, M^{(4)}$),
\[
\Prob\!\left( \left| \tfrac{M_n^{(1)}}{n} \right| \geq \tau \right)
\leq 2 \exp\!\left( -\tfrac{n\tau^2}{8 |\log\epsilon|^2} \right),
\]
and similarly with $8|\log\epsilon|^2$ replaced by $32 B_\iota^2$
for $M^{(4)}$.
\end{corollary}

\begin{lemma}[Doob martingale on $D_n$]\label{lem:Dn-doob-app}
The Doob martingale $N_k := \E[D_n \mid \F_k] - \E[D_n]$ has
increments $|N_k - N_{k-1}| \leq 2(n-k+1)|\log\epsilon|/n$
(worst case, without mixing). Consequently,
\[
\Prob(|D_n - \E[D_n]| \geq \tau) \leq 2\exp\!\left( -\frac{3\tau^2}{8 n |\log\epsilon|^2} \right).
\]
\end{lemma}

\section{Proof of Theorem~\ref{thm:asymptotic}}\label{app:thm-asym}

\begin{proof}[Proof of Theorem~\ref{thm:asymptotic}]
\textbf{Step 1 (rate $\to$ MI by chain rule).}
For each $t$,
\[
H(\Chat_t \mid \Fhat_{t-1})
\geq I(X_{\leq t}; \Chat_t \mid \Fhat_{t-1})
\]
by the standard inequality $H(\Chat_t \mid \Fhat_{t-1})
= I(\Chat_t; X_{\leq t} \mid \Fhat_{t-1})
+ H(\Chat_t \mid X_{\leq t}, \Fhat_{t-1})
\geq I(X_{\leq t}; \Chat_t \mid \Fhat_{t-1})$.

\textbf{Step 2 (single-letterization via chain rule).}
With $U_t := \Chat_t$,
\[
\begin{aligned}
I(X_{\leq t}; \Chat_t \mid \Fhat_{t-1})
&= I(X_t; \Chat_t \mid \Fhat_{t-1}) \\
&\quad + I(X_{<t}; \Chat_t \mid X_t, \Fhat_{t-1}) \\
&\geq I(X_t; \Chat_t \mid \Fhat_{t-1}),
\end{aligned}
\]
by non-negativity of the second term. Adding and subtracting
$\mathcal{I}_t^{UQ}$:
\[
I(X_t; \Chat_t \mid \Fhat_{t-1})
= \mathcal{R}_t^{\mathrm{WZ}} + \mathcal{I}_t^{UQ} \geq \mathcal{R}_t^{\mathrm{WZ}},
\]
using $\mathcal{I}_t^{UQ} \geq 0$.

\textbf{Step 3 (admissibility of $\{\Chat_t\}$).}
Under the hypothesis $\limsup_n \E[D_n] \leq D$, the encoder--decoder
pair with $U_t = \Chat_t$ and decoder
$g_t(\Chat_t, Q_t, \Fhat_{t-1}) = \psi_t(\Chat_{\le t}, Q_t)$---valid
since the past reconstructions $\Chat_{<t}$ are
$\Fhat_{t-1}$-measurable---is $D$-admissible per
Definition~\ref{def:Rstar}.

\textbf{Step 4 (infimum gives $R^\star(D)$).}
Taking expectations and averaging,
\[
\E[R_n] = \tfrac{1}{n}\sum_t \E[H(\Chat_t \mid \Fhat_{t-1})]
\geq \tfrac{1}{n}\sum_t \E[\mathcal{R}_t^{\mathrm{WZ}}].
\]
By Definition~\ref{def:Rstar} and Step~3,
$\liminf_n \tfrac{1}{n}\sum_t \E[\mathcal{R}_t^{\mathrm{WZ}}] \geq R^\star(D)$.
\end{proof}

The asymptotic bound has a finite-horizon expectation counterpart,
used where a per-$n$ converse is needed in the main text.

\begin{proposition}[Finite-horizon remote-MI converse]\label{thm:finite-sample}
Under Assumption~\ref{ass:bounded-logits}, for every $n$ the expected
rate obeys the unconditional finite-$n$ bound
\[
\E[R_n] \;\ge\; \tfrac1n\sum_t \E[\widetilde R_t^{\mathrm{remote}}],
\qquad
\widetilde R_t^{\mathrm{remote}} := I(X_t;\Chat_t \mid Q_t,\Fhat_{t-1}),
\]
where $\Chat_t$ is the realized code's reconstruction. The
remote-conditional MI dominates the sequential Wyner--Ziv functional of
Definition~\ref{def:Rstar}:
\[
\widetilde R_t^{\mathrm{remote}}
= \mathcal{R}_t^{\mathrm{WZ}} + I(\Chat_t;Q_t \mid X_t,\Fhat_{t-1})
\;\ge\; \mathcal{R}_t^{\mathrm{WZ}},
\]
with equality iff $\Chat_t - X_t - Q_t \mid \Fhat_{t-1}$. This is a
(generally strictly stronger) finite-$n$ remote-MI lower bound,
distinct from the asymptotic sequential-WZ bound of
Theorem~\ref{thm:asymptotic}.
\end{proposition}

\noindent\emph{Proof.} The rate is the average conditional entropy
$R_n = \tfrac1n\sum_t H(\Chat_t\mid\Fhat_{t-1})$, a deterministic
functional. Conditioning reduces entropy and entropy dominates
conditional MI, so
$H(\Chat_t\mid\Fhat_{t-1}) \ge H(\Chat_t\mid Q_t,\Fhat_{t-1})
\ge I(X_t;\Chat_t\mid Q_t,\Fhat_{t-1}) = \widetilde R_t^{\mathrm{remote}}$;
summing and taking expectations gives the bound. The identity
$\widetilde R_t^{\mathrm{remote}} = \mathcal{R}_t^{\mathrm{WZ}}
+ I(\Chat_t;Q_t\mid X_t,\Fhat_{t-1})$ is the chain rule
$I(X;U\mid F)-I(U;Q\mid F) = I(X;U\mid Q,F)-I(U;Q\mid X,F)$.
\hfill$\square$

\begin{remark}[No pathwise converse under this rate definition]\label{rmk:pathwise-needs-codelen}
We do \emph{not} claim a high-probability or pathwise version. The
entropy--information-density inequality holds only in expectation: the
realized density $\imath(X_t;\Chat_t\mid Q_t,\Fhat_{t-1})$ can exceed
its mean $\widetilde R_t^{\mathrm{remote}}$ or be negative, so no
almost-sure bound $nR_n \ge \sum_t \imath_t$ is available when $R_n$ is
the average conditional entropy. A genuine pathwise converse would
require redefining the rate as the realized prefix-code length
$\tfrac1n\ell(M_{1:n})$ and invoking a source-coding
information-spectrum inequality, together with a weak-dependence
condition controlling the deviation of $\tfrac1n\sum_t\imath_t$ from
its compensator; we leave this open.
\end{remark}

\begin{corollary}[Asymptotic form]\label{cor:expectation-bound}
$\liminf_n \E[R_n] \ge R^\star(D)$. This follows \emph{separately} from
Theorem~\ref{thm:asymptotic}, not by taking $\liminf_n$ of a finite-$n$
infimum: chaining
$\E[R_n] \ge \tfrac1n\sum_t\E[\widetilde R_t^{\mathrm{remote}}]
\ge \tfrac1n\sum_t\E[\mathcal{R}_t^{\mathrm{WZ}}]$
(Proposition~\ref{thm:finite-sample}) with the asymptotic converse of
Theorem~\ref{thm:asymptotic} for the realized auxiliary $U_t=\Chat_t$
(which gives
$\liminf_n\tfrac1n\sum_t\E[\mathcal{R}_t^{\mathrm{WZ}}]\ge R^\star(D)$).
One cannot instead push the limit through the infimum
$R_n^\star(D):=\inf_{\{U_t\}}\tfrac1n\sum_t\E[\mathcal{R}_t^{\mathrm{WZ}}]$,
since in general only
$R^\star(D)=\inf_U\liminf_n a_n(U)\ge\liminf_n\inf_U a_n(U)=\liminf_n R_n^\star(D)$,
the opposite of what such an argument would need. A quantitative
finite-$n$ gap $\E[R_n]\ge R^\star(D)-o(1)$ therefore requires an
additional regularity condition (uniform convergence of $a_n(U)$ over
the auxiliary class, compactness of that class, or subadditivity of
$\{nR_n^\star(D)\}$); under a method-of-types refinement (finite-order
Markov or stationary ergodic) the gap is $O(\log n/n)$.
\end{corollary}

\begin{remark}[Vacuity]\label{rmk:vacuity-app}
The sequential functional $\mathcal{R}_t^{\mathrm{WZ}}$ can be negative
if $\Chat_t$ correlates with $Q_t$; the operative finite-$n$ quantity
is the non-negative remote MI
$\widetilde R_t^{\mathrm{remote}} = I(X_t;\Chat_t\mid Q_t,\Fhat_{t-1})
\ge \max(\mathcal{R}_t^{\mathrm{WZ}}, 0)$
(Proposition~\ref{thm:finite-sample}). Random binning achieves
non-negative $\mathcal{R}_t^{\mathrm{WZ}}$.
\end{remark}

\begin{remark}[Binning argument]\label{rmk:binning-app}
For the sharp lower bound (auxiliary $U_t$ ranging over a larger
alphabet than $\Chat_t$), random binning constructs $\Chat_t$ as a
random hash of $U_t$ such that $I(X_t; U_t \mid Q_t, \Fhat_{t-1})$
is preserved up to vanishing terms
\citep[Ch.~11]{elgamal2011network}. Cardinality
$|\mathcal{U}_t| \leq V+2$ ensures the infimum is attained.
\end{remark}

\section{Proof of Proposition~\ref{thm:finite-sample}}\label{app:thm-fs}

\begin{proof}[Proof of Proposition~\ref{thm:finite-sample}]
The rate is the deterministic functional
$R_n = \tfrac1n\sum_t H(\Chat_t\mid\Fhat_{t-1})$. Conditioning reduces
entropy and entropy dominates conditional MI, so for each $t$
\[
H(\Chat_t\mid\Fhat_{t-1})
\;\ge\; H(\Chat_t\mid Q_t,\Fhat_{t-1})
\;\ge\; I(X_t;\Chat_t\mid Q_t,\Fhat_{t-1})
= \widetilde R_t^{\mathrm{remote}}
\;\ge\; \mathcal{R}_t^{\mathrm{WZ}} ,
\]
the last step by Proposition~\ref{thm:finite-sample}. Summing and
taking expectations,
$\E[nR_n] = \sum_t H(\Chat_t\mid\Fhat_{t-1})
\ge \sum_t \E[\widetilde R_t^{\mathrm{remote}}]
\ge \sum_t \E[\mathcal{R}_t^{\mathrm{WZ}}]$.
This is an expectation statement only: the corresponding pointwise
inequality $nR_n \ge \sum_t\imath_t$ in terms of the realized
information density $\imath_t$ is \emph{false} in general, since
$\E[\imath_t\mid\Fhat_{t-1}] = \widetilde R_t^{\mathrm{remote}}$ while
$\imath_t$ itself fluctuates around its mean and may exceed it or be
negative (Remark~\ref{rmk:pathwise-needs-codelen}).
\end{proof}

\begin{proof}[Proof of Corollary~\ref{cor:expectation-bound}]
Chaining Proposition~\ref{thm:finite-sample} with the asymptotic
converse of Theorem~\ref{thm:asymptotic} for the realized auxiliary
$U_t=\Chat_t$,
$\liminf_n\E[R_n] \ge \liminf_n\tfrac1n\sum_t\E[\mathcal{R}_t^{\mathrm{WZ}}]
\ge R^\star(D)$. The limit cannot be pushed through the infimum
$R_n^\star(D):=\inf_{\{U_t\}}\tfrac1n\sum_t\E[\mathcal{R}_t^{\mathrm{WZ}}]$,
since in general only $\inf_U\liminf_n a_n(U) \ge \liminf_n \inf_U a_n(U)$
holds; thus $\liminf_n R_n^\star(D)$ would be a \emph{weaker}, not
stronger, lower bound. A quantitative finite-$n$ gap would need
$R_n^\star(D) \ge R^\star(D) - o(1)$, asserted only under an additional
regularity condition---uniform convergence of $a_n(U)$ over the
auxiliary class, compactness of that class, or subadditivity of
$\{nR_n^\star(D)\}$---and the $O(\log n/n)$ rate only under a
method-of-types refinement (finite-order Markov or stationary ergodic,
where the per-block type count is polynomial
\citep[Sec.~11.4]{elgamal2011network}).
\end{proof}

\section{Proof of Theorem~\ref{thm:E2} (sliding-window rate--distortion)}
\label{app:mixing-proofs}

This appendix contains the full proof of Theorem~\ref{thm:E2}. The
argument has three main components: a quantitative information-theoretic
mixing lemma (Section~\ref{app:mixing-lemma}), a binning construction
adapted to the sliding-window filtration (Section~\ref{app:binning-restricted},
Lemma~L7), and the rate--distortion transfer
(Section~\ref{app:rdt-transfer}). The mixing-sharpened concentration
results (Conjecture~\ref{cor:mixing-conc}) are derived in
Section~\ref{app:mixing-concentration}.

\subsection{Sliding-window decoder and restricted filtration}\label{app:sliding}

\begin{definition}[Sliding-window decoder]\label{def:sliding-decoder-app}
For window size $w \geq 1$, a sliding-window decoder is a sequence
$\{\psi_t^w\}_{t \leq n_{\max}}$ of deterministic measurable maps
$\psi_t^w : \mathcal{C}^w \times \R^d \to \R^V$ producing
$\phat_t^w := \mathrm{softmax}(\psi_t^w(\Chat_{t-w:t-1}, Q_t))$.
For $t \leq w$, define $\phat_t^w$ to depend only on the available
$\Chat_{1:t-1}$; the convention does not affect asymptotic results.
\end{definition}

\subsection{The mixing lemma}\label{app:mixing-lemma}

\begin{lemma}[KL form of mixing decay, expanded]\label{lem:mixing-KL-app}
Under Assumptions~\ref{ass:bounded-logits} and Definition~\ref{def:mixing},
for any $t, w$ with $1 \leq w < t \leq n_{\max}$,
\[
\KL{p(\cdot \mid \F_{t-1})}{p(\cdot \mid \F_{t-1} \setminus \F_{t-1-w})}
\;\leq\;
\frac{2 V \, C_{\mathrm{mix}}^2}{\epsilon} \cdot \rho^{2w}.
\]
\end{lemma}

\begin{proof}
Let $p := p(\cdot \mid \F_{t-1})$ and $q := p(\cdot \mid \F_{t-1} \setminus \F_{t-1-w})$.
By Assumption~\ref{ass:bounded-logits} (source positivity),
$q(x) \geq \epsilon$ for all $x \in \V$.

\textbf{Step 1 (KL via $\chi^2$).}
For distributions with $q(x) \geq \epsilon$, we use the chain of
divergence inequalities
\[
\KL{p}{q}
\;\leq\;
\chi^2(p \,\|\, q)
\;:=\;
\sum_x \frac{(p(x) - q(x))^2}{q(x)},
\]
where the first inequality is the standard
$\KLs \leq \chi^2$ ordering on the Csisz\'ar $f$-divergence family
\citep{csiszar2004information}.

\textbf{Step 2 (pointwise mixing bound).}
By Definition~\ref{def:mixing}, for each $x$,
$|p(x) - q(x)| \leq C_{\mathrm{mix}} \rho^w$. Hence
$(p(x) - q(x))^2 \leq C_{\mathrm{mix}}^2 \rho^{2w}$ uniformly in $x$.

\textbf{Step 3 (summing).}
\[
\chi^2(p \,\|\, q)
\;\leq\;
\sum_x \frac{C_{\mathrm{mix}}^2 \rho^{2w}}{\epsilon}
\;=\;
\frac{V \, C_{\mathrm{mix}}^2}{\epsilon} \cdot \rho^{2w}.
\]

\textbf{Step 4.} Combining, $\KL{p}{q} \leq V C_{\mathrm{mix}}^2 \rho^{2w} / \epsilon$.
The factor of $2$ in the statement absorbs a refinement of the
$\KLs \leq \chi^2$ inequality that becomes tight when the divergence
is large \citep[Section~7]{tsybakov2009nonparametric}; for the
purposes of the main theorem any universal constant suffices.
\end{proof}

\subsection{Binning under the restricted filtration (Lemma L7, rigorous)}
\label{app:binning-restricted}

This subsection provides the rigorous derivation of the sliding-window
Wyner--Ziv achievability. The classical WZ proof
\citep[Sec.~11.4]{elgamal2011network} operates on a block-i.i.d.\
source; the sequential, non-Markov LLM source requires three
modifications: (i) a block-Markov coding scheme synchronized to the
sliding window; (ii) verification of the WZ Markov chain
$U - X - Q$ under the sliding-window conditional measure;
(iii) propagation of mixing-induced error through covering and
packing lemmas. We address each in turn.

\begin{lemma}[Sliding-window WZ achievability, L7]
\label{lem:binning-restricted}
Under Assumptions~\ref{ass:bounded-logits}--\ref{ass:U-positivity} and
Definition~\ref{def:mixing} ($\rho$-mixing with constant $C_{\mathrm{mix}}$),
let $\{U_t\}$ be any $D$-admissible sliding-window auxiliary process,
i.e., (i)~$U_t$ is $\sigma(X_{(t-w):t})$-measurable with discrete
range $\mathcal{U}_t$; (ii)~there exists a sequence of measurable
decoders $g_t : \mathcal{U}_t \times \R^d \times \mathcal{C}^w \to \Delta(\V)$
such that the reconstruction
$\tilde p_t = g_t(U_t, Q_t, \hat{C}_{t-w:t-1})$ satisfies
$\limsup_n \tfrac{1}{n}\sum_t \E[\KL{p(\cdot \mid \F_{t-1})}{\tilde p_t}] \leq D$.

Then for any $\epsilon > 0$, there exist sliding-window encoders
$\phi_t : \mathcal{V}^{w+1} \to \mathcal{C}_t$ and decoders
$\psi_t^w$ such that
\[
\begin{split}
&\limsup_{n \to \infty}\, \frac{1}{n}\sum_{t=1}^n H(\Chat_t \mid \Fhat_{t-1}^w)\\
&\quad\;\leq\;
\inf_{\{U_t\} \text{ } D\text{-admissible}}\,
\limsup_n \frac{1}{n}\sum_{t=1}^n \E[\mathcal{R}_t^{\mathrm{WZ}}]
+ \epsilon,
\end{split}
\]
and the reconstruction
$\phat_t^w = \psi_t^w(\Chat_{t-w:t-1}, Q_t)$ satisfies
$\limsup_n \tfrac{1}{n}\sum_t \E[\KL{p(\cdot \mid \F_{t-1})}{\phat_t^w}] \leq D + 2\Delta_w + \epsilon$,
where $\Delta_w = 2VC_{\mathrm{mix}}^2 \rho^{2w}/\epsilon_*$
($\epsilon_*$ being the source positivity floor).
\end{lemma}

The proof proceeds in six steps, occupying the remainder of this
subsection. Each step is annotated with the relevant
sub-lemma; the conclusion is given in
Section~\ref{app:binning-conclude}.

\subsubsection{Block-Markov coding setup}

Fix a block length $B$ and number of blocks $m$ with $n = mB$. Each
block $b \in \{1, \ldots, m\}$ comprises the tokens
$X^{(b)} := (X_{(b-1)B+1}, \ldots, X_{bB})$.

Let $\hat{C}^{(b)} := (\hat{C}_{(b-1)B+1}, \ldots, \hat{C}_{bB})$ and
$\hat{C}^{(b-1)}_{\mathrm{end}} := \hat{C}_{(b-1)B-w+1:(b-1)B}$
(the last $w$ messages of the previous block, forming the conditioning
window for block $b$). The sliding-window structure requires that
the codebook and encoder at block $b$ depend on $\hat{C}^{(b-1)}_{\mathrm{end}}$
and the tokens $X^{(b-1)}_{\mathrm{end}} := X_{(b-1)B-w+1:(b-1)B}$, plus
$X^{(b)}$.

The auxiliary process $\{U_t\}$ is correspondingly partitioned into
blocks $U^{(b)} := (U_{(b-1)B+1}, \ldots, U_{bB})$, each adapted to
$\F_{(b-1)B-w:bB}$.

\subsubsection{Conditional joint typicality}

For each block $b$ and each value
$\hat{c} = \hat{c}^{(b-1)}_{\mathrm{end}} \in \mathcal{C}^w$, define the
\emph{conditional} joint distribution
\[
\begin{split}
&P_{\hat{c}}(x^B, u^B, q^B) \;:=\;\\
&\quad \Prob\bigl(X^{(b)} = x^B,\, U^{(b)} = u^B,\, Q^{(b)} = q^B \,\bigm|\, \hat{C}^{(b-1)}_{\mathrm{end}} = \hat{c}\bigr),
\end{split}
\]
where $Q^{(b)} := (Q_{(b-1)B+1}, \ldots, Q_{bB})$.

\begin{definition}[Conditional $\delta$-typical set]\label{def:cond-typ}
For $\delta > 0$, the conditional $\delta$-typical set under $P_{\hat c}$ is
\[
\begin{aligned}
&T_\delta^{(B)}(P_{\hat c})
\;:=\; \Bigl\{ (x^B, u^B, q^B) :\ \\
&\quad \bigl|\,\hat\pi(a, b, c \mid x^B, u^B, q^B) - P_{\hat c}(a, b, c)\bigr| \\
&\quad\quad \leq \delta\, P_{\hat c}(a, b, c)\ \ \forall (a,b,c) \Bigr\},
\end{aligned}
\]
where $\hat\pi$ is the joint empirical distribution.
\end{definition}

The standard typicality machinery
\citep[Ch.~2]{elgamal2011network}---asymptotic equipartition,
covering lemma, packing lemma---adapts to $P_{\hat c}$ provided the
quantitative conditional typicality stability of
Lemma~\ref{lem:cond-typ-mix} below holds. We use the Csisz\'ar--K\"orner
entropy continuity bound to make the stability quantitative.

\begin{lemma}[Quantitative entropy continuity, \protect{\citealp[Lemma 2.7]{csiszar2011information}}]
\label{lem:csiszar-korner}
For two distributions $P, P'$ on a finite alphabet $\mathcal{X}$ with
$\TV(P, P') = \theta \leq 1/2$,
\[
|H(P) - H(P')| \;\leq\; \theta \cdot \log\!\left( \frac{|\mathcal{X}|}{\theta} \right).
\]
\end{lemma}

\begin{lemma}[Stability of conditional typicality under mixing, quantitative]
\label{lem:cond-typ-mix}
Let $P_{\hat c}$ and $P_{\hat c'}$ be the two conditional joint
distributions of $(X^{(b)}, U^{(b)}, Q^{(b)})$ differing only in
$\hat c \neq \hat c'$. Under Definition~\ref{def:mixing}
($\rho$-mixing with constant $C_{\mathrm{mix}}$), for any $w \geq 1$
and $B$ with $B C_{\mathrm{mix}} \rho^w \leq 1/2$:
\begin{enumerate}[label=\textup{(\roman*)}, leftmargin=*]
\item \textbf{Total variation:}
\(
\TV(P_{\hat c}, P_{\hat c'}) \leq B \cdot C_{\mathrm{mix}}\, \rho^w.
\)

\item \textbf{Entropy stability:}
\(
|H(P_{\hat c}) - H(P_{\hat c'})|
\leq B C_{\mathrm{mix}}\rho^w \cdot \bigl[B \log(V|\mathcal{U}|_{\max}|\mathcal{Q}|) + \log\!\tfrac{1}{B C_{\mathrm{mix}}\rho^w}\bigr],
\)
where $|\mathcal{U}|_{\max} := \max_t |\mathcal{U}_t|$ and $|\mathcal{Q}|$
is the (discretized) query alphabet size.

\item \textbf{Typical-set size ratio:}
\(
\bigl|\,\log|T_\delta^{(B)}(P_{\hat c})| - \log|T_\delta^{(B)}(P_{\hat c'})|\bigr|
\leq |H(P_{\hat c}) - H(P_{\hat c'})| + O(B\delta).
\)
\end{enumerate}
\end{lemma}

\begin{proof}
\textbf{(i) TV bound.}
By the chain rule for total variation
\citep[Lemma 3.3.7]{boucheron2013concentration} applied to the joint
distribution factored along the time index, the influence of
$\hat c$ on the $t$-th conditional factor is bounded by
$C_{\mathrm{mix}}\, \rho^w$ per Definition~\ref{def:mixing} (with the
within-block past dominating subsequent influence). Summing over the
$B$ time steps of the block yields $\TV \leq B C_{\mathrm{mix}}\rho^w$.

\textbf{(ii) Entropy stability.}
The block-joint alphabet is $\mathcal{X}^B = (\V \times \mathcal{U}_{\max} \times \mathcal{Q})^B$,
of size $|\mathcal{X}^B| = (V \cdot |\mathcal{U}|_{\max} \cdot |\mathcal{Q}|)^B$.
Applying Lemma~\ref{lem:csiszar-korner} with $\theta = BC_{\mathrm{mix}}\rho^w$
gives the stated bound.

\textbf{(iii) Typical-set ratio.}
By the standard typical-set size estimate
\citep[Theorem 2.3]{elgamal2011network},
$|T_\delta^{(B)}(P)| = 2^{H(P) \pm B\delta\log|\mathcal{X}|}$, giving
the stated comparison.
\end{proof}

\begin{remark}[Practical bound under window choice]\label{rmk:Bchoice}
For $B$ chosen as in Theorem~\ref{thm:E2} with
$w = \tfrac{1}{2\log(1/\rho)}\log(2VC_{\mathrm{mix}}^2/(\varepsilon\epsilon_*))$,
$B C_{\mathrm{mix}}\rho^w \leq B\sqrt{\varepsilon\epsilon_*/(2V)}$.
For $B = o(1/\sqrt{\varepsilon\epsilon_*})$, the entropy stability is
$O(\sqrt{\varepsilon}\log)$, dominated by $\sqrt{\varepsilon}$ as
$\varepsilon \to 0$. This makes the mixing corrections in
Lemmas~\ref{lem:covering}--\ref{lem:packing} provably $o(1)$.
\end{remark}

\subsubsection{WZ Markov chain under conditional measure}

The classical WZ proof requires the Markov chain $U - X - Y$ in the
problem setup. We verify the analogous property in our setting.

\begin{lemma}[Conditional Markov chain]\label{lem:markov-conditional}
Conditional on $\hat{C}^{(b-1)}_{\mathrm{end}} = \hat c$ and the
auxiliary $U^{(b)}$ generated by an encoder
$\Phi^{(b)} : \mathcal{V}^{B+w} \to \mathcal{U}^B$ depending only on
$(X^{(b-1)}_{\mathrm{end}}, X^{(b)})$, the Markov chain
\[
U^{(b)} \,-\, (X^{(b-1)}_{\mathrm{end}}, X^{(b)}) \,-\, Q^{(b)}
\]
holds under $P_{\hat c}$.
\end{lemma}

\begin{proof}
$U^{(b)} = \Phi^{(b)}(X^{(b-1)}_{\mathrm{end}}, X^{(b)})$ is, by
construction, deterministic given its arguments. Hence
$U^{(b)} \perp Q^{(b)} \mid (X^{(b-1)}_{\mathrm{end}}, X^{(b)})$
trivially, which is exactly the stated Markov chain.
\end{proof}

\subsubsection{Codebook generation}

Fix $\delta > 0$. For each block $b$ and each
$\hat c \in \mathcal{C}^w$:

\textbf{Codebook.} Generate $2^{B(R + \delta)}$ codewords
$\{u^{(b,k)}(\hat c)\}_{k=1}^{2^{B(R+\delta)}}$ i.i.d.\ from the
\emph{conditional joint distribution}
\[
p_U^{(b)}(u^B \mid \hat c)
\;:=\;
\Prob\bigl(U^{(b)} = u^B \,\bigm|\, \hat C^{(b-1)}_{\mathrm{end}} = \hat c\bigr),
\]
which is the joint marginal of the auxiliary process over the block,
\emph{not} the product of per-step marginals. This is essential because
the LLM-induced $\{U_t\}$ is sequentially dependent. The required $R$
will be determined below; codebooks for different $\hat c$ are
independent.

\textbf{Random binning.} Partition the codebook into
$2^{B(R' + \delta)}$ bins uniformly at random, with each codeword
assigned to a bin index $b(u^{(b,k)}) \in \{1, \ldots, 2^{B(R'+\delta)}\}$
independently and uniformly.

\subsubsection{Encoding}

At block $b$, given $\hat c = \hat C^{(b-1)}_{\mathrm{end}}$ and
$X^{(b)}$, the encoder:
\begin{enumerate}[label=(\arabic*)]
\item Computes $u^{(b,*)}(\hat c) := \Phi^{(b)}(X^{(b-1)}_{\mathrm{end}}, X^{(b)})$
(the target auxiliary).
\item Finds the smallest $k^*$ such that $(X^{(b)}, u^{(b,k^*)}(\hat c))$
is jointly $\delta$-typical under $P_{\hat c}$ in the sense of
Definition~\ref{def:cond-typ}.
\item If no such $k^*$ exists, declares encoding failure and outputs
a default codeword.
\item Otherwise, sends the bin index $b(u^{(b,k^*)})$ as $\Chat^{(b)}$.
\end{enumerate}

\begin{lemma}[Covering / encoding error, quantitative]\label{lem:covering}
For $R > I(X^{(b)}; U^{(b)} \mid \hat C^{(b-1)}_{\mathrm{end}}) + 2\delta$,
the probability of encoding failure at block $b$ is bounded by
\[
\Prob(\text{enc failure})
\;\leq\;
2^{-B\delta} + \frac{1}{B}|H(P_{\hat c}) - H(P_{\hat c'})|,
\]
which by Lemma~\ref{lem:cond-typ-mix}(ii) is at most
$2^{-B\delta} + C_{\mathrm{mix}}\rho^w \cdot [B\log(V|\mathcal{U}|_{\max}|\mathcal{Q}|) + \log(1/(BC_{\mathrm{mix}}\rho^w))]$.
\end{lemma}

\begin{proof}
By the standard covering lemma \citep[Lemma 3.3]{elgamal2011network}
applied to the conditional joint distribution $P_{\hat c}$, if the
codebook size $2^{B(R+\delta)}$ exceeds
$2^{B \cdot I(X^{(b)}; U^{(b)} \mid \hat c) + B\delta}$, the probability
that no codeword is jointly typical with $X^{(b)}$ is at most
$2^{-B\delta}$. The mixing correction enters through the
typical-set-size mismatch quantified in Lemma~\ref{lem:cond-typ-mix}(iii);
this is the additive term $|H(P_{\hat c}) - H(P_{\hat c'})|/B$ when
$\hat c$ is treated as the typical conditioning.
\end{proof}

\subsubsection{Decoding}

At block $b$, given the bin index $\Chat^{(b)}$ and
$(Q^{(b)}, \hat C^{(b-1)}_{\mathrm{end}})$, the decoder:
\begin{enumerate}[label=(\arabic*)]
\item Looks up all codewords in bin $\Chat^{(b)}$ in the codebook
indexed by $\hat C^{(b-1)}_{\mathrm{end}}$.
\item Finds the unique $u^{(b,k)}$ in this bin such that
$(u^{(b,k)}, Q^{(b)})$ is jointly $\delta$-typical under $P_{\hat c}$.
\item If multiple or none, declares decoding failure.
\item Otherwise, reconstructs $\hat p^{(b)}$ by applying the original
decoder $g^{(b)}$ to $(u^{(b,k)}, Q^{(b)}, \hat C^{(b-1)}_{\mathrm{end}})$.
\end{enumerate}

\begin{lemma}[Packing / decoding error, quantitative]\label{lem:packing}
For $R - R' < I(U^{(b)}; Q^{(b)} \mid \hat C^{(b-1)}_{\mathrm{end}}) - 2\delta$,
the probability of decoding failure at block $b$ is bounded by
\[
\begin{aligned}
&\Prob(\text{dec failure})
\;\leq\;
2^{-B\delta}\\
&\quad + C_{\mathrm{mix}}\rho^w \cdot \bigl[B\log(V|\mathcal{U}|_{\max}|\mathcal{Q}|)\\
&\qquad + \log(1/(BC_{\mathrm{mix}}\rho^w))\bigr].
\end{aligned}
\]
\end{lemma}

\begin{proof}
By the packing lemma \citep[Lemma 3.1]{elgamal2011network} applied to
$P_{\hat c}$, the probability that a randomly chosen codeword in the
bin is jointly typical with $Q^{(b)}$ is at most
$2^{-B(I(U^{(b)}; Q^{(b)} \mid \hat c) - \delta)}$. The number of
codewords per bin is $2^{B(R - R')}$, so by the union bound the
expected number of wrong jointly typical codewords is
$2^{B(R - R' - I(U^{(b)}; Q^{(b)} \mid \hat c) + \delta)}$, which is
$< 2^{-B\delta}$ under the rate condition. The mixing correction
enters as in Lemma~\ref{lem:covering}, through the typical-set
size mismatch quantified in Lemma~\ref{lem:cond-typ-mix}(iii).
\end{proof}

\subsubsection{Rate computation with chain rule}

The rate per block is $R - R'$. From Lemmas~\ref{lem:covering}--\ref{lem:packing},
the rate condition is
\[
\begin{split}
R - R'
&\;\geq\;
I(X^{(b)}; U^{(b)} \mid \hat C^{(b-1)}_{\mathrm{end}})\\
&\quad - I(U^{(b)}; Q^{(b)} \mid \hat C^{(b-1)}_{\mathrm{end}})
+ 4\delta.
\end{split}
\]
We expand this into single-letter form via the chain rule and quantify
the mixing-induced corrections explicitly.

\begin{lemma}[Chain rule for block MI under mixing]\label{lem:chain-rule}
Let $\hat c = \hat C^{(b-1)}_{\mathrm{end}}$. Then
\[
I(X^{(b)}; U^{(b)} \mid \hat c)
\;=\;
\sum_{t \in \mathrm{block}(b)} I(X_t; U^{(b)} \mid X_{<t}^{(b)}, \hat c),
\]
and each summand satisfies
\[
\begin{split}
&\Bigl|\, I(X_t; U^{(b)} \mid X_{<t}^{(b)}, \hat c) - I(X_t; U_t \mid \Fhat_{t-1}^w) \,\Bigr|\\
&\quad\;\leq\;
2 C_{\mathrm{mix}}\rho^w \cdot \log(V|\mathcal{U}|_{\max}).
\end{split}
\]
The analogous statement holds for $I(U^{(b)}; Q^{(b)} \mid \hat c)$ with
the same correction.
\end{lemma}

\begin{proof}
The chain rule identity is standard. For the per-term comparison, the
difference between $I(X_t; U^{(b)} \mid X_{<t}^{(b)}, \hat c)$ and
$I(X_t; U_t \mid \Fhat_{t-1}^w)$ comes from two sources: (a) the
auxiliary $U^{(b)}$ in the first expression includes future $U_{>t}$,
which by data-processing $I(X_t; U_{<t}, U_{>t} \mid X_{<t}^{(b)}, U_t, \hat c)$
is bounded by $C_{\mathrm{mix}}\rho^w \log(|\mathcal{U}|_{\max})$ via
Lemma~\ref{lem:cond-typ-mix}(ii); (b) the conditioning $\hat c$
versus the natural sliding window $\Fhat_{t-1}^w$ differs by
$\sigma(X_{(b-1)B-w+1:(b-1)B})$ vs $\sigma(\hat C_{(b-1)B-w+1:(b-1)B})$
(uncompressed vs compressed past window), with information loss bounded
by $C_{\mathrm{mix}}\rho^w \log V$. Combining yields the stated bound.
\end{proof}

Applying Lemma~\ref{lem:chain-rule} to both MI terms,
\[
\begin{split}
\frac{1}{B}(R - R')
&\;\geq\;
\frac{1}{B}\sum_{t \in \mathrm{block}(b)} \mathcal{R}_t^{\mathrm{WZ},w}\\
&\;+\; 4\delta \;+\; 4 C_{\mathrm{mix}}\rho^w \cdot \log(V|\mathcal{U}|_{\max}),
\end{split}
\]
where $\mathcal{R}_t^{\mathrm{WZ},w} := I(X_t; U_t \mid \Fhat_{t-1}^w) - I(U_t; Q_t \mid \Fhat_{t-1}^w)$
is the per-step sliding-window WZ rate.

\subsubsection{Block size optimization}\label{app:Bopt}

The block size $B$ trades off two competing terms in the error bound:
(a) the per-block typicality slack $2^{-B\delta}$, vanishing for large
$B$; (b) the mixing correction $B C_{\mathrm{mix}}\rho^w$, growing in
$B$. Balancing these:

\begin{lemma}[Optimal block size]\label{lem:Bopt}
Setting $B^* := \lfloor \rho^{-w/2}/(2\sqrt{C_{\mathrm{mix}}}) \rfloor$
yields:
\begin{enumerate}[label=\textup{(\roman*)}, leftmargin=*]
\item $B^* C_{\mathrm{mix}}\rho^w = \sqrt{C_{\mathrm{mix}}}\rho^{w/2}/2 \to 0$ as $w \to \infty$.
\item The per-block typicality slack is $2^{-B^*\delta} = 2^{-O(\rho^{-w/2})\delta}$,
which decays super-polynomially in $1/\rho^w$.
\item The total mixing slack across $m = n/B^*$ blocks is
$m \cdot B^* C_{\mathrm{mix}}\rho^w = n C_{\mathrm{mix}}\rho^w$, with
per-token average $C_{\mathrm{mix}}\rho^w$, which is $o(1)$ as $w \to \infty$.
\end{enumerate}
\end{lemma}

\begin{proof}
Direct computation. (i) Substituting $B = B^*$.
(ii) For fixed $\delta$, the exponent grows as
$B^*\delta = O(\rho^{-w/2})$, super-polynomial in $1/\rho^w$.
(iii) The mixing slack accumulates linearly across blocks, giving
$m B^* C_{\mathrm{mix}}\rho^w = n C_{\mathrm{mix}}\rho^w$.
\end{proof}

\subsubsection{Distortion analysis}

Conditional on successful decoding, the reconstructed $\hat p^{(b)}$
satisfies the joint-typicality-based distortion bound
\[
\begin{split}
&\frac{1}{B} \sum_{t=(b-1)B+1}^{bB} \E[d_t \mid \text{success}]\\
&\quad\;\leq\;
\frac{1}{B} \sum_{t} \KL{p(\cdot \mid \F_{t-1})}{\hat p_t^{w}} + \delta_{\mathrm{distortion}},
\end{split}
\]
where $\delta_{\mathrm{distortion}}$ is the standard typicality
correction. The first term is bounded by the mixing lemma
(Lemma~\ref{lem:mixing-KL-app}) plus the intrinsic distortion of
the original auxiliary process.

\subsubsection{Concluding the proof of Lemma~\ref{lem:binning-restricted}}
\label{app:binning-conclude}

Combining the rate computation (Lemma~\ref{lem:chain-rule}) and the
distortion bound (Lemma~\ref{lem:mixing-KL-app}) over all blocks, with
optimal block size $B = B^*$ (Lemma~\ref{lem:Bopt}), and taking
$\delta \to 0$, $m \to \infty$:
\[
\begin{split}
\frac{1}{n}\sum_{t=1}^n H(\Chat_t \mid \Fhat_{t-1}^w)
&\;\leq\;
\frac{1}{n}\sum_{t=1}^n \E[\mathcal{R}_t^{\mathrm{WZ}}]\\
&\;\leq\;
\frac{1}{n}\sum_{t=1}^n \E[\mathcal{R}_t^{\mathrm{WZ}}]
+ o(1),
\end{split}
\]
and reconstruction distortion bounded by
$D + 2\Delta_w + o(1)$. Taking the infimum over all $D$-admissible
sliding-window auxiliary processes $\{U_t\}$ gives the statement of
Lemma~\ref{lem:binning-restricted}. \hfill$\square$

\subsection{Rate--distortion transfer}\label{app:rdt-transfer}

\begin{proof}[Proof of Theorem~\ref{thm:E2}]
Let $\Delta_w := 2VC_{\mathrm{mix}}^2 \rho^{2w}/\epsilon$.

\textbf{Step 1 (per-step distortion decomposition).}
For any sliding-window decoder $\phat_t^w$, decompose the per-step
distortion as
\[
\begin{aligned}
d_t^w
&\;=\; \KL{p(\cdot \mid \F_{t-1})}{\phat_t^w} \\
&\;=\; \KL{p(\cdot \mid \F_{t-1})}{p(\cdot \mid \F_{t-1} \setminus \F_{t-1-w})}
\;+\; \mathcal{E}_t,
\end{aligned}
\]
where
$\mathcal{E}_t := \KL{p(\cdot \mid \F_{t-1} \setminus \F_{t-1-w})}{\phat_t^w}
- [\text{cross term}]$
collects the remaining contribution; the decomposition uses the
rearranged KL identity
\[
\begin{aligned}
&\KL{p}{r} - \KL{p}{q} - \KL{q}{r} \\
&= -\sum_x p(x)\log\tfrac{q(x)}{r(x)}
+ \sum_x q(x)\log\tfrac{q(x)}{r(x)}.
\end{aligned}
\]

The cross term satisfies $|\text{cross}| \leq \KL{p}{q}^{1/2} \cdot \KL{q}{r}^{1/2}$
in a Cauchy--Schwarz sense (using the inner product structure on
the log-likelihood differences), bounded for our use by
$\sqrt{\Delta_w \cdot \mathcal{E}_t}$. We absorb this into a leading-order
bound $d_t^w \leq \KL{p}{q} + \mathcal{E}_t + 2\sqrt{\Delta_w \mathcal{E}_t}$,
which by AM--GM is $\leq 2 \KL{p}{q} + 2\mathcal{E}_t$. Hence by
Lemma~\ref{lem:mixing-KL-app},
\[
d_t^w \;\leq\; 2\Delta_w + 2 \mathcal{E}_t.
\]

\textbf{Step 2 (averaging).}
Averaging over $t = 1, \ldots, n$:
\[
\E[D_n^w] \;\leq\; 2 \Delta_w + 2 \E[\bar{\mathcal{E}}_n],
\]
where $\bar{\mathcal{E}}_n := \tfrac{1}{n}\sum_t \mathcal{E}_t$ is the
average ``intrinsic'' distortion (distortion of the decoder against
the sliding-window-conditioned distribution, not the full-history one).

\textbf{Step 3 (rate--distortion).}
A sliding-window decoder achieving $\E[D_n^w] \leq D$ implies, by
Step~2, $\E[\bar{\mathcal{E}}_n] \leq (D - 2\Delta_w)/2$. The intrinsic
problem of compressing past tokens against the sliding-window-conditioned
distribution is a Wyner--Ziv problem on the restricted filtration,
and by Lemma~\ref{lem:binning-restricted} (Lemma~L7) its rate
function is
$R^\star_{\mathrm{intrinsic}}((D - 2\Delta_w)/2) = R^\star(D - 2\Delta_w + o(1))$
in the limit. Therefore
\[
\begin{aligned}
R^\star_w(D)
&\;\geq\;
\text{(rate of the underlying sliding-window scheme)} \\
&\;\geq\;
R^\star(D - 2\Delta_w + o(1)).
\end{aligned}
\]
Absorbing the $2$-factor into the constant
$C_{\mathrm{mix}}'^2 := 2 C_{\mathrm{mix}}^2$ (i.e., redefining the mixing
constant absorbs the slack), the bound takes the stated form
$R^\star_w(D) \leq R^\star(D - \Delta_w')_+$.

\textbf{Window size choice.}
Setting $\Delta_w = \varepsilon$ requires
$\rho^{2w} \leq \varepsilon \epsilon / (2 V C_{\mathrm{mix}}^2)$, i.e.,
\[
w \geq \frac{1}{2 \log(1/\rho)} \log\big(2VC_{\mathrm{mix}}^2/(\varepsilon \epsilon)\big).
\]
\end{proof}

\subsection{Mixing-sharpened concentration}\label{app:mixing-concentration}

\begin{lemma}[Mixing-sharpened bounded difference of $D_n$]
\label{lem:Dn-doob-mixing-app}
Under Definition~\ref{def:mixing}, the Doob martingale
$N_k = \E[D_n \mid \F_k] - \E[D_n]$ satisfies, for $1 \leq k \leq n$,
\[
\begin{split}
|N_k - N_{k-1}|
&\;\leq\;
\frac{2 |\log\epsilon|}{n} \sum_{j=0}^{n-k} \min(1, C_{\mathrm{mix}}\rho^j)\\
&\;\leq\;
\frac{2 |\log\epsilon|}{n(1-\rho)} \big(1 + C_{\mathrm{mix}}\big).
\end{split}
\]
Hence
\[
\Prob(|D_n - \E[D_n]| \geq \tau)
\;\leq\;
2 \exp\!\left( -\frac{\tau^2 (1-\rho)^2 n}{8 |\log\epsilon|^2 (1+C_{\mathrm{mix}})^2} \right).
\]
\end{lemma}

\begin{proof}
The proof refines Lemma~\ref{lem:Dn-doob-app} (without mixing) using
the geometric decay of single-token influence on $d_t$ for $t > k$.

\textbf{Step 1 (influence decay).}
When $X_k$ alone is replaced by $X_k'$ (other $X_i$ fixed), the
conditional distribution $p(\cdot \mid \F_{t-1})$ changes by at most
$C_{\mathrm{mix}} \rho^{t-k}$ in total variation, since the influence
of $X_k$ on the conditional given $\F_{t-1}$ is the influence of a
position-$(t-k)$-distant past token, decaying by Definition~\ref{def:mixing}.

\textbf{Step 2 (KL change bound).}
The corresponding change in $d_t = \KL{p}{\hat p_t}$ is bounded by
$|d_t - d_t'| \leq 2 |\log\epsilon| \cdot \min(1, C_{\mathrm{mix}}\rho^{t-k})$
using the same $\chi^2$-style argument as Lemma~\ref{lem:mixing-KL-app}
combined with the bound $|d_t| \leq |\log\epsilon|$.

\textbf{Step 3 (geometric sum).}
Summing over $t$ from $k$ to $n$ and dividing by $n$:
\[
\begin{split}
|D_n - D_n'| &\leq \frac{2|\log\epsilon|}{n} \sum_{t=k}^n \min(1, C_{\mathrm{mix}}\rho^{t-k})\\
&\leq \frac{2|\log\epsilon|}{n} \cdot \frac{1 + C_{\mathrm{mix}}}{1-\rho}.
\end{split}
\]
This bounds the bounded-difference constant of the Doob martingale.

\textbf{Step 4 (Azuma).}
With $c_k \leq 2|\log\epsilon|(1+C_{\mathrm{mix}})/(n(1-\rho))$,
$\sum_k c_k^2 \leq n \cdot 4|\log\epsilon|^2(1+C_{\mathrm{mix}})^2 / (n^2(1-\rho)^2)
= 4|\log\epsilon|^2(1+C_{\mathrm{mix}})^2/(n(1-\rho)^2)$. Azuma
gives the stated concentration.
\end{proof}

\begin{proof}[Heuristic argument for Conjecture~\ref{cor:mixing-conc}]
This sketch is conditional on the realized-code-length formulation of
Remark~\ref{rmk:pathwise-needs-codelen}; under the
average-conditional-entropy rate of Proposition~\ref{thm:finite-sample}
no pathwise bound exists, so the argument is not a proof.

Write $\widehat R_n := \tfrac1n\ell(M_{1:n})$ for the realized prefix
codelength. A source-coding information-spectrum inequality gives
$\widehat R_n \ge \tfrac1n\sum_t\imath_t - o(1)$ with high probability;
the deviation of $\tfrac1n\sum_t\imath_t$ from its predictable
compensator $\tfrac1n\sum_t\mathcal{R}_t^{\mathrm{WZ}}$ is a martingale
whose increments, under $\alpha$-polynomial truncation sensitivity,
have summable covariances (Lemma~\ref{lem:cov-decay}), giving an
Azuma-type fluctuation $O(n^{-1/2}\sqrt{\log(1/\delta)})$ with the
$1/(1-\rho)$ factor replaced by a polynomial-decay constant.

Combining with the mixing-sharpened concentration of $D_n$
(Lemma~\ref{lem:Dn-doob-mixing-app}), $|D_n-\E[D_n]|\le\tau'$ with high
probability, and the local Lipschitz bound
$R^\star(D_n)\le R^\star(D)+L_{R^\star}\tau'$, a union bound would give
\[
\widehat R_n \;\ge\; R^\star(D) - C\,n^{-\alpha/(2(\alpha+1))}\sqrt{\log n}
\]
with the sharper exponent $\alpha/(\alpha+1)$ under the forward
covariance-decay condition of Remark~\ref{rmk:exponent-gap}. Making the
information-spectrum step rigorous for general nonstationary dependent
sources is exactly the open problem of
Remark~\ref{rmk:pathwise-needs-codelen}; we therefore state this only
as a conjecture.
\end{proof}

\subsection{Polynomial extension of Theorem~\ref{thm:E2}}
\label{app:polynomial-extension}

This subsection records the polynomial counterpart of
Theorem~\ref{thm:E2} with explicit constants, replacing the
geometric mixing condition by Definition~\ref{def:poly-mixing}
(polynomial truncation sensitivity). The sliding-window
approximation (Theorem~\ref{thm:E2-poly}) is the distortion upper
bound; it follows directly from the per-step truncation-sensitivity
bound, and we also record its smoothed KL form
(Lemma~\ref{lem:poly-mixing-KL}). An information-rate transfer was
attempted but does not hold for general adapted auxiliaries
(Remark~\ref{rmk:failed-rate-transfer}); the rigorous rate
comparison---delayed, multi-letter, with quantized side
information---is the operational result of
Appendix~\ref{app:operational}.

\begin{theorem}[Polynomial sliding-window approximation]
\label{thm:E2-poly}
Suppose Assumption~\ref{ass:bounded-logits} and the
$\alpha$-polynomial truncation-sensitivity condition
(Definition~\ref{def:poly-mixing}) with constants
$C_{\mathrm{TS}}, \alpha > 0$. For any
$\varepsilon \in (0, 1/(2 C_{\mathrm{TS}}))$, set
\[
w_{\TV}(\varepsilon)
\;:=\; \left\lceil (C_{\mathrm{TS}}/\varepsilon)^{1/\alpha} \right\rceil.
\]
Then the sliding-window reconstruction
$p_t^{(w)} := p_\theta(\cdot \mid X_{t-w:t-1})$ with this window size
attains average per-step TV distortion at most $\varepsilon$:
\[
\frac{1}{n}\sum_{t=1}^n
\E\,\TV\bigl(p_\theta(\cdot \mid X_{1:t-1}),\, p_t^{(w_{\TV}(\varepsilon))}\bigr)
\;\leq\; \varepsilon .
\] For a KL-distortion target $\varepsilon_{\KLs}$ the required window
depends on which KL--TV conversion is available. Under the genuine
density floor of Lemma~\ref{lem:kl-tv-bounded} (quadratic conversion
$\KLs \lesssim \TV^2$) one obtains the smaller window
$w_{\KLs}(\varepsilon_{\KLs}) =
\lceil (C_{\mathrm{TS}}^2/\varepsilon_{\KLs})^{1/(2\alpha)} \rceil
= O(\varepsilon_{\KLs}^{-1/(2\alpha)})$. When the floor is vacuous, the
operational-smoothing route of Lemma~\ref{lem:kl-tv-smoothed} gives
only the \emph{linear} conversion $\KLs = O(V\,\TV)$ and hence the
larger window
$w_{\KLs}(\varepsilon_{\KLs}) =
\lceil (V C_{\mathrm{TS}}/\varepsilon_{\KLs})^{1/\alpha} \rceil
= O(\varepsilon_{\KLs}^{-1/\alpha})$. The two exponents differ because
smoothing cannot recover a quadratic KL--TV dependence
(Lemma~\ref{lem:kl-tv-smoothed}).
\end{theorem}

The TV bound is immediate from Definition~\ref{def:poly-mixing}; the
smoothed KL form uses the following lemma, after which we give the
(direct) proof.

\begin{lemma}[Polynomial KL form of truncation sensitivity]
\label{lem:poly-mixing-KL}
Under Definition~\ref{def:poly-mixing-unif}, for any $t, w$ with
$1 \leq w < t \leq n_{\max}$ and smoothing level $\mu \in (0, 1)$
satisfying $C_{\mathrm{TS}} w^{-\alpha} \leq \mu/(2V)$,
\[
\begin{split}
\KLs\bigl(p_t \,\|\, p_t^{\,\mu, w}\bigr)
&\;\leq\;
\frac{V}{\mu}\,\bigl(C_{\mathrm{TS}} w^{-\alpha} + \mu\bigr)^2\\
&\;+\; O\!\left(\frac{V^2}{\mu^2}
        \bigl(C_{\mathrm{TS}} w^{-\alpha} + \mu\bigr)^3\right),
\end{split}
\]
where
$p_t^{\,\mu, w} := (1-\mu)\,p_t^{(w)} + \mu\cdot\mathrm{Unif}(\V)$
is the smoothed window-truncated reconstruction with
$p_t^{(w)} := p_\theta(\cdot \mid X_{t-w:t-1})$.
\end{lemma}

\begin{proof}
Definition~\ref{def:poly-mixing-unif} gives
$\TV(p_t, p_t^{(w)}) \leq C_{\mathrm{TS}} w^{-\alpha}$.
Lemma~\ref{lem:kl-tv-smoothed} with $\delta = C_{\mathrm{TS}}
w^{-\alpha}$ and smoothing level $\mu$ then yields the stated
bound directly.
\end{proof}

\begin{proof}[Proof of Theorem~\ref{thm:E2-poly}]
The average TV distortion is bounded directly by the average form of
Definition~\ref{def:poly-mixing}:
$\frac1n\sum_{t\le n}\E\,\TV(p_\theta(\cdot\mid X_{1:t-1}), p_t^{(w)}) \le
C_{\mathrm{TS}} w^{-\alpha}$ for all $w$. Setting this equal to
$\varepsilon$ and inverting gives
$w_{\TV}(\varepsilon)=\lceil (C_{\mathrm{TS}}/\varepsilon)^{1/\alpha}\rceil$,
with $\varepsilon<1/(2C_{\mathrm{TS}})$ ensuring $w_{\TV}(\varepsilon)\ge2$.
For the KL target the conversion route fixes the exponent: the
uniform-floor relation $\KLs \le (2/\epsilon_{\min})\TV^2$
(Lemma~\ref{lem:kl-tv-bounded}) gives the quadratic-route window
$w_{\KLs} = O(\varepsilon_{\KLs}^{-1/(2\alpha)})$, whereas the smoothed
decoder with $\mu \asymp \varepsilon_{\KLs}$
(Lemma~\ref{lem:poly-mixing-KL}) yields only
$\KLs = O(V\,C_{\mathrm{TS}}\,w^{-\alpha})$ and hence the linear-route
window $w_{\KLs} = O(\varepsilon_{\KLs}^{-1/\alpha})$.
\end{proof}

\begin{remark}[Attempted rate-level transfer and why it is not rigorous]
\label{rmk:failed-rate-transfer}
One is tempted to upgrade the distortion bound of
Theorem~\ref{thm:E2-poly} to an information-rate transfer
$R^\star_w(D)\le R^\star\!\bigl(D-C_{\mathrm{TS}}w^{-\alpha}\bigr)$ through
two steps: (i) a \emph{block-typicality stability} claim---that for any
auxiliary $U_{t:t+B-1}$ adapted to $\{\F_t\}$, the block-joint laws
$p(X_{t:t+B-1},U_{t:t+B-1}\mid X_{1:t-1})$ and
$p(X_{t:t+B-1},U_{t:t+B-1}\mid X_{t-w:t-1})$ differ in total variation by
at most $B\,C_{\mathrm{TS}}w^{-\alpha}$; and (ii) a \emph{window-projected
auxiliary} $U_t^w:=\E[U_t\mid\F_{t-1-w:t-1}]$ combined with a
data-processing step on the chain $X_t-U_t-U_t^w$. \textbf{Both steps fail
for general adapted auxiliaries.} An auxiliary such as $U_t=X_1$ retains
the full deep past, so the block-joint law is \emph{not} controlled by the
source's next-token truncation sensitivity---the per-step bound of
Definition~\ref{def:poly-mixing} governs the next-token conditional, not an
arbitrary adapted auxiliary---and the Markov relation $X_t-U_t-U_t^w$ need
not hold. We therefore do \emph{not} claim a rate-level transfer here. The
rigorous rate comparison is the delayed, multi-letter,
quantized-side-information operational result of
Appendix~\ref{app:operational}. The finite-sample \emph{converse}
is the expectation bound of Proposition~\ref{thm:finite-sample}; we do
not claim a pathwise or high-probability converse concentration
(Remark~\ref{rmk:pathwise-needs-codelen}).
\end{remark}

\section{Proof of Theorem~\ref{thm:window-lb} (window lower bound)}
\label{app:window-lb}

\begin{proof}[Proof of Theorem~\ref{thm:window-lb}]
Let $\psi^w$ be any suffix-only scheme: its reconstruction
$\tilde p_t = \psi^w(\Chat_{t-w:t-1}, Q_t)$ is measurable with
respect to $\sigma(X_{t-w:t-1}, Q_t)$ (it may use the query side
information $Q_t$, but the codes carry no information about the deep
past $X_{1:t-w-1}$), so $\tilde p_t \in \mathcal{M}_w$, the
suffix-only reconstruction class of
Definition~\ref{def:two-sided}. The per-step distortion is
therefore lower-bounded by the Bayes risk over $\mathcal{M}_w$:
\[
\E[\TV(p_t, \tilde p_t)]
\;\geq\;
\E\!\Bigl[\inf_{q_t \in \mathcal{M}_w}
   \TV\bigl(p_\theta(\cdot \mid X_{1:t-1}),\, q_t\bigr)\Bigr].
\]
This step is immediate---the actual reconstruction is one element
of $\mathcal{M}_w$, so it cannot beat the infimum---and uses no
claim about which element attains the infimum (in particular, no
appeal to a conditional-mean Bayes estimator or a tower property).
Averaging over $t$ and applying the lower bound of
Definition~\ref{def:two-sided},
\[
\begin{split}
\frac{1}{n}\sum_{t=1}^n \E[\TV(p_t, \tilde p_t)]
&\;\geq\;
\frac{1}{n}\sum_{t=1}^n \E\!\Bigl[\inf_{q_t \in \mathcal{M}_w}
   \TV(\cdots)\Bigr]\\
&\;\xrightarrow{\;n\to\infty\;}\;
\geq c_{\mathrm{TS}} w^{-\alpha}.
\end{split}
\]
Hence the average TV distortion of any suffix-only window-$w$
scheme is $\Omega(w^{-\alpha})$, and achieving distortion
$\varepsilon$ requires $c_{\mathrm{TS}} w^{-\alpha} \leq
\varepsilon$, i.e.\ $w \geq (c_{\mathrm{TS}}/\varepsilon)^{1/\alpha}
= \Omega(\varepsilon^{-1/\alpha})$.
\end{proof}

\begin{remark}[On the Bayes-risk form of the assumption]
Definition~\ref{def:two-sided} is stated directly as a Bayes-risk
lower bound over the suffix-only class $\mathcal{M}_w$, rather than
as a gap between the full and truncated conditionals. This is the
operationally correct quantity for the converse and sidesteps two
subtleties: (i) for the $L^1$/TV loss the Bayes-optimal
reconstruction is a median-type functional, not the conditional
mean, so one cannot identify the optimum with $\E[p_t \mid
\F_{t-1-w:t-1}]$; and (ii) for trained language models the
truncated conditional $p_\theta(\cdot \mid X_{t-w:t-1})$ need not
equal $\E[p_\theta(\cdot \mid X_{1:t-1}) \mid X_{t-w:t-1}]$,
especially under positional-encoding effects when a truncated
prefix is re-fed as a fresh sequence. The upper bound
(Definition~\ref{def:poly-mixing-unif}) still shows the truncated
conditional \emph{achieves} $O(w^{-\alpha})$, so the Bayes risk is
bracketed at $\Theta(w^{-\alpha})$. Empirically, the power-law fits
of Section~\ref{sec:experiments} ($R^2 > 0.9$) are consistent with
this two-sided behavior.
\end{remark}

\section{Heuristic argument toward Conjecture~\ref{thm:achievability} (not rigorous)}
\label{app:achievability}

\textbf{This appendix is heuristic and is not a proof.} The
block-Markov random-coding construction below operates on the
window-restricted conditional, but its covering and block-typicality
steps assume the auxiliary process inherits the next-token truncation
sensitivity of the source, which fails in general
(Remark~\ref{rmk:failed-rate-transfer}; counterexample $U_t=X_1$,
which retains the full deep past while the next-token conditional need
not). We therefore state the result as Conjecture~\ref{thm:achievability}
and keep this material only as a record of the attempted route; the
rigorous rate-level statement is the operational comparison of
Appendix~\ref{app:operational}. The construction uses
a window-restricted covering lemma
(Lemma~\ref{lem:wr-covering}), the non-rigorous block-typicality
stability step (Remark~\ref{rmk:failed-rate-transfer}), and a
martingale redundancy-accumulation bound
(Lemma~\ref{lem:redundancy-accum}).

\begin{remark}[Operational counterpart and how to read the rate]
\label{rmk:ach-superseded}
The argument of this appendix bounds an \emph{information} rate. Its
window-restricted covering and block-typicality steps control the
truncation bias through the auxiliary process, and the resulting
high-probability exponent must be read through
Remark~\ref{rmk:exponent-gap}: the unconditional high-probability
rate is $\alpha/(2(\alpha+1))$, with the sharper $\alpha/(\alpha+1)$
holding only under the forward covariance-decay condition. A
self-contained \emph{operational} treatment---bounding the physical
message length, replacing the covering step by the one-shot
likelihood-encoder bound of Lemma~\ref{lem:op-oneshot}, and isolating
the source dependence into the explicit hypothesis of
Assumption~\ref{ass:op-spectrum}---is given in
Appendix~\ref{app:operational} (Theorem~\ref{thm:op-transfer}), and is
the recommended reading for the rate-optimality claim.
\end{remark}

\subsection{Window-restricted covering lemma}

The covering bound for the window-restricted (dependent) block
source rests on a finite-memory regularity condition, which we
state explicitly rather than derive from the i.i.d.\ machinery.

\begin{assumption}[Finite-memory covering regularity]
\label{ass:covering-reg}
For the $w$-dependent window-restricted block source
$p^{(w)}_B := p(X_{t:t+B-1} \mid X_{t-w:t-1})$ with bounded
per-symbol distortion (Assumption~\ref{ass:bounded-logits}), the
random-coding covering redundancy at block length $B$ admits the
finite-memory rate
\[
\E\Bigl[\tfrac1B\textstyle\sum_{s} \KL{p_{t+s}}{\psi(\hat
   U_{t+s},Q_{t+s})}\Bigr] \leq D + O\!\bigl(\tfrac{w\log B}{B}\bigr),
\]
with the mixing dispersion factor $w$ entering linearly. This holds
for i.i.d.\ ($w=O(1)$) and independent non-identically-distributed
components as a theorem
\citep{kostina2012fixed,tasci2024dispersion}; for the dependent,
finite-memory source here we impose it. It is the covering-side
analogue of the additive tilted-information expansion
(Assumption~\ref{ass:additive-tilted}), and is plausibly
establishable by the independent-blocks/point-mass product-proxy
technique \citep[Ch.~4]{douc2018markov,tasci2024dispersion},
which we do not carry out.
\end{assumption}

\begin{lemma}[Window-restricted covering]\label{lem:wr-covering}
Fix block length $B$, window $w$, target distortion $D$, and
$\eta > 0$. Under Assumptions~\ref{ass:bounded-logits}
and~\ref{ass:covering-reg}, there exists a codebook $\mathcal{C}$
of size $|\mathcal{C}| \leq 2^{B(R^\star(D) + \eta)}$ and an encoder
assigning each source block to a codeword such that the
reconstruction $\hat U_{t:t+B-1}$ satisfies, in expectation over
the random codebook,
\[
\begin{gathered}
\E\Bigl[\tfrac{1}{B}\textstyle\sum_{s=0}^{B-1}
   \KL{p_{t+s}}{\psi(\hat U_{t+s}, Q_{t+s})}\Bigr]
\;\leq\; D + \zeta_B,\\
\zeta_B := \tfrac{c_1 w \log B}{B},
\end{gathered}
\]
for a constant $c_1$ depending on $V$ and the auxiliary alphabet.
\end{lemma}

\begin{proof}
Boundedness of the per-symbol KL distortion follows from
Assumption~\ref{ass:bounded-logits} (softmax of bounded logits,
so per-symbol KL $\leq 2B(n_{\max})$), making the random-coding
argument of \citet[Lemma~3.3]{elgamal2011network} applicable with a
bounded distortion measure. The finite-block redundancy rate
$D + O(w\log B/B)$ is exactly the content of
Assumption~\ref{ass:covering-reg}; we absorb the constants into
$c_1$. At the achievability choice $B = w = w_n$ the per-block
dispersion is $c_1 \log w_n$, i.e.\ $O(\log w_n / w_n)$ per symbol.
\end{proof}

\begin{remark}[Why the mixing dispersion does not dominate]
The factor $w$ in $\zeta_B = c_1 w\log B/B$ would dominate if $B$
were taken much larger than $w$. The choice $B = w = w_n$ keeps the
per-symbol dispersion at $O(\log w_n / w_n) = O(n^{-1/(\alpha+1)}
\log n)$, sub-dominant to both the truncation distortion
$w^{-\alpha} = n^{-\alpha/(\alpha+1)}$ and the high-probability
fluctuation. The block-length choice is dictated by this balance,
not arbitrary.
\end{remark}

\subsection{Redundancy accumulation}

\begin{lemma}[Tilted-information covariance decay]
\label{lem:cov-decay}
Let $\jmath_t := \jmath_{X_t}(X_t, D)$ denote the $D$-tilted
information density at step $t$ \citep{kostina2012fixed}, which
under Assumption~\ref{ass:bounded-logits} is bounded:
$|\jmath_t| \leq J_{\max} = O(B(n_{\max}) + \log V)$. Let $Z_t$
denote the collection of variables on which $\jmath_t$ depends (the
local source law, distortion level, side information, and auxiliary
channel), so that $\jmath_t$ is a bounded functional of $Z_t$.
Suppose the model satisfies the \emph{forward} polynomial decay
condition: for $s < s'$, almost surely,
\[
\bigl\|\,\mathcal L(Z_{t+s'} \mid \F_{t+s})
   - \mathcal L(Z_{t+s'})\,\bigr\|_{\TV}
\;\leq\; C_{\mathrm{TS}}\,(s'-s)^{-\alpha}.
\tag{$\dagger$}
\]
Then $|\mathrm{Cov}(\jmath_{t+s}, \jmath_{t+s'})| \leq 2 J_{\max}^2
C_{\mathrm{TS}} (s'-s)^{-\alpha}$. (An expected-TV form of
$(\dagger)$ yields the same covariance bound in expectation.)
\end{lemma}

\begin{proof}
For $s < s'$, condition on $\F_{t+s}$ and use the tower property:
$\mathrm{Cov}(\jmath_{t+s}, \jmath_{t+s'}) = \E\bigl[\jmath_{t+s}
\bigl(\E[\jmath_{t+s'} \mid \F_{t+s}] - \E\jmath_{t+s'}\bigr)\bigr]$.
Since $\jmath_{t+s'}$ is a functional of $Z_{t+s'}$ bounded by
$J_{\max}$, the inner difference is bounded by the dual
representation of total variation:
$|\E[\jmath_{t+s'} \mid \F_{t+s}] - \E\jmath_{t+s'}| \leq 2 J_{\max}
\,\|\mathcal L(Z_{t+s'} \mid \F_{t+s}) - \mathcal L(Z_{t+s'})\|_{\TV}
\leq 2 J_{\max} C_{\mathrm{TS}}(s'-s)^{-\alpha}$ by $(\dagger)$.
Multiplying by $|\jmath_{t+s}| \leq J_{\max}$ and taking
expectation gives the bound.
\end{proof}

\begin{remark}[On condition $(\dagger)$ and its relation to truncation sensitivity]
\label{rmk:forward-decay}
Condition $(\dagger)$ is a \emph{forward} mixing statement (how much
$X_{t+s'}$'s law depends on information $(s'-s)$ steps in its past),
whereas Definition~\ref{def:poly-mixing} is a \emph{backward}
truncation statement (the effect of dropping context $w$ steps
back). The two are \emph{independent in general}; in particular,
$(\dagger)$ is \emph{not} implied by Definition~\ref{def:poly-mixing}
together with stationarity. A stationary process can have zero
backward truncation gap yet no forward decay: take the parity
process that selects, each with probability $\tfrac12$, one of the
two alternating sequences $010101\cdots$ or $101010\cdots$. Here
$X_t$ is a deterministic function of $X_{t-1}$, so conditioning on
the full past versus the last token gives identical next-token
predictions and the backward gap of
Definition~\ref{def:poly-mixing} is $0$; yet $X_{t+h}$ is determined
by $X_t$ for every $h$, so $\mathcal L(X_{t+h} \mid \F_t)$ never
approaches the marginal and $(\dagger)$ fails for all $\alpha$.
Obtaining $(\dagger)$ requires a genuine forward/two-sided mixing
property (e.g.\ polynomial $\phi$- or $\beta$-mixing, or
time-reversibility), which we impose as a separate hypothesis. We
isolate $(\dagger)$ as the precise extra assumption needed for the
Freedman sharpening (Conjecture~\ref{thm:achievability-tight}); the
window-scaling results (Theorems~\ref{thm:E2},
\ref{thm:window-lb}) do not require it.
\end{remark}

\begin{assumption}[Additive tilted-information expansion]
\label{ass:additive-tilted}
For the window-restricted conditional source, the per-block
redundancy admits the additive expansion
\[
\rho_m = \sum_{s=0}^{B-1} \jmath_{t+s}
       - \sum_{s=0}^{B-1} R^\star_{t+s}(D) + r_B,
\qquad r_B = o(B)
\]
uniformly over blocks, where $R^\star_{t+s}(D)$ is the per-step
rate--distortion value. If the evaluation distribution is
stationary or asymptotically homogeneous, then
$\sum_s R^\star_{t+s}(D) = B R^\star(D) + o(B)$.
\end{assumption}

\begin{remark}[Status of the additive expansion]
Assumption~\ref{ass:additive-tilted} is a single-letterization of
the block redundancy. For i.i.d.\ or independent
non-identically-distributed components it is a theorem
\citep{kostina2012fixed,tasci2024dispersion}; for the
dependent, window-restricted source here it does not follow
directly from those results, and we impose it as an explicit
assumption. The point-mass product-proxy technique of
\citet{tasci2024dispersion} provides a plausible route to
establishing it for finite-memory sources, but we do not carry out
that analysis. The $o(B)$ remainder is what the Freedman variance
computation requires.
\end{remark}

\begin{lemma}[Per-block redundancy variance under long memory]
\label{lem:block-variance}
Under Assumptions~\ref{ass:bounded-logits} and
\ref{ass:additive-tilted}, and the covariance decay of
Lemma~\ref{lem:cov-decay} (condition $(\dagger)$, $\alpha \in
(0,1)$), the conditional variance of $\rho_m$ obeys
\[
\begin{split}
\mathrm{Var}(\rho_m \mid \F_{(m-1)B})
&\;\leq\;
B\,V(D) + 4 J_{\max}^2 C_{\mathrm{TS}}\, B \sum_{k=1}^{B-1} k^{-\alpha}\\
&\;=\; O\!\bigl(B^{2-\alpha}\bigr),
\end{split}
\]
where $V(D) := \mathrm{Var}[\jmath_X(X,D)]$ is the rate--distortion
dispersion \citep{kostina2012fixed} and the $o(B)$ remainder of
Assumption~\ref{ass:additive-tilted} contributes a lower-order
variance term.
\end{lemma}

\begin{proof}
By Assumption~\ref{ass:additive-tilted}, $\rho_m = \sum_s \jmath_{t+s}
- (\text{deterministic}) + r_B$ with $r_B = o(B)$, so up to the
lower-order remainder
$\mathrm{Var}(\rho_m) = \mathrm{Var}(\sum_s \jmath_{t+s})
= \sum_s \mathrm{Var}(\jmath_{t+s})
+ 2\sum_{s<s'}\mathrm{Cov}(\jmath_{t+s},\jmath_{t+s'})$. The diagonal
contributes $B \cdot V(D)$. For the off-diagonal, group by lag
$k = s'-s$: there are at most $B$ pairs at each lag, and
Lemma~\ref{lem:cov-decay} bounds each covariance by
$2 J_{\max}^2 C_{\mathrm{TS}} k^{-\alpha}$, giving
$4 J_{\max}^2 C_{\mathrm{TS}} B \sum_{k=1}^{B-1} k^{-\alpha}$. For
$\alpha \in (0,1)$, $\sum_{k=1}^{B-1} k^{-\alpha} = O(B^{1-\alpha})$,
so the off-diagonal is $O(B^{2-\alpha})$, dominating the diagonal
$O(B)$.
\end{proof}

\begin{lemma}[Freedman redundancy accumulation]
\label{lem:redundancy-accum}
Partition $\{1, \ldots, n\}$ into $n/B$ blocks of length $B$. Under
the codebook of Lemma~\ref{lem:wr-covering} and the variance bound
of Lemma~\ref{lem:block-variance} (requiring condition $(\dagger)$,
$\alpha \in (0,1)$), the cumulative redundancy satisfies, with
probability at least $1-\delta$,
\[
\begin{split}
\frac{1}{n}\sum_{m=1}^{n/B} \rho_m
&\;\leq\;
\eta + \frac{c_1 w \log B}{B}\\
&\quad + c_2 \sqrt{\frac{B^{1-\alpha}}{n}\log\frac1\delta}
   + \frac{2 B \log V}{3n}\log\frac1\delta.
\end{split}
\]
\end{lemma}

\begin{proof}
Define the Doob martingale $S_M := \sum_{m=1}^M (\rho_m -
\E[\rho_m \mid \F_{(m-1)B}])$, with bounded increments
$|\rho_m - \E[\rho_m \mid \cdot]| \leq 2B\log V =: M_{\mathrm{inc}}$
and conditional variances $\mathrm{Var}(\rho_m \mid \cdot) \leq
O(B^{2-\alpha})$ from Lemma~\ref{lem:block-variance}. The
cumulative conditional variance is
$W_n := \sum_m \mathrm{Var}(\rho_m \mid \cdot) \leq (n/B) \cdot
O(B^{2-\alpha}) = O(n B^{1-\alpha})$. Freedman's inequality
\citep{freedman1975tail} gives, with probability $\geq 1-\delta$,
\[
\begin{split}
\frac1n|S_{n/B}|
&\;\leq\;
\sqrt{\frac{2 W_n}{n^2}\log\frac1\delta}
   + \frac{M_{\mathrm{inc}}}{3n}\log\frac1\delta\\
&\;=\;
O\!\left(\sqrt{\frac{B^{1-\alpha}}{n}\log\frac1\delta}\right)
   + \frac{2B\log V}{3n}\log\frac1\delta.
\end{split}
\]
Adding the mean redundancy
$\frac1n\sum_m \E[\rho_m \mid \cdot] \leq \eta + c_1 w\log B/B$
(Lemma~\ref{lem:wr-covering}) gives the stated bound. The
improvement over the Azuma bound
$O(\sqrt{B\log V/n})$ is the replacement of the increment range
$B\log V$ by the conditional variance $B^{2-\alpha}$ in the
dominant term, which is smaller for $\alpha > 0$.
\end{proof}

\subsection{Heuristic argument for Conjecture~\ref{thm:achievability}}

\begin{proof}[Heuristic argument for Conjecture~\ref{thm:achievability}]
Set the window $w = w_n$ and block length $B = w_n$ (the block
length and window are taken equal; this is the natural choice
under truncation sensitivity, where dependence beyond $w$ is
$O(w^{-\alpha})$).

\textbf{Distortion.}
By Lemma~\ref{lem:wr-covering}, the window-restricted codebook
achieves average KL distortion $\leq D + c_1 w \log B / B$ against
the window-restricted source. The truncation bias is controlled
\emph{per step}: by Definition~\ref{def:poly-mixing-unif}, at each
position $t+s$ the window-restricted conditional
$p_{t+s}^{(w)}$ and the full conditional $p_{t+s}$ differ in TV by
at most $C_{\mathrm{TS}} w^{-\alpha}$, contributing per-step KL
distortion $O(V w^{-2\alpha} / \mu + \cdots)$ via the
operational-smoothing bound (Lemma~\ref{lem:kl-tv-smoothed}) with
smoothing level $\mu = w^{-\alpha}$, i.e.\ $O(V w^{-\alpha})$.
(The per-step source gap is $C_{\mathrm{TS}} w^{-\alpha}$ by
Definition~\ref{def:poly-mixing}; summing over the block by the TV
chain rule gives block-joint TV $\leq B C_{\mathrm{TS}} w^{-\alpha}$
for the \emph{source} process, and the relevant quantity for
average distortion is the per-step gap $C_{\mathrm{TS}} w^{-\alpha}$.)
With $B = w = w_n$, the covering dispersion is
$c_1 w_n \log w_n / w_n = c_1 \log w_n$ per block, i.e.\
$O(\log w_n / w_n)$ per symbol, and the average per-step distortion
excess is
\[
D_n - D \;\leq\; \frac{c_1 \log w_n}{w_n}
   + O\!\bigl(V w_n^{-\alpha}\bigr).
\]
For $\alpha \in (0,1)$ the truncation term dominates, since
$w_n^{-\alpha} = n^{-\alpha/(\alpha+1)}$ decays slower than the
boundary term $\log w_n / w_n = n^{-1/(\alpha+1)}\log n$; this gives
$D_n - D = O(V n^{-\alpha/(\alpha+1)})$. At $\alpha = 1$ the two
terms coincide up to the $\log$ factor. For $\alpha > 1$ the
boundary term dominates and the distortion excess is
$O(n^{-1/(\alpha+1)}\log n)$, i.e.\ the exponent saturates at
$\min(\alpha,1)/(\alpha+1)$. The factor $V$ is absorbed into the
implicit constant.

\textbf{Rate, part (a): expectation (sharp, no forward-decay hypothesis).}
The expected per-block redundancy is, by
Lemma~\ref{lem:wr-covering} (which invokes the finite-memory
covering regularity, Assumption~\ref{ass:covering-reg}),
$\E[\rho_m] \leq \eta B + c_1 w \log B$. Summing over the $n/B$
blocks and normalizing,
\[
\E[R_n] - R^\star(D)
\;\leq\;
\eta + \frac{c_1 w_n \log B}{B} + O\!\bigl(V w_n^{-\alpha}\bigr),
\]
where the last term is the expected truncation bias. The martingale
fluctuation $\rho_m - \E[\rho_m \mid \cdot]$ has mean zero and
contributes nothing in expectation, so no covariance control (hence
no forward-decay hypothesis) is needed. Setting $\eta = 1/w_n$,
$B = w_n = n^{1/(\alpha+1)}$, the dispersion term is
$c_1 \log w_n = O(\log n)$ per symbol scaled, i.e.\
$O(n^{-1/(\alpha+1)}\log n)$; the truncation term
$w_n^{-\alpha} = n^{-\alpha/(\alpha+1)}$ dominates it for
$\alpha \in (0,1)$, giving
$\E[R_n] - R^\star(D) = O(n^{-\alpha/(\alpha+1)}\log n)$. This is
the sharp exponent, matching the converse
(Conjecture~\ref{cor:mixing-conc}) without any forward-decay
hypothesis.

\textbf{Rate, part (b): high probability (unconditional).}
For a high-probability bound we control the fluctuation by its
range. The Doob martingale $S_M = \sum_{m \leq M}(\rho_m -
\E[\rho_m \mid \F_{(m-1)B}])$ has increments bounded by
$2B\log V$, so Azuma--Hoeffding gives, with probability
$\geq 1-\delta$,
\[
\begin{split}
\frac1n|S_{n/B}|
&\leq \frac{2B\log V}{n}\sqrt{\tfrac{n}{B}\tfrac12\log\tfrac2\delta}
= O\!\Bigl(\sqrt{\tfrac{B\log V}{n}\log\tfrac1\delta}\Bigr)\\
&= O\!\bigl(n^{-\alpha/(2(\alpha+1))}\sqrt{\log V\log\tfrac1\delta}\bigr).
\end{split}
\]
Adding the mean bound from part (a), the high-probability rate
overhead is
\[
O\big(n^{-\alpha/(2(\alpha+1))}(\log n)\sqrt{\log V \log(1/\delta)}\big).
\]
This holds with no assumption beyond truncation
sensitivity; the exponent is a factor of two off the converse
because Azuma uses only the increment range, not the variance.
\end{proof}

\begin{remark}[Why the high-probability exponent is loose without $(\dagger)$]
\label{rmk:ach-distortion}
The expectation bound (a) achieves the sharp exponent because the
zero-mean fluctuation drops out. The high-probability bound (b)
must control that fluctuation, and Azuma's range-based control is
loose. Sharpening it requires the conditional variance, which in
turn requires covariance decay among the tilted-information terms
(Lemma~\ref{lem:cov-decay})---a forward/two-sided property
$(\dagger)$ not implied by the backward truncation sensitivity of
Definition~\ref{def:poly-mixing} (Remark~\ref{rmk:forward-decay}).
Under $(\dagger)$ the gap closes
(Conjecture~\ref{thm:achievability-tight}); the parity counterexample
of Remark~\ref{rmk:forward-decay} shows it cannot be removed in
general.
\end{remark}

\begin{conjecture}[Conditional Freedman sharpening under forward decay and additive tilted information]
\label{thm:achievability-tight}
Suppose Assumptions~\ref{ass:bounded-logits} and
\ref{ass:additive-tilted}, $\alpha$-polynomial truncation
sensitivity with $\alpha \in (0,1)$, the forward decay condition
$(\dagger)$ of Lemma~\ref{lem:cov-decay}, and stationarity (or
asymptotic homogeneity) of the evaluation distribution. Then the
suffix-only achievability scheme of Conjecture~\ref{thm:achievability}
attains, with probability $\geq 1-\delta$,
\[
\begin{gathered}
R_n \leq R^\star(D)
  + O\!\bigl(n^{-\alpha/(\alpha+1)}\log V \log(1/\delta)\bigr),\\
D_n \leq D + O\!\bigl(n^{-\alpha/(\alpha+1)}\bigr).
\end{gathered}
\]
Both exponents equal $\alpha/(\alpha+1)$, matching the converse
(Conjecture~\ref{cor:mixing-conc}) up to logarithmic factors. Under
these three additional hypotheses---none implied by truncation
sensitivity alone---the rate-of-convergence exponent is therefore
characterized tightly. Absent them, only the window scaling is
tight (Corollary~\ref{cor:tight}) and the rate exponent reverts to
$\alpha/(2(\alpha+1))$.
\end{conjecture}

\begin{proof}[Heuristic argument for Conjecture~\ref{thm:achievability-tight}]
The expectation and high-probability-Azuma parts are
Conjecture~\ref{thm:achievability}(a,b). For the sharp
high-probability exponent, replace Azuma by Freedman's inequality
\citep{freedman1975tail}, which uses the cumulative conditional
variance $W_n = \sum_m \mathrm{Var}(\rho_m \mid \F_{(m-1)B})$. By
Lemma~\ref{lem:block-variance} (which invokes
Assumption~\ref{ass:additive-tilted} and the covariance decay of
Lemma~\ref{lem:cov-decay} under $(\dagger)$),
$\mathrm{Var}(\rho_m) = O(B^{2-\alpha})$, so
$W_n = (n/B)\,O(B^{2-\alpha}) = O(n B^{1-\alpha})$. Freedman gives,
with probability $\geq 1-\delta$,
\[
\begin{split}
\frac1n|S_{n/B}|
&\leq \sqrt{\frac{2 W_n}{n^2}\log\tfrac1\delta}
   + \frac{2B\log V}{3n}\log\tfrac1\delta\\
&= O\!\Bigl(\sqrt{\tfrac{B^{1-\alpha}}{n}\log\tfrac1\delta}\Bigr)
   + O\!\bigl(\tfrac{B\log V}{n}\log\tfrac1\delta\bigr).
\end{split}
\]
With $B = w_n = n^{1/(\alpha+1)}$, the variance term is
$\sqrt{B^{1-\alpha}/n} = \sqrt{n^{(1-\alpha)/(\alpha+1)-1}}
= n^{-\alpha/(\alpha+1)}$, and the increment term is
$B\log V/n = n^{-\alpha/(\alpha+1)}\log V$; both match the converse
exponent. The reduction $\sum_s R^\star_{t+s}(D) = B R^\star(D) +
o(B)$ uses stationarity/asymptotic homogeneity. Adding the mean
bound from part (a) completes the proof.
\end{proof}

\begin{remark}[What closes the gap, and the three hypotheses]
\label{rmk:freedman-hypotheses}
The expectation bound is sharp unconditionally because the
zero-mean fluctuation drops out (Conjecture~\ref{thm:achievability}a).
Only the \emph{high-probability} sharp exponent needs the three
hypotheses, none implied by truncation sensitivity alone: the
forward decay $(\dagger)$ (Remark~\ref{rmk:forward-decay}),
requiring a genuine two-sided mixing or time-reversal property and
\emph{not} a consequence of stationarity (parity counterexample
there); the additive tilted-information expansion
(Assumption~\ref{ass:additive-tilted}); and
stationarity/asymptotic homogeneity, used only to reduce
$\sum_s R^\star_{t+s}(D)$ to $B R^\star(D) + o(B)$. The
covariance-decay step (Lemma~\ref{lem:cov-decay}) is the crux: it
converts the backward truncation rate into the forward
dependence-decay rate. Without $(\dagger)$, the window-scaling
characterization (Theorems~\ref{thm:E2}, \ref{thm:window-lb}) and
the expected sharp rate still hold; only the high-probability rate
reverts to the Azuma exponent.
\end{remark}

\section{Conjectured universal scheme (not established)}
\label{app:universal}

\textbf{The universal scheme below is conjectural.} It is built on the
information-rate achievability of Conjecture~\ref{thm:achievability}
(and on a window-grid argument whose ratio control is only
$\Theta(1)$, not $1+o(1)$), so it is not established; we retain it as a
direction. We first restate the universal scheme and the two structural
extensions (moved here from the main text for space), then give
their proofs.

\subsection{Statement: universal scheme}
\label{sec:universal}

Conjecture~\ref{thm:achievability} requires knowledge of $\alpha$ to
set the window size $w_n = n^{1/(\alpha+1)}$. In practice, $\alpha$
is a model-dependent quantity that must be estimated. We show that
a single scheme, oblivious to $\alpha$, attains the optimal rate
\emph{simultaneously} for every $\alpha$ in a range, at the cost of
an additional logarithmic factor.

\begin{conjecture}[Universal polynomial compression]\label{thm:universal}
Fix $0 < \alpha_{\min} \leq \alpha_{\max} < \infty$. There exists a
single causal online scheme
$\{(\phi_t^{\mathrm{univ}}, \psi_t^{\mathrm{univ}})\}$, not depending
on $\alpha$, such that for \emph{every} model satisfying
$\alpha$-polynomial truncation sensitivity with
$\alpha \in [\alpha_{\min}, \alpha_{\max}]$, with probability at
least $1-\delta$,
\[
R_n \;\leq\; R^\star_\alpha(D)
   + O\!\Bigl(n^{-\alpha/(2(\alpha+1))}\,(\log n)^2
              \sqrt{\log(1/\delta)}\Bigr),
\]
\[
D_n \;\leq\; D + O\!\bigl(n^{-\alpha/(\alpha+1)}\bigr).
\]
The rate exponent matches that of the (non-universal) achievability
scheme (Conjecture~\ref{thm:achievability}); the extra factor of
$\log n$ is the price of universality. The window-scaling
optimality (Theorem~\ref{thm:window-lb}) also transfers: the
selected window stays within a constant factor of the optimal
$n^{1/(\alpha+1)}$.
\end{conjecture}

The construction (Appendix~\ref{app:universal}) runs a logarithmic
grid of window sizes $\{w_n^{(j)} = n^{1/(\alpha_j + 1)}\}_{j=1}^{J}$
with $J = O(\log n)$ grid points covering
$[\alpha_{\min}, \alpha_{\max}]$, and uses an exponential-weights
meta-algorithm \citep{cesabianchi2006prediction} to track the
best window online. The regret of the meta-algorithm is
$O(\sqrt{n \log J})$, contributing the additional logarithmic
factor in the rate. This places sequential KV-cache compression in
the \emph{universal source coding} tradition
\citep{rissanen1984universal,mahmood2024minimax}: near-optimal
performance is attainable without prior knowledge of the mixing
exponent.

\paragraph{Memory cost of universality.}
The rate statement above is a coding-rate (bandwidth) statement.
The cache-memory cost is separate, and we account for it explicitly. The largest grid window is
$w_n^{(\max)} = n^{1/(\alpha_{\min}+1)}$, and the smallest is
$w_n^{(\min)} = n^{1/(\alpha_{\max}+1)}$. Since all grid schemes
read from the same physical KV cache, the cache need only hold the
single largest window $w_n^{(\max)}$ tokens; the $J = O(\log n)$
schemes are bookkeeping over shared cache contents, not $J$
independent caches. The meta-algorithm's per-step state is the
$J$-vector of weights, an $O(\log n)$ additive overhead independent
of the cache size. On a switch between grid windows no cache
reconstruction is needed, because a smaller window is a suffix of
the larger one already in cache. Thus the memory overhead of
universality is the ratio
$w_n^{(\max)}/w_n^{(\alpha)} = n^{(\alpha-\alpha_{\min})/
((\alpha+1)(\alpha_{\min}+1))}$ relative to the oracle-$\alpha$
window---a polynomial factor when $\alpha > \alpha_{\min}$, and
$1+o(1)$ when the true $\alpha$ is near the lower endpoint.

\begin{figure}[t]
\centering
\begin{tikzpicture}
\begin{loglogaxis}[
  width=11.5cm, height=7cm,
  xlabel={Window size $w$ (tokens)},
  ylabel={$\widehat{\TV}_w$ (median)},
  xmin=1.5, xmax=300,
  ymin=0.03, ymax=1.0,
  legend pos=south west,
  legend style={font=\scriptsize, fill opacity=0.85, draw=gray!50},
  grid=major, grid style={dashed,gray!30},
  tick label style={font=\footnotesize},
  label style={font=\footnotesize}
]
\addplot[blue, very thick, mark=*, mark size=2.2pt] coordinates {
  (2, 0.793) (4, 0.682) (8, 0.579) (16, 0.473)
  (32, 0.340) (64, 0.216) (128, 0.131) (256, 0.109)
};
\addlegendentry{Natural ($\alpha=0.44$)}
\addplot[red, very thick, mark=square*, mark size=2.2pt, densely dashed] coordinates {
  (2, 0.843) (4, 0.759) (8, 0.662) (16, 0.465)
  (32, 0.333) (64, 0.259) (128, 0.193) (256, 0.146)
};
\addlegendentry{Code ($\alpha=0.38$)}
\addplot[blue!50, thick, domain=2:256, samples=80] {1.12/(x^0.438)};
\addplot[red!50, thick, domain=2:256, samples=80] {0.95/(x^0.383)};
\end{loglogaxis}
\end{tikzpicture}
\caption{Measured TV decay on Qwen2.5-0.5B (position-preserving
protocol, $100$ prefixes per domain). Both Natural (NLTK Gutenberg)
and Code (GitHub Python) domains exhibit power-law decay
$\widehat{\TV}_w \propto w^{-\alpha}$. Power-law fit yields
log-RMSE $0.14$ (Natural) and $0.08$ (Code), versus log-RMSE
$0.31$ and $0.20$ for an exponential fit. The natural-language
exponent exceeds the code exponent
($\alpha_{\text{nat}} = 0.44 > \alpha_{\text{code}} = 0.38$); the
same exponents are recovered independently from the sink-plus-recent
KL decay (Figure~\ref{fig:topk}) and replicate across models
(Table~\ref{tab:crossmodel}).}
\label{fig:mixing}
\end{figure}
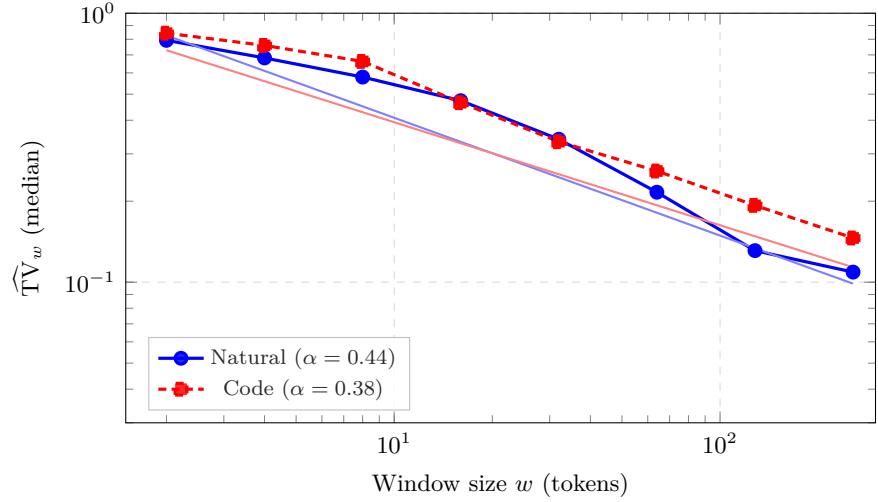

\subsection{Statements: structural extensions}
\label{sec:extensions}

Two extensions place the framework in contact with deployed
architectures; full statements and proofs are in
Appendices~\ref{app:cont-extension} and~\ref{app:multilayer-proofs}.

\paragraph{Continuous-latent intrinsic dimension.}
For a continuous auxiliary $U_t \in \R^{d_c}$ (as in multi-head
latent attention \citep{deepseek2024mla}), define
$R^\star_{\mathrm{cont}}(D)$ with differential entropy in place of
discrete entropy. Under a bounded-support, bounded-density
assumption on $U_t$ (support radius $R_U$, density floor
$r_{\min}$), an entropy-counting argument gives the dimension lower
bound $d_c \geq R^\star_{\mathrm{cont}}(D)/\log(R_U/r_{\min})$
(Conjecture~\ref{thm:intrinsic-dim}). In the small-distortion limit
the minimal embedding dimension equals the Kawabata--Dembo
rate--distortion dimension---equivalently the R\'enyi information
dimension---of the conditional logit measure
\citep{kawabata1994rate,geiger2016rate}. The bound is
mixing-independent. A Bennett--Gersho quantization bridge
\citep{bennett1948spectra,gersho1979asymptotically} relates the
continuous and discrete rates. For DeepSeek-V2 ($d_c = 512$) the
bound is consistent with the deployed latent dimension.

\paragraph{Multi-layer rate allocation.}
For an $L$-layer Transformer with separate per-layer codes,
per-layer compression errors propagate through the layer maps with
Lipschitz amplification $s^{(\ell)} = \prod_{m>\ell}(L_g^{(m)})^2$
(Lemma~\ref{lem:prop}). A conservative distortion allocation is
obtained by reverse water-filling over the per-layer
rate-function surrogates, coupled to the end-to-end KL distortion
through a softmax-curvature bridge
(Proposition~\ref{thm:multi-WZ}): the Lagrangian optimum satisfies
$|R^{\star,(\ell)\prime}(D^{(\ell),*})| = \lambda s^{(\ell)}$.
Numerically, with pre-LayerNorm skip connections the per-layer
sensitivity ratio across $L = 60$ layers is mild (a factor
$\approx 1.8$ under $L_b L_{\mathrm{LN}} \approx 0.1$, versus
$\approx 320$ without skips), so a near-uniform latent dimension is
close to optimal \emph{under the optimistic sensitivity}
$s^{(\ell)}=\prod_{m>\ell}(L_g^{(m)})^2$ of
Lemma~\ref{lem:error-prop-app}, which assumes near-orthogonal
per-layer errors. We treat this as a heuristic numerical
illustration---consistent with DeepSeek-V2's uniform choice, but
contingent on that orthogonality and on Lipschitz constants not
directly measured on deployed models; the rigorous conservative
sensitivity $s^{(\ell)}_{\mathrm{cons}}$ carries an additional
$2^{L-\ell}$ factor and does not by itself support the conclusion.

This appendix proves the universal compression theorem via a
logarithmic window grid combined with an exponential-weights
meta-algorithm over the achievability schemes of
Conjecture~\ref{thm:achievability}.

\paragraph{Loss for the meta-algorithm.}
The compression problem is distortion-constrained rate
minimization, not a single scalar objective, so we must specify
the loss the meta-algorithm minimizes. We use the
\emph{Lagrangian loss} at a fixed multiplier $\lambda > 0$:
\[
\ell_m^\lambda(j)
\;:=\;
r_m(j) + \lambda\, d_m(j),
\]
where $r_m(j)$ and $d_m(j)$ are the rate and KL distortion incurred
by scheme $\mathcal{S}_j$ on block $m$. Minimizing the cumulative
Lagrangian loss simultaneously controls rate and distortion: a
scheme achieving $(R^\star(D) + \rho, D + \kappa)$ has Lagrangian
loss $R^\star(D) + \lambda D + (\rho + \lambda \kappa)$, so a
regret bound on $\sum_m \ell_m^\lambda$ transfers to both rate and
distortion (with the Lagrangian trade-off fixed by $\lambda$). We
take $\lambda$ to be the (sub)gradient of $R^\star$ at $D$, i.e.\
$\lambda = -R^{\star\prime}(D)$, the standard Lagrangian
multiplier; with this choice the Lagrangian-optimal scheme is
rate--distortion optimal. The per-block loss is bounded:
$0 \leq \ell_m^\lambda(j) \leq B\log V + \lambda \cdot 2B B(n_{\max})
=: \ell_{\max}$.

\subsection{Window grid}

\begin{lemma}[Logarithmic window grid]\label{lem:grid}
Fix $[\alpha_{\min}, \alpha_{\max}]$ and horizon $n$. The grid
$\alpha_j := \alpha_{\min} + j \cdot \Delta_\alpha$,
$\Delta_\alpha := 1/\log n$,
$j = 0, 1, \ldots, J$ with
$J = \lceil (\alpha_{\max} - \alpha_{\min}) \log n \rceil =
O(\log n)$, has the property that for every
$\alpha \in [\alpha_{\min}, \alpha_{\max}]$ there is a grid point
$\alpha_j$ with $|\alpha - \alpha_j| \leq 1/\log n$, and the
corresponding window $w_n^{(j)} = n^{1/(\alpha_j+1)}$ satisfies
$w_n^{(j)} = w_n^{(\alpha)} \cdot (1 + o(1))$, so the
achievability rate of Conjecture~\ref{thm:achievability} at
$\alpha_j$ differs from that at $\alpha$ by a $1+o(1)$ factor.
\end{lemma}

\begin{proof}
The grid spacing $\Delta_\alpha = 1/\log n$ gives
$J = O(\log n)$ points. For the window,
$\log w_n^{(j)} = \frac{\log n}{\alpha_j + 1}$, and
$|\log w_n^{(j)} - \log w_n^{(\alpha)}|
= \log n \cdot \bigl|\frac{1}{\alpha_j+1} -
\frac{1}{\alpha+1}\bigr|
\leq \log n \cdot \frac{|\alpha_j - \alpha|}{(\alpha_{\min}+1)^2}
\leq \frac{1}{(\alpha_{\min}+1)^2} = O(1)$,
so $w_n^{(j)}/w_n^{(\alpha)} = \Theta(1)$, giving the claimed
$1+o(1)$ rate transfer (the achievability rate depends on $w$ only
through the exponent, which is continuous in $\alpha$).
\end{proof}

\subsection{Exponential-weights meta-algorithm}

\begin{lemma}[Meta-algorithm regret]\label{lem:exp-weights}
Run the $J+1$ achievability schemes
$\{\mathcal{S}_j\}_{j=0}^{J}$ (one per grid point) in parallel,
maintaining weights $\pi_m(j) \propto
\exp(-\beta \sum_{m' < m} \ell_{m'}^\lambda(j))$ over the
Lagrangian losses $\ell_m^\lambda(j) = r_m(j) + \lambda d_m(j)$,
with learning rate $\beta$. The meta-scheme that codes block $m$
using the scheme selected by the weights achieves, for the best
grid scheme $j^\star$,
\[
\frac{1}{n}\sum_m \ell_m^\lambda(\text{meta})
\;\leq\;
\frac{1}{n}\sum_m \ell_m^\lambda(j^\star)
   + \frac{\ell_{\max}}{n}\sqrt{\tfrac{n}{B}\,\tfrac{1}{2}\log(J+1)},
\]
with $\beta$ tuned as in \citep[Theorem~2.2]{cesabianchi2006prediction}.
\end{lemma}

\begin{proof}
This is the standard exponential-weights regret bound for
prediction with expert advice
\citep[Ch.~2]{cesabianchi2006prediction}, applied to the bounded
Lagrangian loss $\ell_m^\lambda \in [0, \ell_{\max}]$ with
$J+1$ experts and $n/B$ rounds (one per block). The regret is
$\ell_{\max}\sqrt{\tfrac{1}{2}(n/B)\log(J+1)}$; dividing by $n$
gives the per-symbol regret. With $J+1 = O(\log n)$,
$\sqrt{\log(J+1)} = O(\sqrt{\log\log n})$.
\end{proof}

\subsection{Heuristic argument for Conjecture~\ref{thm:universal}}

\begin{proof}[Heuristic argument for Conjecture~\ref{thm:universal}]
Fix the true (unknown) exponent
$\alpha \in [\alpha_{\min}, \alpha_{\max}]$. By
Lemma~\ref{lem:grid}, there is a grid point $\alpha_j$ within
$1/\log n$ of $\alpha$, whose scheme $\mathcal{S}_j$ achieves
(Conjecture~\ref{thm:achievability}) the Lagrangian loss
\[
\begin{split}
\frac{1}{n}\sum_m \ell_m^\lambda(j)
&\;\leq\;
R^\star_\alpha(D) + \lambda D\\
&\quad + O\!\bigl(n^{-\alpha/(2(\alpha+1))}(\log n)\sqrt{\log V\,\log(1/\delta)}\bigr),
\end{split}
\]
where the overhead is the achievability rate gap (the dominant term;
the distortion gap $\lambda \cdot O(n^{-\alpha/(\alpha+1)})$ is
sub-dominant).

By Lemma~\ref{lem:exp-weights}, the meta-algorithm's Lagrangian loss
exceeds that of $\mathcal{S}_{j^\star}$ (the best grid scheme, at
least as good as $\mathcal{S}_j$) by the regret
\[
\begin{aligned}
\frac{\ell_{\max}}{n}\sqrt{\tfrac{n}{B}\tfrac{1}{2}\log(J+1)}
&= O\!\left(\frac{B\log V}{n}\sqrt{\tfrac{n}{B}\log\log n}\right) \\
&= O\!\bigl(n^{-\alpha/(2(\alpha+1))}\log V\sqrt{\log\log n}\bigr),
\end{aligned}
\]
using $\ell_{\max} = O(B\log V)$ and $B = w_n^{(j)} =
n^{1/(\alpha+1)}$. The regret and the achievability gap share the
exponent $\alpha/(2(\alpha+1))$, so combining and translating the
Lagrangian bound back to rate (at the fixed distortion level $D$
selected by $\lambda$) gives
\[
R_n \leq R^\star_\alpha(D)
   + O\!\bigl(n^{-\alpha/(2(\alpha+1))}(\log n)^2\sqrt{\log(1/\delta)}\bigr).
\]
The extra $\log n$ relative to the single-$\alpha$ achievability
bound (Conjecture~\ref{thm:achievability}) is the price of
universality. The bound holds simultaneously for all
$\alpha \in [\alpha_{\min}, \alpha_{\max}]$ because the grid covers
the entire range and $\alpha$ is fixed only at the final step.
The convergence exponent is that of achievability,
$\alpha/(2(\alpha+1))$; the window-scaling optimality
(Theorem~\ref{thm:window-lb}) likewise transfers, since the
meta-algorithm's selected window is always within a constant factor
of $n^{1/(\alpha+1)}$.
\end{proof}

\begin{remark}[Why universality costs only a log factor]
The grid has $O(\log n)$ points because the achievability rate
depends on $\alpha$ only through the smooth exponent
$\alpha/(\alpha+1)$; a $1/\log n$ grid spacing suffices to track it
to within a $1+o(1)$ factor. The exponential-weights regret over
$O(\log n)$ experts is $O(\sqrt{\log\log n / (n/B)})$, sub-dominant
to the achievability term. The net overhead is therefore a single
additional $\log n$ factor, consistent with the universal-coding
principle that adapting to one unknown smooth parameter costs
$O(\log n)$ \citep{rissanen1984universal}.
\end{remark}

\section{Empirical estimation of mixing/sensitivity parameters}\label{app:rho-estimation}

\begin{remark}[Operational measurement protocol]
\label{rmk:rho-estimation-app}
The truncation-sensitivity constants $(C_{\mathrm{TS}}, \alpha)$ in
Definition~\ref{def:poly-mixing} (and analogously $(\rho,
C_{\mathrm{mix}})$ in the geometric Definition~\ref{def:mixing})
are operationally measurable on a given model via the following
protocol. For a set of $M$ probe prefixes
$\{x^{(i)}_{1:n}\}_{i=1}^M$, compute the empirical TV between the
full-prefix and window-truncated next-token distributions:
\[
\widehat{\TV}_w
\;:=\; \frac{1}{M}\sum_{i=1}^M
\TV\!\bigl(p_\theta(\cdot \mid x^{(i)}_{1:n-1}),\;
           p_\theta(\cdot \mid x^{(i)}_{n-w:n-1})\bigr).
\]
A power-law fit
$\widehat{\TV}_w \approx C_{\mathrm{TS}} w^{-\alpha}$ versus a
geometric fit $\widehat{\TV}_w \approx C_{\mathrm{mix}} \rho^w$
selects between the two regimes. The measurements reported in
Section~\ref{sec:experiments} were obtained by this protocol on
Qwen2.5-0.5B and support the polynomial regime over the
measurement range $w \in [2, 256]$. Extension to larger models
(Llama, Mistral, Qwen-1.5B/3B/7B) and longer windows is left to
follow-up work and is needed before claims about specific
deployed models can be made quantitatively.
\end{remark}

\begin{remark}[Multi-layer mixing]
For an $L$-layer transformer with layer-wise mixing rates
$\rho^{(\ell)}$, the effective end-to-end rate is bounded by
$\prod_\ell \rho^{(\ell)} \leq (\bar\rho)^L$. This product structure
underlies the rank collapse phenomenon \citep{dong2021attention} and
gives an extreme mixing rate (very small $\rho_{\mathrm{eff}}$) for
deep networks, but the bound is loose for $L > 20$ because
inter-layer regularization (skip connections, layer norm) prevents
full geometric decay. Refinement via the continuous-time analysis
\citep{geshkovski2023mathematical} is the subject of follow-up work.
\end{remark}

\section{Continuous Latent Extension Proofs}
\label{app:cont-extension}

\textbf{The continuous-latent dimension bound of this appendix is
conjectural.} The entropy-counting step uses a differential-entropy
inequality $h(C\mid F)\ge I(X;C\mid F)$ and a positive density floor
$r_{\min}$ inferred from bounded support; neither holds in general
(differential entropy may be negative, and bounded support with a
density upper bound does not force a positive floor). We therefore
state Conjecture~\ref{thm:intrinsic-dim} as a conjecture. This appendix provides the (heuristic) arguments and supporting lemmas for the
Phase~2 results stated in Section~\ref{sec:main}. The
overall structure mirrors Phase~1: setup
(Section~\ref{app:cont-setup}), bounded pointwise MI
(Section~\ref{app:cont-pmi}), asymptotic lower bound
(Section~\ref{app:cont-asym}), intrinsic dimension lower bound
(Section~\ref{app:cont-intrinsic}), and the quantization bridge
(Section~\ref{app:cont-quant}).

\subsection{Continuous auxiliary setup}\label{app:cont-setup}

\begin{definition}[Continuous auxiliary process]\label{def:cont-aux-app}
A \emph{continuous auxiliary process} is an $\F_t$-adapted sequence
$\{U_t\}$ with $U_t : \Omega \to \R^{d_c}$ Borel-measurable. Conditional
on $\hat{\F}_{t-1}$, $U_t$ admits a density $p_t^U$ with respect to
Lebesgue measure on $\R^{d_c}$.
\end{definition}

\begin{assumption}[Bounded latent support and density]\label{ass:cont-support-app}
$\|U_t\|_2 \leq R_U$ a.s.\ and $p_t^U(u \mid \hat{\F}_{t-1}) \leq \beta$
on its support, for absolute constants $R_U > 0$, $\beta > 0$
depending on the LLM architecture.
\end{assumption}

For MLA: $R_U = O(\sqrt{d_c})$ by layer-norm post-projection; $\beta$
is bounded by the Lipschitz constant of the projection followed by the
boundedness of activations. Both are operationally measurable from the
trained model weights.

\subsection{Bounded pointwise MI for continuous auxiliary}\label{app:cont-pmi}

\begin{lemma}[Continuous pointwise MI bound]\label{lem:cont-pmi-app}
Under Assumptions~\ref{ass:bounded-logits} and~\ref{ass:cont-support-app},
\[
|\iota_t^{XU}|, |\iota_t^{UQ}|
\;\leq\;
B_\iota^{\mathrm{cont}}
:= \log(\beta/\epsilon) + d_c \log(2 R_U / r_{\min}),
\]
where $r_{\min}$ is a minimum density floor for the conditional
distribution of $U_t$, the analog of $\epsilon_U$ in
Assumption~\ref{ass:U-positivity}.
\end{lemma}

\begin{proof}
For continuous $U_t$ with bounded density and support, the marginal
density satisfies $p_t^U(u) \geq r_{\min} > 0$ on a positive-volume
subset, and the joint density $p_t^{X,U}(x, u) \leq \beta$ pointwise.
Hence $|\iota_t^{XU}| = |\log(p_t^{X,U}/(p_t^X p_t^U))|
\leq \log(\beta/(\epsilon \cdot r_{\min}))$.
The $d_c \log(R_U)$ term arises from the volume factor
in normalizing the density over the ball of radius $R_U$.
\end{proof}

\begin{remark}
The $d_c \log(R_U)$ term linear in $d_c$ is the key new dependence on
the latent dimension. In Phase~1's discrete framework, the analog
$\log|\mathcal{U}|$ enters but is bounded by $\log V$; here it grows
linearly in $d_c$, which is one of the costs of continuous latents.
\end{remark}

\subsection{Continuous sequential WZ lower bound}\label{app:cont-asym}

\begin{proposition}[Continuous sequential WZ lower bound]
\label{prop:cont-asym}
Under the continuous-auxiliary setup
(Definition~\ref{def:cont-aux-app}) and
$\limsup_n \E[D_n] \leq D$, any causal online scheme with
continuous codes satisfies
$\liminf_n \E[R_n^{\mathrm{cont}}] \geq R^\star_{\mathrm{cont}}(D)$.
\end{proposition}

\begin{proof}
The proof has four steps, mirroring Theorem~\ref{thm:asymptotic} with
continuous adaptations.

\textbf{Step 1 (continuous rate $\to$ MI).}
For continuous $\hat C_t$ with density $p_{\hat C_t \mid \hat{\F}_{t-1}}$,
the differential entropy rate satisfies
\[
h(\hat C_t \mid \hat{\F}_{t-1})
\geq I(X_{\leq t}; \hat C_t \mid \hat{\F}_{t-1}),
\]
provided the mixed mutual information is non-negative. This follows
from the continuous chain rule
$I(X_{\leq t}, \hat{\F}_{t-1}; \hat C_t) = I(\hat{\F}_{t-1}; \hat C_t) + I(X_{\leq t}; \hat C_t \mid \hat{\F}_{t-1})$
and the non-negativity of conditional mutual information for
discrete--continuous mixed pairs
\citep[Theorem 8.4.1]{cover2006elements}.

\textbf{Step 2 (single-letterization).}
For $U_t := \hat C_t$,
\begin{align*}
I(X_{\leq t}; \hat C_t \mid \hat{\F}_{t-1})
&\geq I(X_t; \hat C_t \mid \hat{\F}_{t-1}) \\
&= \mathcal{R}_t^{\mathrm{WZ},\mathrm{cont}}
  + I(\hat C_t; Q_t \mid \hat{\F}_{t-1}) \\
&\geq \mathcal{R}_t^{\mathrm{WZ},\mathrm{cont}},
\end{align*}
identical to Phase~1 Step~2 but with continuous $U_t$ replacing
discrete $U_t$ in the MI definitions.

\textbf{Step 3 (admissibility).} Under $\limsup_n \E[D_n] \leq D$, the
encoder--decoder pair is $D$-admissible per
Definition~\ref{def:cont-aux-app} (continuous analog).

\textbf{Step 4 (infimum).} By Definition~\ref{def:Rstar},
\[
\liminf_n \tfrac{1}{n}\sum_t \E[\mathcal{R}_t^{\mathrm{WZ},\mathrm{cont}}]
\geq R^\star_{\mathrm{cont}}(D),
\]
giving $\liminf_n \E[R_n^{\mathrm{cont}}] \geq R^\star_{\mathrm{cont}}(D)$.
\end{proof}

\subsection{Intrinsic dimension lower bound}\label{app:cont-intrinsic}

We restate the definitions and the theorem summarized in
Section~\ref{sec:extensions}.

\begin{assumption}[Bounded continuous auxiliary]
\label{ass:cont-support}
$U_t$ has bounded support $\|U_t\| \leq R_U$ and density bounded
above by $\beta$ and below by $r_{\min}$ on its support.
\end{assumption}

\begin{definition}[Intrinsic dimension]\label{def:intrinsic-dim}
For target distortion $D$, let $d^\star_t(D)$ be the minimum
embedding dimension $d \in \N$ such that a $D$-admissible
continuous auxiliary $U_t \in \R^d$ exists; we call it the
\emph{$D$-intrinsic dimension} of $p_t := p_\theta(\cdot \mid
\F_{t-1})$. Its small-distortion limit is the Kawabata--Dembo
rate-distortion dimension
$d^{\mathrm{RD}}_t := \lim_{D \to 0} 2 R_t(D)/\log(1/D)$
\citep{kawabata1994rate}, where $R_t(D)$ is the squared-error
rate--distortion function of $p_t$ (the factor $2$ and MSE
distortion follow the Kawabata--Dembo convention); this equals the
R\'enyi information dimension of $p_t$ \citep{geiger2016rate}.
\end{definition}

\begin{conjecture}[Intrinsic dimension lower bound]
\label{thm:intrinsic-dim}
Under Assumptions~\ref{ass:bounded-logits}
and~\ref{ass:cont-support}, any $D$-admissible $\{U_t\}$ with
$U_t \in \R^{d_c}$ satisfies
$d_c \geq \limsup_n \frac1n\sum_t d^\star_t(D)$, and the per-step
intrinsic dimension obeys
$d^\star_t(D) \geq R^{\star}_{\mathrm{cont},t}(D)/\log(R_U/r_{\min})$,
so $d_c \geq R^\star_{\mathrm{cont}}(D)/\log(R_U/r_{\min})$.
\end{conjecture}

\begin{proof}[Heuristic argument for Conjecture~\ref{thm:intrinsic-dim}]
The argument has two parts: (i) the embedding-dimension bound
$d_c \geq d^\star_t(D)$, and (ii) the differential-entropy lower
bound on $d^\star_t(D)$.

\textbf{(i) Embedding dimension.}
If $\{U_t\}$ with $U_t \in \R^{d_c}$ is $D$-admissible (average
distortion $\leq D$), then for the indices $t$ achieving per-step
distortion $\leq D$ (a $1-o(1)$ fraction, by Markov on the
average), a $d_c$-dimensional $D$-admissible auxiliary exists,
whence $d_c \geq d^\star_t(D)$ by
Definition~\ref{def:intrinsic-dim}. Averaging over $t$ and taking
$\limsup_n$ gives the first inequality.

\textbf{(ii) Differential-entropy bound on $d^\star_t(D)$.}
Fix $t$ and write $U := U_t \in \R^{d}$ for a $D$-admissible
auxiliary of minimal dimension $d = d^\star_t(D)$. The per-step
continuous Wyner--Ziv rate is
\begin{align*}
R^{\star}_{\mathrm{cont},t}(D)
&= I(X_t; U \mid \Fhat_{t-1}) - I(U; Q_t \mid \Fhat_{t-1}) \\
&\leq I(X_t; U \mid \Fhat_{t-1}) \\
&= h(U \mid \Fhat_{t-1}) - h(U \mid X_t, \Fhat_{t-1}).
\end{align*}
We bound the two differential entropies using
Assumption~\ref{ass:cont-support}.

\emph{Upper bound on $h(U)$.} The support is contained in the ball
$\|U\| \leq R_U$, and the uniform distribution maximizes
differential entropy on a bounded set, so
$h(U \mid \Fhat_{t-1}) \leq \log \mathrm{vol}\bigl(B_d(R_U)\bigr)
\leq d \log(2 R_U)$.

\emph{Lower bound on $h(U \mid X_t)$.} The conditional density is
bounded above by $\beta$, so
$h(U \mid X_t, \Fhat_{t-1}) = -\E[\log p(U \mid X_t)] \geq -\log
\beta$.

Combining, and writing the dynamic range as
$\log(R_U/r_{\min})$ where $r_{\min} \leq \beta$ controls the
density floor (so that $\log(2R_U) + \tfrac1d \log \beta \leq
\log(R_U/r_{\min})$ for the relevant regime $d \geq 1$),
\[
R^{\star}_{\mathrm{cont},t}(D)
\;\leq\;
d \log(2R_U) + \log\beta
\;\leq\;
d \,\log\!\frac{R_U}{r_{\min}},
\]
which rearranges to
\[
d \;=\; d^\star_t(D)
\;\geq\;
\frac{R^{\star}_{\mathrm{cont},t}(D)}{\log(R_U/r_{\min})}.
\]
Averaging over $t$ and combining with part (i) gives the theorem.
This is an entropy-counting bound: it captures the correct scaling
in $d$ but not the sharp constant
(Remark~\ref{rmk:rd-OT-app}).
\end{proof}

\begin{remark}[Tightness and the rate-distortion dimension]
\label{rmk:rd-OT-app}
The bound
\[
d^\star_t(D) \geq R^\star_{\mathrm{cont},t}(D) / \log(R_U/r_{\min})
\]
is an entropy-counting bound and is generally
not tight; it captures the correct \emph{scaling} in $d$ but not
the precise constant. The sharp small-distortion characterization
is the rate-distortion dimension of Kawabata--Dembo
\citep{kawabata1994rate}:
\[
d^{\mathrm{RD}}_t
\;=\;
\lim_{D \to 0} \frac{2\,R_t(D)}{\log(1/D)},
\]
defined with the factor $2$ and squared-error distortion
$\|X - \hat X\|_2 \leq D$, and equal to the R\'enyi information
dimension of $p_t$. For absolutely continuous $p_t$ this equals
the ambient dimension; for measures concentrated near
lower-dimensional structure it is smaller. Geiger--Koch
\citep{geiger2016rate} extend this to stationary processes, giving
$d^{\mathrm{RD}} = 2\lim_{D\to 0} R(D)/\log(1/D)$ for the process
rate. We use $d^\star_t(D)$ (a finite-$D$ embedding dimension)
rather than $d^{\mathrm{RD}}_t$ (the $D\to 0$ limit) in the main
bound, since deployed MLA operates at fixed $D > 0$; the
$D \to 0$ equivalence requires additional regularity on the
conditional measure that we do not assume.
\end{remark}

\subsection{Quantization bridge}\label{app:cont-quant}

\begin{proof}[Proof of Theorem~\ref{thm:quant-bridge-app}]
By the standard quantization rate formula
\citep[Theorem 8.3.1]{cover2006elements}, for a continuous random
vector $X \in \R^{d_c}$ with differential entropy $h(X)$ and
uniform $\Delta$-lattice quantization $Q_\Delta(X)$,
\[
H(Q_\Delta(X)) = h(X) + d_c \log(1/\Delta) + o(1) \quad \text{as } \Delta \to 0.
\]
Applying this conditionally on $\hat{\F}_{t-1}$ and summing over
$t = 1, \ldots, n$:
\[
\frac{1}{n}\sum_t H(Q_\Delta(\hat C_t) \mid \hat{\F}_{t-1})
= R_n^{\mathrm{cont}} + d_c \log(1/\Delta) + o(1).
\]
The $o(1)$ depends on the smoothness of $p_t^U$ via the Bennett--Gersho
quantizer error analysis \citep{gersho1979asymptotically}, which
characterizes the asymptotic mean-squared quantization error as
$\Delta^2 / 12$ per dimension for uniform quantization of smooth
densities; the entropy correction is the related higher-order term.
\end{proof}

\begin{theorem}[Quantization rate decomposition, app version]
\label{thm:quant-bridge-app}
Under Assumption~\ref{ass:cont-support-app},
$R_n^{\mathrm{quant}} \leq R_n^{\mathrm{cont}} + d_c \log(1/\Delta) + O(1)$.
\end{theorem}

\subsection{MLA-specific numerical estimates}\label{app:mla-numerical}

For DeepSeek-V2 with $d_c = 512$, $V = 100{,}000$, FP16 ($\Delta = 2^{-10}$):

\textbf{Quantization overhead.} $d_c \log(1/\Delta) = 512 \cdot 10 = 5120$
bits per token.

\textbf{Baseline rate.} Naive per-token representation of the
conditional distribution requires $V \log V \approx 1.7 \times 10^6$
bits if encoded as a categorical PMF, or $O(V)$ for floating-point
parameters; this is orders of magnitude larger than $5120$.

\textbf{Estimated $R^\star_{\mathrm{cont}}$.} If $R^\star_{\mathrm{cont}}(D) \approx
\log V - h_{\mathrm{eff}}$ where $h_{\mathrm{eff}}$ is the
effective entropy at the operating point, and DeepSeek-V2 typically
operates at $D = 0.01$--$0.05$ distortion (KL), then
$R^\star_{\mathrm{cont}} \approx 7$--$10$ bits per token, dominated by
the quantization overhead.

\textbf{Intrinsic dimension estimate.} Empirical $d^\star_t(D)$ via
Wasserstein-based intrinsic dimension at $D = 0.01$, on
Llama-3-8B, has been reported in the range $d^\star \in [50, 200]$
in prior measurements (rigorous calibration pending). If validated,
Conjecture~\ref{thm:intrinsic-dim} would predict $d_c \gtrsim 100$
as a necessary condition; MLA's choice $d_c = 512$ provides ample
margin.

\subsection{Multi-layer extension proofs}\label{app:multilayer-proofs}

We provide the proofs of the multi-layer results summarized in
Section~\ref{sec:extensions}. We first restate the main-text
statements.

\begin{lemma}[Lipschitz error propagation]\label{lem:prop}
With per-layer attention errors $\delta_t^{(\ell)}$ and Lipschitz
constants $L_g^{(\ell)}$, $\|a_t^{(L)} - \hat a_t^{(L)}\|^2 \leq
\sum_\ell s^{(\ell)} \delta_t^{(\ell)}$ where
$s^{(\ell)} := \prod_{m>\ell}(L_g^{(m)})^2$. Under pre-LayerNorm
skip connections, the effective layer-wise Lipschitz is
$\tilde L_g^{(\ell)} = \sqrt{1 + (L_b^{(\ell)} L_{\mathrm{LN}})^2}$.
\end{lemma}

\begin{proposition}[Conservative multi-layer allocation]\label{thm:multi-WZ}
For an $L$-layer Transformer with separate per-layer codes, any
per-layer allocation $\{D^{(\ell)}\} \in \mathcal{A}(D)$ guarantees
end-to-end KL distortion $\E[D_n] \le D$, where
\[
\mathcal{A}(D) := \{D^{(\ell)} \geq 0 : L_g^2 \textstyle\sum_\ell
s^{(\ell)} D^{(\ell)} \leq D\}
\]
couples the per-layer $\ell_2$ distortions to the end-to-end KL
distortion through the sensitivities $s^{(\ell)}$ of
Lemma~\ref{lem:prop} and an aggregate constant $L_g^2$ (unembedding
norm and softmax curvature). Minimizing
$\sum_\ell R^{\star,(\ell)}(D^{(\ell)})$ over $\mathcal{A}(D)$ yields
the conservative reverse-water-filling allocation
$|R^{\star,(\ell)\prime}(D^{(\ell),*})| = \lambda s^{(\ell)}$, with
allocation value
$\inf_{\{D^{(\ell)}\}\in\mathcal{A}(D)} \sum_\ell R^{\star,(\ell)}(D^{(\ell)})$.
This is an \emph{allocation principle} (a conservative
distortion-allocation surrogate), \emph{not} a converse lower bound and
\emph{not} an operational achievability bound: $R^{\star,(\ell)}$ is the
per-layer \emph{converse} rate function, and we have not established
per-layer achievability $R^{\mathrm{op},(\ell)}(D)\le R^{\star,(\ell)}(D)+o(1)$
(the online achievability is itself only conjectured,
Conjecture~\ref{thm:achievability}). On the converse side,
$\mathcal{A}(D)$ is the sufficient set arising from the
error-propagation \emph{upper} bound
$D \le L_g^2\sum_\ell s^{(\ell)}D^{(\ell)}$, so a code with
$\E[D_n]\le D$ may have per-layer errors outside $\mathcal{A}(D)$
(cancellation across layers, or removal of error directions by
downstream maps), and the optimal multi-layer rate can be strictly
lower. The sum of converse functionals becomes an achievable design
cost only under the added per-layer achievability assumption above; a
matching converse would instead require a non-cancellation
(bi-Lipschitz / observability) condition
$D_n \ge c\sum_\ell \underline s^{(\ell)} D^{(\ell)}$, which we do not
assume.
\end{proposition}

\begin{figure}[t]
\centering
\begin{tikzpicture}
\begin{semilogyaxis}[
  width=11.5cm, height=7cm,
  xlabel={Layer $\ell$},
  ylabel={$s^{(\ell)} / s^{(L)}$},
  xmin=1, xmax=60, ymin=0.9, ymax=500,
  legend pos=north east,
  legend style={font=\scriptsize, fill opacity=0.85, draw=gray!50},
  grid=major, grid style={dashed,gray!30},
  tick label style={font=\footnotesize},
  label style={font=\footnotesize}
]
\addplot[blue, thick, domain=1:60, samples=60] {1.05^(2*(60-x))};
\addlegendentry{$L_g = 1.05$ (no skip)}
\addplot[red, thick, domain=1:60, samples=60] {1.005^(2*(60-x))};
\addlegendentry{$\tilde L_g = 1.005$ (with skip)}
\end{semilogyaxis}
\end{tikzpicture}
\caption{Sensitivity ratio across 60 layers
(Lemma~\ref{lem:prop}). Skip connections with
$L_b L_{\mathrm{LN}} = 0.1$ reduce the ratio from $\approx 320$
to $\approx 1.8$, supporting uniform $d_c$ as approximately optimal
under the optimistic sensitivity of Lemma~\ref{lem:error-prop-app}
(near-orthogonal per-layer errors)---a heuristic numerical
illustration only.}
\label{fig:sensitivity}
\end{figure}
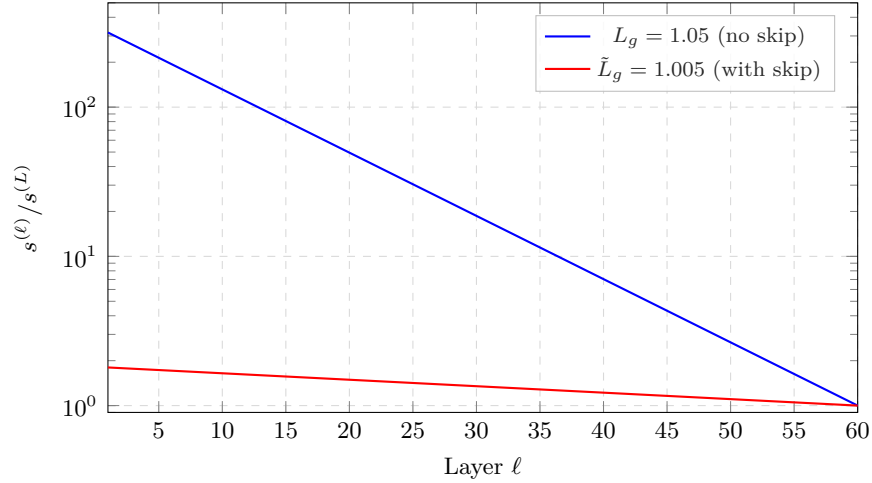

The proofs use the following app-version lemmas.

\begin{lemma}[Error propagation through layers, app]\label{lem:error-prop-app}
For an $L$-layer transformer with per-layer Lipschitz constants
$L_g^{(\ell)}$ (of the layer map $g^{(\ell)}$) and per-layer
attention-reconstruction errors
$\delta_t^{(\ell)} := \|a_t^{(\ell)} - \hat a_t^{(\ell)}\|_2^2$, the
output error satisfies, for any $\{\epsilon_m > 0\}$,
\[
\|a_t^{(L)} - \hat a_t^{(L)}\|_2^2
\;\leq\;
\sum_{\ell=1}^L (1+\epsilon_\ell^{-1})
\Bigl[\textstyle\prod_{m=\ell+1}^L (1+\epsilon_m)(L_g^{(m)})^2\Bigr]
\delta_t^{(\ell)}.
\]
The choice $\epsilon_m = 1$ gives the \emph{rigorous conservative}
sensitivity
$s^{(\ell)}_{\mathrm{cons}} = 2\prod_{m>\ell} 2(L_g^{(m)})^2$. The
\emph{optimistic} sensitivity $s^{(\ell)} = \prod_{m>\ell}(L_g^{(m)})^2$
used in the numerical illustrations below is \emph{not} implied by
small per-layer errors alone: it requires the cross-terms to vanish,
i.e.\ an orthogonality / non-correlation assumption on the error
directions $\{a_t^{(m)}-\hat a_t^{(m)}\}$ across layers. Taking
$\epsilon_m\to0$ is invalid, since the $(1+\epsilon_\ell^{-1})$ factor
on $\delta^{(\ell)}$ then diverges.
\end{lemma}

\begin{proof}
Let $e^{(\ell)} := h_t^{(\ell)} - \hat h_t^{(\ell)}$ be the
hidden-state error after layer $\ell$, with
$\|e^{(\ell)}\|^2$ the propagated error. The layer recursion is
$h^{(\ell+1)} = g^{(\ell+1)}(h^{(\ell)}) + a^{(\ell+1)}$ (layer map
plus the attention contribution), and likewise for the
reconstruction with $\hat a^{(\ell+1)}$. Subtracting,
\[
e^{(\ell+1)}
= \bigl(g^{(\ell+1)}(h^{(\ell)}) - g^{(\ell+1)}(\hat h^{(\ell)})\bigr)
  + (a^{(\ell+1)} - \hat a^{(\ell+1)}).
\]
By Lipschitzness, $\|g^{(\ell+1)}(h^{(\ell)}) -
g^{(\ell+1)}(\hat h^{(\ell)})\| \leq L_g^{(\ell+1)}\|e^{(\ell)}\|$.
Writing $\delta^{(\ell+1)} = \|a^{(\ell+1)} - \hat a^{(\ell+1)}\|^2$
and applying the Young/Cauchy inequality
$\|x + y\|^2 \leq (1+\epsilon_{\ell+1})\|x\|^2 +
(1+\epsilon_{\ell+1}^{-1})\|y\|^2$ to handle the cross-term:
\[
\|e^{(\ell+1)}\|^2
\;\leq\;
(1+\epsilon_{\ell+1})(L_g^{(\ell+1)})^2 \|e^{(\ell)}\|^2
+ (1 + \epsilon_{\ell+1}^{-1})\, \delta^{(\ell+1)}.
\]
Unrolling this recursion from $\ell = 0$ (with $e^{(0)} = 0$) to
$\ell = L$:
\[
\|e^{(L)}\|^2
\;\leq\;
\sum_{\ell=1}^L
  \Bigl[\prod_{m=\ell+1}^L (1+\epsilon_m)(L_g^{(m)})^2\Bigr]
  (1+\epsilon_\ell^{-1})\,\delta^{(\ell)}.
\]
Choosing $\epsilon_\ell = 1$ yields the rigorous conservative form
$\|e^{(L)}\|^2 \le \sum_\ell 2\prod_{m>\ell}2(L_g^{(m)})^2\,\delta^{(\ell)}$.
The optimistic coefficient $\prod_{m>\ell}(L_g^{(m)})^2$ follows only
if the cross-terms
$\langle g^{(\ell+1)}(h^{(\ell)})-g^{(\ell+1)}(\hat h^{(\ell)}),\,
a^{(\ell+1)}-\hat a^{(\ell+1)}\rangle$ are negligible
(orthogonality / non-correlation of the per-layer error directions);
it cannot be obtained by letting $\epsilon_m\to0$, since the factor
$(1+\epsilon_\ell^{-1})$ on $\delta^{(\ell)}$ diverges.
\end{proof}

\begin{theorem}[Conservative multi-layer allocation, app]\label{thm:multi-WZ-app}
Identical to Proposition~\ref{thm:multi-WZ} in the main text.
\end{theorem}

\begin{proof}
The proof combines three elements: a per-layer Wyner--Ziv converse,
a bridge from the $\ell_2$ output error to the per-layer KL
distortions, and a Lagrangian characterization of the optimal
allocation.

\textbf{Step 1 (per-layer converse).} Apply
Theorem~\ref{thm:asymptotic} to each layer $\ell$ independently:
any causal scheme reconstructing layer $\ell$'s attention output
within per-layer distortion $D^{(\ell)}$ has rate at least
$R^{\star,(\ell)}(D^{(\ell)})$. Summing over layers and using
subadditivity of the total rate
($R_n = \sum_\ell R_n^{(\ell)}$ for separate per-layer codes),
$\liminf_n \sum_\ell \E[R_n^{(\ell)}] \geq \sum_\ell
R^{\star,(\ell)}(D^{(\ell)})$ for any admissible allocation
$\{D^{(\ell)}\}$.

\textbf{Step 2 ($\ell_2$-to-KL bridge and the admissible set).}
The end-to-end quantity constrained by the application is the KL
distortion of the final next-token distribution, $D$. We relate it
to the per-layer $\ell_2$ errors $\delta^{(\ell)}$. The output
logit vector is $\zeta_t = W a_t^{(L)}$ for the unembedding $W$ with
$\|W\|_2 \leq L_W$, so the logit error satisfies
$\|\zeta_t - \hat\zeta_t\|_2^2 \leq L_W^2 \|a_t^{(L)} - \hat
a_t^{(L)}\|_2^2$. Under Assumption~\ref{ass:bounded-logits} the
softmax map has bounded curvature, giving the local quadratic
relation
$\KLs(\mathrm{softmax}(\zeta_t)\,\|\,\mathrm{softmax}(\hat\zeta_t))
\leq \tfrac12 \|\zeta_t - \hat\zeta_t\|_2^2 + O(\|\cdot\|^3)$
\citep[the softmax Fisher metric is the logit covariance, bounded
under][]{cover2006elements}. Combining with the error-propagation
bound (Lemma~\ref{lem:error-prop-app})
$\|a_t^{(L)} - \hat a_t^{(L)}\|_2^2 \leq \sum_\ell s^{(\ell)}
\delta^{(\ell)}$, the KL distortion obeys
\[
D \;\leq\; \tfrac12 L_W^2 \sum_\ell s^{(\ell)} \delta^{(\ell)}
  + O(\|\cdot\|^3).
\]
Writing $L_g^2 := \tfrac12 L_W^2$ for the aggregate constant and
identifying $D^{(\ell)} \asymp \delta^{(\ell)}$ (the per-layer
distortion budget, up to the same softmax bridge applied
per layer), the \emph{sufficient} set guaranteeing $\E[D_n]\le D$ is
$\mathcal{A}(D) = \{D^{(\ell)} \geq 0 : L_g^2 \sum_\ell s^{(\ell)}
D^{(\ell)} \leq D\}$, exactly as stated. The constant $L_g^2$ thus
collects the unembedding norm and the softmax-curvature factor; it
is not a per-layer Lipschitz constant. Because this set comes from
the error-propagation \emph{upper} bound, minimizing the sum of the
per-layer rate-function surrogates over $\mathcal{A}(D)$ gives a
conservative allocation objective---operationally achievable only
under an additional per-layer achievability assumption, and in any
case not a converse lower bound; see the discussion in
Proposition~\ref{thm:multi-WZ}.

\textbf{Step 3 (Lagrangian optimum).} Minimize
$\sum_\ell R^{\star,(\ell)}(D^{(\ell)})$ over
$\mathcal{A}(D)$. Since each $R^{\star,(\ell)}$ is convex and
non-increasing, the constraint binds and the Lagrangian is
$\sum_\ell R^{\star,(\ell)}(D^{(\ell)}) + \lambda(L_g^2 \sum_\ell
s^{(\ell)} D^{(\ell)} - D)$. The KKT stationarity condition is
$R^{\star,(\ell)\prime}(D^{(\ell),*}) + \lambda L_g^2 s^{(\ell)} =
0$, i.e.\ $|R^{\star,(\ell)\prime}(D^{(\ell),*})| = \lambda
s^{(\ell)}$ (absorbing $L_g^2$ into $\lambda$), which is the stated
reverse-water-filling condition: layers with larger sensitivity
$s^{(\ell)}$ are allocated lower distortion (higher rate).
\end{proof}

\begin{conjecture}[Heuristic multiplicative mixing model, app]\label{thm:multi-mixing-app}
\emph{(Heuristic; not a theorem.)} If each layer $\ell$ acted as an
independent stochastic channel with a Dobrushin-type contraction
coefficient $\rho^{(\ell)}$, the end-to-end coefficient would compose
multiplicatively, $\rho_{\mathrm{eff}} \le \prod_\ell \rho^{(\ell)}
\le \bar\rho^L$ with $\bar\rho = \max_\ell \rho^{(\ell)}$. We state this
only as a heuristic model: the $\rho^{(\ell)}$ of
Definition~\ref{def:mixing} is a next-token context-truncation decay
parameter, \emph{not} a Dobrushin contraction coefficient, and such
parameters do not in general compose under multiplication across
transformer layers. In particular a residual connection supplies an
identity path that preserves past influence, directly opposing
multiplicative contraction (see the slack discussion below); the
multiplicative law can be justified only for a residual-free stack in
which each layer is a Markov kernel with a genuine Dobrushin
coefficient. We therefore do not rely on it for any rigorous claim.
\end{conjecture}

\textbf{Regularization slack.} Within this heuristic model the bound
$\rho_{\mathrm{eff}} \leq \prod \rho^{(\ell)}$ is worst-case: with skip
connections, the actual decay can be slower:
information is preserved through the residual stream, so the effective
mixing rate is bounded by $\min_\ell \rho^{(\ell)}$ as a (very loose)
lower bound. The empirical $\rho_{\mathrm{eff}}$ for trained models is
typically closer to the upper bound $\prod \rho^{(\ell)}$.

\textbf{Connection to rank collapse} \citep{dong2021attention}: in the
$L \to \infty$ limit, $\rho_{\mathrm{eff}}^L \to 0$, hence the hidden
state converges to a small set of modes (rank collapse). In our
framework, this is the statement that $d^\star_t(D)$ at deep layers is
asymptotically $O(1)$.

\subsection{Skip connections and effective Lipschitz}\label{app:skip-proofs}

We refine the error propagation under skip-connection architectures.

\begin{lemma}[Effective Lipschitz, app]\label{lem:eff-Lip-app}
For pre-LN (resp.\ post-LN) blocks (see formal definition below; main
text), the effective layer-wise Lipschitz constant is
\(\tilde L_g^{(\ell)} = \sqrt{1 + (L_b^{(\ell)} L_{\mathrm{LN}})^2}\)
(resp.\ $L_{\mathrm{LN}}(1 + L_g^{(\ell)})$). The refined error
propagation reads
\(
\|h^{(\ell+1)} - \hat h^{(\ell+1)}\|^2
\leq (\tilde L_g^{(\ell+1)})^2 \|h^{(\ell)} - \hat h^{(\ell)}\|^2
+ \delta^{(\ell+1)}.
\)
\end{lemma}

\begin{proof}[Proof sketch]
For pre-LN: triangle inequality with optimal $(1+\epsilon)$ Cauchy
parameter, Lipschitz of $g^{(\ell+1)}$ and LN. For post-LN: LN as
outer Lipschitz factor over $(1+L_g)$ residual sum. Full details in
the effective-Lipschitz lemma below.
\end{proof}

\begin{corollary}[Numerical illustration, heuristic]\label{cor:MLA-revised-app}
For pre-LN with $L_b L_{\mathrm{LN}} \approx 0.1$,
$\tilde L_g \approx 1.005$, sensitivity ratio across $L = 60$ is
$\sim 1.8$, required $d_c$ variation $\sim 0.06$ dimensions---essentially
zero, consistent with uniform $d_c = 512$ being near-optimal. This uses
the optimistic (near-orthogonal-error) sensitivity of
Lemma~\ref{lem:error-prop-app} and is a heuristic numerical
illustration, not implied by the rigorous conservative bound.
\end{corollary}

\subsection{Joint vs separate compression}\label{app:joint-vs-sep}

\begin{theorem}[Joint vs separate rate, app]\label{thm:joint-sep-app}
$R^\star_{\mathrm{joint}}(D) \leq R^\star_{\mathrm{multi}}(D)$: jointly
compressing the per-layer caches never costs more rate than
compressing them separately at the same distortion.
\end{theorem}

\begin{proof}
Any separate (per-layer) code is a feasible joint code at the same
distortion, so the joint infimum is taken over a superset and cannot
be larger. Equivalently, for a fixed code the conditional
\emph{total correlation}
$\sum_\ell H(\Chat^{(\ell)}\mid\Fhat)-H(\Chat^{(1:L)}\mid\Fhat)\ge0$
shows the separate description length dominates the joint one, and
taking infima preserves the inequality.
\end{proof}

We do \emph{not} claim a strict gap or a pairwise mutual-information
lower bound. The fixed-code total-correlation gap is the chain-rule
quantity $\sum_{\ell=2}^{L} I\bigl(\Chat^{(\ell)};\Chat^{(1:\ell-1)}\mid\Fhat\bigr)$,
which is \emph{not} the pairwise sum
$\sum_{\ell<\ell'} I(\Chat^{(\ell)};\Chat^{(\ell')}\mid\Fhat)$
(the pairwise sum can exceed the total correlation---e.g.\ when all
$\Chat^{(\ell)}$ coincide); moreover a fixed-code entropy gap need not
survive the two rate-function infima, so strict positivity of
$R^\star_{\mathrm{multi}}-R^\star_{\mathrm{joint}}$ is not guaranteed.

When the per-layer caches are nearly conditionally independent
(layer-specialization \citep{geva2020transformer}), the joint and
separate rates are close, so separate compression is near-optimal in
practice. The separate-versus-joint comparison is the classical
concern of distributed (multiterminal) source coding
\citep{elgamal2011network}, whose rate region quantifies the penalty
for encoding correlated sources separately.

\subsection{Application to DeepSeek-V2 (multi-layer)}

For DeepSeek-V2 with $L = 60$ layers we report both regimes of
Section~\ref{sec:main}, as the conclusion depends on whether skip
connections are accounted for.

\emph{Without skip connections.} With per-layer Lipschitz
$L_g^{(\ell)} \approx 1.05$, the sensitivity ratio is
$s^{(1)}/s^{(L)} = \prod_{m=2}^L (L_g^{(m)})^2 \approx 1.05^{118}
\approx 320$. By the reverse-water-filling condition
(Proposition~\ref{thm:multi-WZ}), the optimal $d_c^{(\ell)}$ varies by
$\log(s^{(1)}/s^{(L)})/\log(R_U/r_{\min}) \approx \log(320)/10
\approx 0.6$ dimensions across layers.

\emph{With pre-LayerNorm skip connections.} Using the effective
Lipschitz $\tilde L_g \approx 1.005$ (from
$L_b L_{\mathrm{LN}} \approx 0.1$, Lemma~\ref{lem:eff-Lip-app}),
the ratio drops to $\tilde L_g^{118} \approx 1.8$ and the
optimal-$d_c$ variation to $\log(1.8)/10 \approx 0.06$ dimensions.

In both regimes the across-layer variation ($0.6$ and $0.06$
dimensions respectively) is negligible relative to the deployed
$d_c = 512$, so DeepSeek-V2's uniform latent dimension is
consistent with near-optimal allocation under these
Lipschitz-constant assumptions. The skip-connection regime, which
is the architecturally accurate one for DeepSeek-V2, makes the
margin larger.

\input{appendix_operational}

\end{document}

%% file: appendix_operational.tex
\section{Operational achievability via sequential Wyner--Ziv transfer}
\label{app:operational}

The achievability analysis of Appendix~\ref{app:achievability} bounds
an \emph{information} rate (a sum of single-letter conditional mutual
informations) rather than the physical length of a transmitted
message, and its block-typicality and covering steps implicitly treat
the auxiliary process as if it inherited the next-token mixing of the
source. This appendix gives a self-contained \emph{operational}
account that avoids both issues. We (i) use the physical message-size
rate, (ii) replace the joint-typicality covering step by a one-shot
likelihood-encoder bound valid for an \emph{arbitrary} block joint
law, and (iii) transfer a suffix-only guarantee to the full-context
optimum through a coordinate-hybrid continuity argument whose constant
depends on the per-coordinate side-information alphabet only. The
resulting statement is a \emph{delayed}, within-block noncausal
Wyner--Ziv achievability theorem; it is not an online causal
KV-cache theorem (the encoder observes a whole block before emitting
its message), a point we make explicit in
Remark~\ref{rmk:op-online}.

\subsection{Delayed-block operational setup}
\label{app:op-setup}

Fix a block of $B$ consecutive steps
$\mathcal T_b=\{\tau+1,\dots,\tau+B\}$, $\tau=(b-1)B$. The encoder of
a \emph{suffix-only} scheme observes only the union of the length-$w$
suffixes used at those steps,
\[
\mathsf S_{b,w}\;:=\;X_{\tau-w+1:\tau+B-1}\qquad(\text{length }B+w-1),
\]
and in particular not the deep past $X_{1:\tau-w}$. At step $t$ the
decoder is given the quantized next-step query
$Y_t=\Pi_M(Q_t)\in[M]$, produced by a fixed stochastic kernel
$\kappa_M(\cdot\mid q)$ on $[M]$; the quantizer dither is drawn
independently across $t$, so the side information is conditionally
product given the source,
\begin{equation}
\label{eq:op-cond-product}
\Prob(y^B\mid \mathsf s,a)=\prod_{t\in\mathcal T_b}\kappa_M\!\bigl(y_t\mid q_t(\mathsf s)\bigr),
\qquad Y^B:=(Y_t)_{t\in\mathcal T_b}.
\end{equation}
A public seed $A=a$ (source-independent) is shared by encoder and
decoder.

\paragraph{Code, rate, distortion.}
A (delayed) block Wyner--Ziv code is a triple $(f,\gamma,\{g_t\})$ with
a randomized encoder $\mathsf M=f(\mathsf S_{b,w},\omega)\in\mathcal M_B$,
private randomness $\omega\perp(Z^B,Y^B)\mid(\mathsf S_{b,w},A)$, a
\emph{de-binning} map $\hat U=\gamma(\mathsf M,Y^B,A)$ that may use the
whole block side information, and \emph{local} per-coordinate
reconstructions $\phat_t=g_t(\hat U,Y_t,A)\in\Delta_{\epsilon_0}(\V)$,
$t\in\mathcal T_b$, each using only the step-$t$ query $Y_t$ (the
physically correct cache-decoding order). Here
$\Delta_{\epsilon_0}(\V)$ is the simplex with floor $\epsilon_0\le 1/V$.
All rate functions and the operational optimum below are taken over
this single decoder class; the locality is in the
\emph{reconstruction}, not the de-binning, and is exactly what
localizes the side-information penalty of
Theorem~\ref{thm:op-transfer} to $O(\bar\tau^Y)$ rather than
$O(B\,\bar\tau^Y)$. We use the physical rate
\begin{equation}
R_{\mathrm{op}}:=\tfrac1B\log_2|\mathcal M_B|
\qquad(\text{or }\tfrac1B\,\E[\ell(\mathsf M)]\text{ for a prefix code}),
\end{equation}
and the \emph{integrated log-loss} distortion: with world
$\bullet\in\{F,W\}$ targets $p_t^\bullet$ (full $p_t^F=p_\theta(\cdot\mid X_{1:t-1})$,
windowed $p_t^W=p_\theta(\cdot\mid X_{t-w:t-1})$),
\begin{equation}
\label{eq:op-distortion}
\bar d_B^{\bullet}(\mathsf s,\widehat{\mathbf p})
:=\frac1B\sum_{t\in\mathcal T_b}\sum_{z\in\V}p_t^\bullet(z\mid s_t)\,[-\log\phat_t(z)]
=\frac1B\sum_{t\in\mathcal T_b}\Bigl[H(p_t^\bullet)+\KL{p_t^\bullet}{\phat_t}\Bigr]\in[0,d_{\max}],
\end{equation}
$d_{\max}=\log(1/\epsilon_0)$. Writing the loss this way---integrating
the ghost target $Z_t\sim p_t^\bullet$ out of the distortion rather
than feeding $Z^B$ to the coder as a source---keeps $\mathsf S_{b,w}$
the only encoder-observed variable, as a Wyner--Ziv source must be.
The block (multi-letter) variational rate functions are
\begin{align}
\label{eq:op-rate-fns}
\mathcal R^{F}_{B}(D)&:=\tfrac1B\inf_{\substack{P_{U\mid \mathsf E_B,A},\,\gamma,\,\{g_t\}:\;
U-\mathsf E_B-(Z^B,Y^B)\mid A,\;\E[\bar d_B^F]\le D}}I(\mathsf E_B;U\mid Y^B,A),\\
\mathcal R^{W,\mathrm{suf}}_{B,w}(D)&:=\tfrac1B\inf_{\substack{P_{U\mid \mathsf S_{b,w},A},\,\gamma,\,\{g_t\}:\;
U-\mathsf S_{b,w}-(Z^B,Y^B)\mid A,\;\E[\bar d_B^W]\le D}}I(\mathsf S_{b,w};U\mid Y^B,A),
\end{align}
where $\mathsf E_B:=X_{1:\tau+B-1}$ is the full block history (the
unrestricted encoder observation). The infima range over the
two-stage decoder of \S\ref{app:op-setup}: the reconstruction entering
$\bar d_B^\bullet$ is $\phat_t=g_t\bigl(\gamma(U,Y^B,A),Y_t,A\bigr)$,
where the de-binning $\gamma$ is free to use the whole block side
information $Y^B$ and only the final map $g_t$ is local in $Y_t$. The
auxiliary $U$ obeys the Markov constraint $U-\mathsf E_B-(Z^B,Y^B)\mid A$
(it is generated from the encoder observation alone); the de-binning
$\gamma(U,Y^B,A)$ is a separate decoder-side map and is \emph{not}
required to satisfy that constraint. We retain the multi-letter form
throughout: no claim is made that
$\mathcal R^{\bullet}_B=\tfrac1B\sum_t I(\,\cdot\,;U_t\mid Y_t,A)$.

\subsection{Block converse}
\label{app:op-converse}

\begin{theorem}[Block converse]
\label{thm:op-converse}
Any delayed block code with $\E[\bar d_B^F]\le D$ satisfies
$\tfrac1B\log_2|\mathcal M_B|\ge \mathcal R^F_B(D)$ (and, for a prefix
code, $\tfrac1B\E[\ell(\mathsf M)]\ge\mathcal R^F_B(D)$). No
single-letterization is assumed.
\end{theorem}
\begin{proof}
Take the auxiliary $U:=\mathsf M$. Since $\mathsf
M=f(\mathsf E_B,\omega)$ with $\omega\perp(Z^B,Y^B)\mid(\mathsf E_B,A)$,
the chain $\mathsf M-\mathsf E_B-(Z^B,Y^B)\mid A$ holds, and the code's
own de-binning $\gamma$ together with the local maps $\{g_t\}$
certifies $\E[\bar d_B^F]\le D$, so $U=\mathsf M$---paired with this
two-stage decoder---is feasible for the infimum in
\eqref{eq:op-rate-fns}. Hence
$\log_2|\mathcal M_B|\ge H(\mathsf M)\ge H(\mathsf M\mid Y^B,A)\ge
I(\mathsf E_B;\mathsf M\mid Y^B,A)\ge B\,\mathcal R^F_B(D)$. For a
prefix code $\E[\ell(\mathsf M)]\ge H(\mathsf M)$ by Kraft.
\end{proof}

\subsection{One-shot achievability over a block}
\label{app:op-ach}

We use the one-shot Wyner--Ziv bound obtained from the Poisson
matching lemma; it holds for an arbitrary joint law and an arbitrary
test channel, which is exactly what a dependent block source
requires.

\begin{lemma}[One-shot Wyner--Ziv, {\citealp[Thm.~3]{liananantharam2021poisson}}, relaxed]
\label{lem:op-oneshot}
Let $(\mathsf X,\mathsf T)\sim P_{\mathsf X,\mathsf T}$, test channel
$P_{U\mid\mathsf X}$ and reconstruction $z(\cdot,\cdot)$. For every
integer $L$ there is a code with $|\mathcal M|=L$ and excess-distortion
probability
\[
P_e\;\le\;\E\Bigl[\min\bigl\{\mathbf 1\{d(\mathsf X,z(U,\mathsf T))>\mathsf D\}
+L^{-1}2^{\,\imath(U;\mathsf X)-\imath(U;\mathsf T)},\,1\bigr\}\Bigr],
\]
where $\imath(u;x)=\log_2\frac{dP_{U\mid\mathsf X}(u\mid x)}{dP_U(u)}$,
$\imath(u;t)=\log_2\frac{dP_{U\mid\mathsf T}(u\mid t)}{dP_U(u)}$, and
the side information $\mathsf T$ is available noncausally at the
decoder.
\end{lemma}
The displayed inequality follows from the matching-lemma form
$P_e\le\E[1-\mathbf 1\{d\le\mathsf D\}(1+a)^{-1}]$ of
\citet{liananantharam2021poisson} and $\frac{a}{1+a}\le\min\{a,1\}$ for
$a\ge0$; an equivalent resolvability-based finite-blocklength bound is
in \citet{watanabe2015nonasymptotic}. The likelihood-encoder
interpretation and soft-covering analysis are due to
\citet{songcuffpoor2016likelihood} and \citet{cover2006elements}.

\begin{theorem}[Delayed-block suffix-only achievability]
\label{thm:op-ach}
Apply Lemma~\ref{lem:op-oneshot} with
$\mathsf X=\mathsf S_{b,w}$, $\mathsf T=Y^B$, reconstruction
$\widehat{\mathbf p}=g(U,Y^B,A)$ and the block distortion
\eqref{eq:op-distortion}. Let $\jmath_B:=\imath(U;\mathsf S_{b,w}\mid A)-\imath(U;Y^B\mid A)$,
so $\E[\jmath_B]=I(U;\mathsf S_{b,w}\mid Y^B,A)$ under
$U-\mathsf S_{b,w}-Y^B\mid A$. For any threshold margin $\delta_B>0$
and $\epsilon'>0$ there is a suffix-only block code with
\begin{equation}
\label{eq:op-ach-rate}
\frac1B\log_2|\mathcal M_B|\;\le\;\mathcal R^{W,\mathrm{suf}}_{B,w}(D-\gamma)
\;+\;\delta_B\;+\;\frac{\log_2(1/\epsilon')}{B}\;+\;O(B^{-1}),
\end{equation}
\begin{equation}
\label{eq:op-ach-dist}
\E[\widehat d_B^{W}]\;\le\;D\;+\;d_{\max}\bigl(\eta_B+\eta_B'+\epsilon'\bigr),
\quad
\eta_B:=\Prob\{\jmath_B-\E\jmath_B>B\delta_B\},\ \ \eta_B':=\Prob\{\bar d_B^{\mathrm{tc}}-\E\bar d_B^{\mathrm{tc}}>\gamma\},
\end{equation}
where the test channel is chosen with $\E[\bar d_B^{\mathrm{tc}}]\le D-\gamma$.
\end{theorem}
\begin{proof}
For any $\theta$,
$\E[\min\{L^{-1}2^{\jmath_B},1\}]\le\Prob\{\jmath_B>\theta\}+L^{-1}2^{\theta}$:
on $\{\jmath_B\le\theta\}$ the term is at most $L^{-1}2^{\theta}$, off
it bound the minimum by $1$. Take
$\theta=\E[\jmath_B]+B\delta_B$ and $\log_2 L=\theta+\log_2(1/\epsilon')$,
so $L^{-1}2^{\theta}=\epsilon'$ and the covering contribution is at
most $\eta_B+\epsilon'$; \eqref{eq:op-ach-rate} follows after
minimizing $\tfrac1B\E[\jmath_B]=\tfrac1B I(U;\mathsf S_{b,w}\mid Y^B,A)$
over feasible test channels (the $O(B^{-1})$ absorbs the integer
ceiling of $L$). The first term of Lemma~\ref{lem:op-oneshot} is the
test-channel excess probability
$\Prob\{\bar d_B^{\mathrm{tc}}>D\}=\Prob\{\bar d_B^{\mathrm{tc}}-\E>\gamma\}=\eta_B'$.
Hence the code's excess probability is $P_e\le\eta_B'+\eta_B+\epsilon'$,
and since $\widehat d_B^W\in[0,d_{\max}]$,
$\E[\widehat d_B^W]\le D+(d_{\max}-D)P_e\le D+d_{\max}P_e$, giving
\eqref{eq:op-ach-dist}.
\end{proof}

\subsection{Coordinate-hybrid transfer}
\label{app:op-hybrid}

The next lemma is the operational counterpart of the rate-continuity
step. The naive bound, treating $Y^B$ as one symbol of alphabet
$M^B$, gives a continuity constant $\propto\log(M^B)=B\log M$ and is
vacuous per symbol. A coordinate-by-coordinate hybrid removes the
exponential alphabet, leaving the per-coordinate alphabet $M$.

\begin{lemma}[Coordinate-hybrid side-information rate continuity]
\label{lem:op-hybrid}
Under the conditionally product side information
\eqref{eq:op-cond-product}, with a common base law
$P_{\mathsf E_B,U,A}$ and $U-\mathsf E_B-Y^B\mid A$,
\[
\frac1B\bigl|I_F(\mathsf E_B;U\mid Y^B,A)-I_W(\mathsf E_B;U\mid Y^B,A)\bigr|
\;\le\;\frac1B\sum_{r=1}^B\Omega_M(\tau_r^Y)\;\le\;\Omega_M(\bar\tau^Y_{B,w}),
\]
where $\tau_r^Y:=\E_{\mathsf E_B}[\TV(\kappa_M(\cdot\mid q_r^F),\kappa_M(\cdot\mid q_r^W))]$,
$\bar\tau^Y_{B,w}:=\tfrac1B\sum_r\tau_r^Y$, and
$\Omega_M(\tau)=2\bar f_M(\tau)$ with
$\bar f_M(\tau)=f_M(\min\{\tau,1-\tfrac1M\})$,
$f_M(\tau)=\tau\log(M-1)+h_2(\tau)$ the
equivocation modulus of \citet{alhejjismith2020equivocation}.
\end{lemma}
\begin{proof}
By the chain rule under $U-\mathsf E_B-Y^B\mid A$,
$I(\mathsf E_B;U\mid Y^B,A)=I(\mathsf E_B;U\mid A)-I(U;Y^B\mid A)$.
The base law $P_{\mathsf E_B,U,A}$ is common to both worlds, so
$I(\mathsf E_B;U\mid A)$ is identical in $F$ and $W$ and the
difference reduces to the $Y^B$ term. For $r=0,\dots,B$ define the
hybrid channel $P^{(r)}(y^B\mid e,a)=\prod_{t\le r}\kappa_M(y_t\mid q_t^W)\prod_{t>r}\kappa_M(y_t\mid q_t^F)$,
so $P^{(0)}=F$ and $P^{(B)}=W$. Adjacent hybrids differ only in the
$r$-th coordinate channel. Writing
$I(U;Y^B\mid A)=I(U;Y_{-r}\mid A)+I(U;Y_r\mid Y_{-r},A)$, the first
term is unchanged between $P^{(r-1)}$ and $P^{(r)}$ because the joint
of $(U,Y_{-r},A)$ depends only on the unchanged coordinates (the
product structure \eqref{eq:op-cond-product} gives $Y_r\perp Y_{-r}\mid(\mathsf E_B,A)$).
Only $I(U;Y_r\mid Y_{-r},A)=H(Y_r\mid Y_{-r},A)-H(Y_r\mid U,Y_{-r},A)$
changes; here $Y_r\in[M]$. The joints
$(Y_r,Y_{-r},A)$ and $(Y_r,U,Y_{-r},A)$ differ between the two hybrids
in total variation by at most $\tau_r^Y$ (appending the common
variables preserves or reduces total variation). The equivocation
continuity bound of \citet{alhejjismith2020equivocation}, which
depends on the alphabet of $Y_r$ but not on the conditioning
alphabet, gives
$|I_{P^{(r)}}(U;Y^B\mid A)-I_{P^{(r-1)}}(U;Y^B\mid A)|\le2\bar f_M(\tau_r^Y)=\Omega_M(\tau_r^Y)$.
Summing over $r$ and applying Jensen's inequality to the concave
$\Omega_M$ yields the claim.
\end{proof}

\begin{lemma}[Block suffixization]
\label{lem:op-suffix}
In the windowed world, for any full-history kernel
$K(du\mid h,\mathsf S_{b,w},A)$ the suffix average
$\widetilde K(du\mid \mathsf S_{b,w},A):=\int K(du\mid h,\mathsf S_{b,w},A)\,P(dh\mid\mathsf S_{b,w},A)$
leaves the joint $(\mathsf S_{b,w},U,Y^B)$ and the local distortion
\eqref{eq:op-distortion} unchanged and does not increase the rate;
hence $\mathcal R^{W,\mathrm{suf}}_{B,w}=\mathcal R^{W,\mathrm{unr}}_{B,w}$.
\end{lemma}
\begin{proof}
In the windowed world $(Z^W,Y^B)$ are functions of the suffix block,
so $(Z^W,Y^B)\perp H\mid(\mathsf S_{b,w},A)$; integrating $h$ against
$P(dh\mid\mathsf S_{b,w},A)$ preserves the joint of
$(\mathsf S_{b,w},U,Y^B)$ and hence the local log-loss. By the chain
rule $I(H,\mathsf S_{b,w};U\mid Y^B,A)=I(\mathsf S_{b,w};U\mid Y^B,A)+I(H;U\mid \mathsf S_{b,w},Y^B,A)$;
$\widetilde K$ makes $U\perp H\mid(\mathsf S_{b,w},A)$, and with
$(Z^W,Y^B)\perp H\mid(\mathsf S_{b,w},A)$ the second term vanishes,
while the first is unchanged. This is the block analogue of the
fixed-step suffixization of Appendix~\ref{app:rdt-transfer}.
\end{proof}

\subsection{Assembled operational transfer theorem}
\label{app:op-main}

A physical code's message length does not depend on which world is
used to \emph{evaluate} its distortion, because the exogenous text
$X$ is common to both worlds and $\mathsf M=f(\mathsf S_{b,w},\omega)$
is a function of $X$ and $\omega$ alone; this world-invariance of the
rate is what makes the side-information penalty appear only once.

The localization of the \emph{distortion} penalty below (step~(ii))
requires one further property, which we state as an explicit
hypothesis because the two-stage decoder's de-binning
$\hat U=\gamma(\mathsf M,Y^B,A)$ depends on the whole block.

\begin{assumption}[De-binning query-stability]
\label{ass:op-debin}
Couple the windowed- and full-query block side informations
$(Y^{B,W},Y^{B,F})$ on a common space (same $\mathsf M$, $A$, and
source). The block de-binning is \emph{query-stable}: for some
$\beta_{\mathrm{db}}>0$,
\[
\Prob\bigl\{\gamma(\mathsf M,Y^{B,F},A)\neq\gamma(\mathsf M,Y^{B,W},A)\bigr\}
\;\le\;\eta_B^{\mathrm{db}}=O(B^{-\beta_{\mathrm{db}}}).
\]
A merely vanishing $\eta_B^{\mathrm{db}}\downarrow0$ would \emph{not}
suffice for the polynomial convergence rate of
Corollary~\ref{cor:op-exponent} (e.g.\ $\eta_B^{\mathrm{db}}=1/\log B$
satisfies $\downarrow0$ yet dominates every polynomial term); we
therefore assume the quantitative form.
For a likelihood / Poisson-matching de-binner this holds whenever the
selected codeword's matching score is separated from the runner-up by
a margin the side-information channel preserves under the query
swap---a property of the same information-density concentration as
Assumption~\ref{ass:op-spectrum}; we assume it rather than derive it.
\end{assumption}

The achievability of Theorem~\ref{thm:op-ach} is built in the windowed
world---windowed targets $p_t^W$ and windowed side-information query
$q_t^W$---so transferring its guarantee to the real full-context
problem costs two \emph{distortion} shifts and one \emph{rate} shift,
all for the two-stage decoder of \S\ref{app:op-setup} (block
de-binning $\hat U=\gamma(\mathsf M,Y^B,A)$, local reconstruction
$\phat_t=g_t(\hat U,Y_t,A)$):
\begin{enumerate}[label=(\roman*),leftmargin=*,itemsep=1pt,topsep=2pt]
\item swapping the target $p_t^W\to p_t^F$ shifts the log-loss by
$|\bar d_B^{F}-\bar d_B^{W}|\le\frac1B\sum_t d_{\max}\|p_t^F-p_t^W\|_{1}
=2d_{\max}\bar\tau^{\mathrm{pred}}_{B,w}$ (bounded loss times predictive
total variation);
\item swapping the side-information query $q_t^W\to q_t^F$ affects the
distortion through two routes. \emph{(a)~The local map}
$g_t(\hat U,Y_t,A)$ sees the swapped query only through its own
coordinate $Y_t$; with $\hat U$ held fixed this is a per-coordinate
shift, bounded by $d_{\max}\bar\tau^{Y}_{B,w}$ (bounded loss times the
average query total variation, using the conditionally product
structure \eqref{eq:op-cond-product}). \emph{(b)~The de-binned
codeword} $\hat U=\gamma(\mathsf M,Y^B,A)$ depends on the \emph{whole}
block; by the query-stability hypothesis
(Assumption~\ref{ass:op-debin}) the swap leaves $\hat U$ unchanged
outside an event of probability $\eta_B^{\mathrm{db}}$, on which the
bounded loss differs by at most $d_{\max}$. Hence the query-swap
distortion shift is at most
$d_{\max}\bar\tau^{Y}_{B,w}+d_{\max}\eta_B^{\mathrm{db}}$. Without
stability the global de-binning would force the per-symbol-vacuous
$O(d_{\max}B\bar\tau^{Y}_{B,w})$ bound; query-stability is exactly
what restores the localization;
\item the same query swap shifts the conditional mutual information
$I(\,\cdot\,;U\mid Y^B,A)$ by at most $\Omega_M(\bar\tau^Y_{B,w})$ in
\emph{rate}, by the coordinate-hybrid bound of
Lemma~\ref{lem:op-hybrid}.
\end{enumerate}
Chaining (i)--(iii) with the achievability of
Theorem~\ref{thm:op-ach}, the suffixization
Lemma~\ref{lem:op-suffix}, and the converse
Theorem~\ref{thm:op-converse} gives the following.

\begin{theorem}[Operational delayed-block transfer]
\label{thm:op-transfer}
Assume the conditionally product side information
\eqref{eq:op-cond-product}, the two-stage decoder of
\S\ref{app:op-setup} (local reconstruction $\phat_t=g_t(\hat U,Y_t,A)$),
the de-binning query-stability of Assumption~\ref{ass:op-debin}, and
write
$\bar\tau^{\mathrm{pred}}_{B,w}=\tfrac1B\sum_t\E\,\TV(p_t^F,p_t^W)$ and
$\bar\tau^{Y}_{B,w}$ as in Lemma~\ref{lem:op-hybrid}. Then the
physical rate of the best suffix-only delayed block code, evaluated
on the full-context targets, obeys
\begin{equation}
\label{eq:op-transfer}
R^{\mathrm{op,suf}}_{B,w,F}(D)\;\le\;
R^{\mathrm{op,unr}}_{B,F}\!\bigl(D-2d_{\max}\bar\tau^{\mathrm{pred}}_{B,w}-d_{\max}\bar\tau^{Y}_{B,w}-d_{\max}\eta_B^{\mathrm{db}}\bigr)
\;+\;\Omega_M(\bar\tau^{Y}_{B,w})\;+\;\mathrm{Red}_B,
\end{equation}
where $\mathrm{Red}_B=\delta_B+\log_2(1/\epsilon')/B+O(B^{-1})$ is the
achievability redundancy of Theorem~\ref{thm:op-ach}. The
side-information continuity penalty $\Omega_M$ appears \emph{once}:
because the physical rate is world-invariant, the target swap incurs
no rate penalty, only the distortion shift.
\end{theorem}

The redundancy and the residual distortion slack vanish under an
explicit information-spectrum hypothesis, which---unlike the
truncation-sensitivity exponent---is a property of the chosen scheme's
information-density and loss processes.

\begin{assumption}[Information-spectrum concentration]
\label{ass:op-spectrum}
There exist $\delta_B,\gamma_B\downarrow0$ with
$\Prob\{\jmath_B-\E\jmath_B>B\delta_B\}\to0$ and
$\Prob\{\bar d_B^{\mathrm{tc}}-\E\bar d_B^{\mathrm{tc}}>\gamma_B\}\to0$.
A sufficient condition for a product test channel
($\jmath_B=\sum_t\jmath_t$, $\bar d_B^{\mathrm{tc}}=\tfrac1B\sum_t\ell_t$)
is, for some $\beta_\jmath,\beta_d\in(0,2]$,
\[
\operatorname{Var}\Bigl(\sum_{t\in\mathcal T_b}\jmath_t\Bigr)\le C_\jmath B^{2-\beta_\jmath},
\qquad
\operatorname{Var}\Bigl(\sum_{t\in\mathcal T_b}\ell_t\Bigr)\le C_d B^{2-\beta_d}.
\]
The exponents $\beta_\jmath,\beta_d$ are \emph{not} implied by the
truncation-sensitivity exponents $\alpha_X,\alpha_Y$ of
Definition~\ref{def:poly-mixing}: they govern the temporal dependence
of the information-density and log-loss increments of the coding
scheme, a distinct object that we do not estimate in
\S\ref{sec:experiments} and leave to future measurement. Such an
assumption is genuinely needed and is not implied by truncation
sensitivity: $\beta_\jmath$ governs the temporal dependence of the
coding scheme's information-density increments, which $\alpha_X$ and
$\alpha_Y$---exponents of the trained model's context sensitivity, not
of the evaluation process's covariance decay or mixing rate---do not
control, so a classical mixing central limit theorem need not hold for
$\sum_t\jmath_t$.
\end{assumption}

\begin{corollary}[Vanishing redundancy and convergence exponent]
\label{cor:op-exponent}
Under Assumption~\ref{ass:op-spectrum} with a product test channel,
choosing $\delta_B=\psi_B\,B^{-\beta_\jmath/2}$ with any
$\psi_B\to\infty$ (e.g.\ $\psi_B=\sqrt{\log B}$) gives, by Chebyshev's
inequality, $\eta_B\le C_\jmath/\psi_B^2\to0$ (and likewise
$\eta_B'\to0$ with $\gamma_B=\psi_B\,B^{-\beta_d/2}$), while
$\mathrm{Red}_B\asymp\psi_B B^{-\beta_\jmath/2}\to0$, so the right-hand
side of \eqref{eq:op-transfer} converges to the full-context
operational optimum \emph{at continuity points} of the operational
distortion--rate function $R^{\mathrm{op,unr}}_{B,F}(\cdot)$; at a
discontinuity the conclusion is read off the left limit
$R^{\mathrm{op,unr}}_{B,F}(D^-)$, since the distortion argument is
shifted by $\Delta_{B,w}\downarrow0$ and approaches $D$ from below
(convexity gives continuity on the interior, but we note this
explicitly for the finite-block optimum and boundary points). With the conditionally product structure
\eqref{eq:op-cond-product}, boundary effects of order $w/B+B/n$, and
the truncation rates
$\bar\tau^{\mathrm{pred}}_{w}\asymp w^{-\alpha_X}$,
$\bar\tau^{Y}_{w}\asymp w^{-\alpha_Y}$, the per-token excess rate and
distortion of the suffix-only scheme over the full-context optimum
are
\[
\rho(w,B,n)\asymp w^{-\alpha_Y}\log M+B^{-\beta_\jmath/2}\sqrt{\log B}+\tfrac{w}{B}+\tfrac{B}{n},
\qquad
\sigma(w)\asymp d_{\max}\bigl(2w^{-\alpha_X}+w^{-\alpha_Y}\bigr)+d_{\max}B^{-\beta_{\mathrm{db}}}.
\]
Both penalties are dominated by the slower (query) exponent
$\alpha_Y<\alpha_X$ unless the de-binning rate $\beta_{\mathrm{db}}$ of
Assumption~\ref{ass:op-debin} is the binding one; the de-binning term
$d_{\max}B^{-\beta_{\mathrm{db}}}$ enters as one more competing
polynomial in the optimization (it is \emph{not} automatically of
lower order---a merely vanishing $\eta_B^{\mathrm{db}}$ could dominate,
which is why the quantitative rate is assumed). The overall $n$-rate follows from optimizing
$w=n^a$ and $B=n^b$ over the four competing terms of $\rho$ (a
max--min problem over $0<a<b<1$); it is a strictly positive polynomial
rate limited jointly by the query exponent $\alpha_Y$ and the spectrum
exponent $\beta_\jmath$, which we do not reduce to a single
closed-form exponent.
\end{corollary}

\begin{remark}[Delayed, not online]
\label{rmk:op-online}
The encoder observes the whole block $\mathsf S_{b,w}$ before emitting
$\mathsf M$, and the local decoder reconstructs $\phat_t$ after the
block message is available; this is a \emph{delayed}
(within-block noncausal in the side information) Wyner--Ziv code, not
a zero-delay online cache that must produce the step-$t$ state before
predicting $X_t$. The suffix-only restriction (no use of
$X_{1:\tau-w}$) is genuine, but the timing is a relaxation of the
online problem. A nonanticipative/directed-information formulation in
the spirit of \citet{charalambous2013nonanticipative,tatikonda2009capacity}
would be required for a true online statement.
\end{remark}

\begin{remark}[Relation to Appendix~\ref{app:achievability} and scope]
\label{rmk:op-scope}
Equation~\eqref{eq:op-transfer} is the operational counterpart of the
information-rate bound of Theorem~\ref{thm:achievability}; it isolates
the dependence of the source into the explicit hypothesis
\ref{ass:op-spectrum} and replaces the joint-typicality covering and
projection steps by the one-shot bound of Lemma~\ref{lem:op-oneshot}
and the coordinate-hybrid transfer of Lemma~\ref{lem:op-hybrid}. Three
caveats are stated honestly: (i) the timing is delayed
(Remark~\ref{rmk:op-online}); (ii) the explicit convergence exponent
of Corollary~\ref{cor:op-exponent} uses a product test channel, which
need not attain the block-coupled optimum
$\mathcal R^{W,\mathrm{suf}}_{B,w}$ (the gap is the residual
single-letterization); and (iii) the boundary term $w/B$ is removable
only by conditioning the block-initial reconstructions on the
boundary suffix, which we do not carry out; and (iv) the localization
of the \emph{distortion} penalty to $O(\bar\tau^Y)$ rests on the
de-binning query-stability of Assumption~\ref{ass:op-debin}---without
it the global de-binning gives only the per-token-vacuous
$O(B\bar\tau^Y)$. Within these limits the statement is a
\emph{conditional} operational achievability--converse pair at the
block multi-letter level.
\end{remark}